\documentclass{IEEEtran}
\usepackage{amsmath}
\usepackage{amsthm}
\usepackage{amssymb}
\usepackage{epsfig}
\usepackage{subfigure}
\usepackage{cite}
\newtheorem{lem}{Lemma}
\IEEEoverridecommandlockouts
\begin{document}
\title{A Non-stationary Service Curve Model for Estimation of Cellular Sleep Scheduling}
\author{\IEEEauthorblockN{Nico Becker and Markus Fidler}
\thanks{N. Becker and M. Fidler are with the Institute of Communications Technology, Leibniz Universit\"at Hannover, 30167 Hanover, Germany,
(E-mail: nico.becker@ikt.uni-hannover.de and markus.fidler@ikt.uni-hannover.de).}
\thanks{This work was supported by an ERC Starting Grant (UnIQue, StG 306644). Parts of this work have been presented at ITC 27~\cite{becker:tsc}.}}
\maketitle
\begin{abstract}
While steady-state solutions of backlog and delay have been derived for essential wireless systems, the analysis of transient phases still poses significant challenges. Considering the majority of short-lived and interactive flows, transient startup effects, as caused by sleep scheduling in cellular networks, have, however, a substantial impact on the performance. To facilitate reasoning about the transient behavior of systems, this paper contributes a notion of non-stationary service curves. Models of systems with sleep scheduling are derived and transient backlogs and delays are analyzed. Further, measurement methods that estimate the service of an unknown system from observations of selected probe traffic are developed. Fundamental limitations of existing measurement methods are explained by the non-convexity of the transient service and further difficulties are shown to be due to the super-additivity of network service processes. A novel two-phase probing technique is devised that first determines the shape of a minimal probe and subsequently obtains an accurate estimate of the unknown service. In a comprehensive measurement campaign, the method is used to evaluate the service of cellular networks with sleep scheduling (2G, 3G, and 4G), revealing considerable transient backlog and delay overshoots that persist for long relaxation times.
\end{abstract}
%
%
\section{Introduction}
\label{sec:introduction}
The majority of flows in today's computer networks are short-lived~\cite{mellia:shortlivedtcp} and hence dominated by various transient effects, such as TCP slow start, that can have a significant impact on their performance. Another case in point is energy saving in wireless networks to increase the mobiles' battery lifetime by deactivating power-intensive parts, including the display and the radio interface. In cellular radio, mobiles enter different types of sleep states and wake up according to a defined schedule to receive information in the downlink or in case an uplink transmission is requested. The feature is referred to as discontinuous reception (DRX)~\cite{3gpp:MAC:rel8}.

An example of DRX is illustrated in Fig.~\ref{fig:sleepschedule}. Initially, the mobile is in idle state and wake up is triggered by a channel request at $P_0$. Activating the channel takes  $T_0$ units of time before data can be transmitted. After the data transfer, the mobile enters sleep state at $P_1$ and wakes up again according to a defined sleep cycle after $T_1$. The sleep, respectively, wake-up duration $T_i$ can be deterministic or random. Similarly, the service during data transfer can have a constant or a variable service rate, e.g., due to wireless transmission outages. Generally while the mobile is asleep, data are buffered for later transmission, causing a non-negligible transient backlog overshoot and corresponding delays.

The effects of DRX on power consumption and battery lifetime of mobiles have been evaluated, e.g., for 2G and 3G in~\cite{perrucci:impact2g3g}, and for 4G in~\cite{huang:closeexamination}. Apart from energy saving, the re-activation of a dormant connection requires, however, time~\cite{huang:closeexamination, becker:lte}, e.g., for the exchange of signalling messages. This gives rise to an important trade-off between energy saving and delay~\cite{ra:energydelaytradeoff}, that motivated significant work on the optimization of the parameters of the DRX mode for certain traffic types in 3G~\cite{yang:dynamicpowersaving3g, yang:modelingumtspowersaving} and 4G~\cite{zhou:performanceltedrx, bhamber:analyticLTEpowersavingburstytraffic, wu:performancedrx, zhou:ltedrxm2m} networks. An established method is to model the DRX states by a semi-Markov chain to analyze the stationary mean wake-up delay~\cite{yang:modelingumtspowersaving, zhou:performanceltedrx, bhamber:analyticLTEpowersavingburstytraffic, wu:performancedrx,zhou:ltedrxm2m}. Stationary queueing delays are also derived for the DRX mode using an M$\mid$G$\mid$1 model with vacations~\cite{yang:dynamicpowersaving3g}.
\begin{figure}
\centering
\includegraphics[width=0.9\columnwidth]{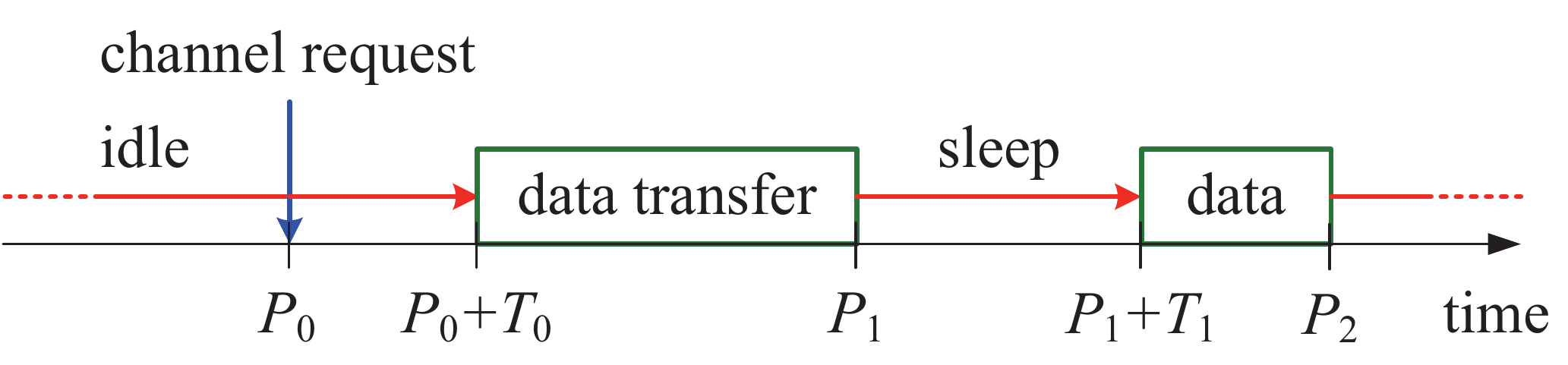}
\caption{Sleep scheduling as implemented in discontinuous reception (DRX).}
\label{fig:sleepschedule}
\end{figure}

The derivation of transient performance measures, however, causes fundamental difficulties. As an example consider the basic M$\mid$M$\mid$1 queue, where the stationary state distribution follows readily from a set of linear balance equations, e.g.,~\cite{ross:probabilitymodels}. The transient behavior, on the other hand, is expressed by a set of differential equations for which mainly approximate or numerical solutions are known~\cite{zhang:transientmm1}. As a consequence, transient solutions of queueing systems are sparse~\cite{wang:transientatm, souza1998algorithm, horvath2012transient} or tailored to specific problems like TCP congestion control~\cite{mellia:shortlivedtcp}.

A framework that does without an assumption of stationarity is the deterministic network calculus~\cite{cruz:networkdelaycalculus, leboudec:networkcalculus, chang:performanceguarantees} that is a theory of time-invariant linear systems under a min-plus algebra~\cite[pp. xiv-xviii]{leboudec:networkcalculus}. Examples of time-invariant linear systems include, e.g., links with a constant capacity or leaky-bucket traffic shapers. Non-linear, time-variant, and possibly non-stationary systems are replaced by time-invariant linear bounds that consider the worst-case behavior. Hence, transient phases are not excluded, however, only maximal backlogs and delays can be expressed, as we will show in Sec.~\ref{sec:systemmodel}. The stochastic network calculus\cite{chang:performanceguarantees, burchard:endtoendstatisticalcalculus, li:effectivebandwidthcalculus2, ciucu:networkservicecurvescaling2, fidler:momentcalculus, jiang:stochasticnetworkcalculus, fidler:netcalcsurvey, ciucu:goodvalue, fidler:netcalcguide}, on the other hand, takes time-variant systems into account, typically by either assuming stationary random processes or using stationary bounds.

In this work, we use the notion of time-variant systems~\cite{chang:dynamicserviceguarantees, agrawal:timevaryingservice} to model non-stationary service characteristics. Time-variant systems are described by bivariate instead of univariate processes that enable considering changes over time. We contribute a non-stationary service curve model and derive solutions for systems with sleep scheduling that provide insights into their transient behavior and lay a basis to reason about the impact of sleep scheduling on the service quality. Measures of interest include the transient overshoot and the relaxation time until the steady-state is approached.

Secondly, we examine methods for estimation of a system's service curve from measurements of probe traffic. We refine known measurement methods to enable the estimation of non-stationary service curves, where we discover specific limitations that are explained by the non-convexity and super-additivity of the service. We devise a novel minimal probing technique that estimates a non-stationary service curve within a defined measure of accuracy. Using this technique, we conduct a comprehensive measurement campaign to evaluate the transient uplink service characteristics of cellular 2G (EDGE), 3G (HSPA), and 4G (LTE) sleep scheduling.

The remainder of this work is structured as follows. In Sec.~\ref{sec:systemmodel}, we define non-stationary service curves, show a method for construction, and derive models of systems with sleep scheduling. In Sec.~\ref{sec:estimation}, we investigate the measurement-based estimation of non-stationary service curves. We reveal difficulties that arise and devise a new minimal probing method. In Sec.~\ref{sec:measurements}, we use our minimal probing method to estimate the service of productive cellular networks with sleep scheduling. We discuss further related works in the respective sections. Sec.~\ref{sec:conclusion} presents brief conclusions.
%
%
\section{Service Models and Performance Bounds}
\label{sec:systemmodel}
In this section, we define regenerative service processes, where regeneration points mark the start of new transient phases, as the basic model of non-stationary systems and show a first application that evaluates the performance of deterministic sleep scheduling (Sec.~\ref{sec:systemmodelregenerative}). Considering systems with random service and random sleep scheduling, we derive non-stationary service curves as a stochastic characterization of the regenerative service processes (Sec.~\ref{sec:systemmodelrandom}). While we restrict the exposition to non-stationary systems, we note that non-stationary traffic can be dealt with in the same way.
%
%
\subsection{Regenerative Service Processes}
\label{sec:systemmodelregenerative}
We consider systems with cumulative arrivals $A(t)$, where $A(t)$ denotes the number of bits that arrive in the time interval $(0,t]$. By convention, there are no arrivals for $t \le 0$ so that we generally consider $t \ge 0$. Clearly, $A(t)$ is a non-negative, non-decreasing function, and $A(0)=0$. Shorthand notation $A(\tau,t) = A(t) - A(\tau)$ is used to denote the arrivals in $(\tau,t]$ where $t \ge \tau \ge 0$. Trivially, $A(t,t)=0$ for all $t \ge 0$. Similarly, $D(t)$ denotes the cumulative departures from the system.

The service that is provided by the system is characterized by a time-variant service process $S(\tau,t)$ that specifies the service as a bivariate function of the interval $(\tau,t]$. It establishes the departure guarantee~\cite{chang:performanceguarantees,chang:timevaryingfiltering,chang:dynamicserviceguarantees,agrawal:timevaryingservice}
\begin{equation}
D(t) \ge \inf_{\tau \in [0,t]} \{ A(\tau) + S(\tau,t) \} =: A \otimes S(t).
\label{eq:serviceprocess}
\end{equation}
By convention, $S(\tau,t)$ is non-negative and $S(t,t) = 0$ for all $t \ge 0$. The operator $\otimes$ that is defined by Eq.~\eqref{eq:serviceprocess} is known as convolution under a min-plus algebra~\cite{chang:performanceguarantees,leboudec:networkcalculus}. We note that $\otimes$ is associative but not commutative in general. For linear systems\footnote{Given pairs of arrivals and corresponding departures $A_i(t), D_i(t)$ with index $i$. A system is min-plus linear if any min-plus linear combination of arrivals $\min \{A_i(t),A_j(t)\} + c$ results in the corresponding linear combination of the departures $\min \{D_i(t),D_j(t)\} + c$ where $c$ is an arbitrary constant.}, such as a work-conserving link with a time-variant capacity, Eq.~\eqref{eq:serviceprocess} holds with equality~\cite[Ex. 5.2.4]{chang:performanceguarantees}. Further examples of Eq.~\eqref{eq:serviceprocess} include scheduling with cross-traffic~\cite{fidler:momentcalculus} and networks of systems where a network service process $S^{net}(\tau,t)$ is computed from the service processes $S^i(\tau,t)$ of the individual systems $i=1 \dots n$ by recursive insertion of Eq.~\eqref{eq:serviceprocess} as~\cite{chang:performanceguarantees}
\begin{equation}
S^{net}(\tau,t) = S^1 \otimes S^2 \otimes \dots \otimes S^n(\tau,t).
\label{eq:networkserviceprocess}
\end{equation}
The notion of a network service process characterizes a network of systems $S^1, S^2, \dots, S^n$ by a single equivalent system $S^{net}$, implying an immediate extension of results for single systems to networks. Given this quality, we will use the terms system and network interchangeably.

The service guarantee of Eq.~\eqref{eq:serviceprocess} enables the derivation of performance bounds. An upper bound of the backlog $B(t) = A(t) - D(t)$ at time $t$ follows by insertion of Eq.~\eqref{eq:serviceprocess} as
\begin{equation}
B(t) \le \sup_{\tau \in[0,t]} \{A(\tau,t) - S(\tau,t)\}.
\label{eq:backlogbound}
\end{equation}
The definition of $B(t)$ considers all data that are in the system including data in buffers as well as data in transmission~\cite{leboudec:networkcalculus}. The first-come first-serve (fcfs) delay at time $t$ defined as $W(t) = \inf\{w \ge 0: A(t) \le D(t+w) \}$ can be derived similarly as
\begin{equation}
W(t) \le \inf \biggl\{w \ge 0: \! \sup_{\tau \in[0,t]} \{A(\tau,t) - S(\tau,t+w) \} \le 0 \biggr\}.
\label{eq:delaybound}
\end{equation}

As opposed to the state-of-the-art, that assumes stationary service processes, we consider the service as a non-stationary, regenerative process with regeneration points $\mathbb{P} = \{P_0,P_1,P_2,\dots\}$ where $P_0=0$ and $P_i < P_{i+1}$ for all $i \ge 0$. We divide $S(\tau,t)$ into segments
\begin{equation}
S_i(\tau,t) = S(\tau+P_i,t+P_i)
\label{eq:regenerativeservicsample}
\end{equation}
for all $0 \le \tau \le t \le P_{i+1} - P_i$ and $i \ge 0$, so that $S_i(\tau,t)$ is the service process between the $i$th and the $(i+1)$th regeneration point. The defining characteristic of a regenerative process~\cite{ross:probabilitymodels} is that the $S_i(\tau,t)$ are statistical replicas, i.e.,
\begin{equation}
\mathsf{P}[S_i(\tau,t) \le x] = \mathsf{P}[S_j(\tau,t) \le x]
\label{eq:statisticalreplica}
\end{equation}
for all $i,j,x \ge 0$, and $0 \le \tau \le t \le \min \{P_{i+1}\!-\!P_i, P_{j+1}\!-\!P_j \}$. Owing to Eq.~\eqref{eq:statisticalreplica}, we omit the index $i$ in the sequel. Also, we will not explicitly mention the constraint $t \le P_{i+1}-P_i$ and assume that the next regeneration point is generally spaced sufficiently apart.

In the following, we present a first application to deterministic sleep scheduling, as introduced in Fig.~\ref{fig:sleepschedule}. We define the respective time-variant service process (Sec.~\ref{sec:deterministicsleeping}) and derive performance bounds thereof (Sec.~\ref{sec:performancebounds}). Further, we include a related time-invariant model of a link with vacations from the literature to show the general limitations of the time-invariant approach (Sec.~\ref{sec:linkwithvacations}). Finally, we include exact results for a special case to verify the tightness of the bounds (Sec.~\ref{sec:exactmarkovmodel}).
%
%
\subsubsection{Deterministic Sleep Scheduling}
\label{sec:deterministicsleeping}
We consider a system that if idle goes to sleep state according to a defined protocol, see Fig.~\ref{fig:sleepschedule}. Wake up is scheduled deterministically $T$ units of time after entering the sleep state. The transmission rate in sleep state is zero and otherwise it is $R$. Taking advantage of the fact that each transition to sleep state is a regeneration point, we define the time-variant service process for $t \ge \tau \ge 0$ as
\begin{equation*}
S^{tlr}(\tau,t) =
\begin{cases}
0, & t \le T \\
R (t-T), & t > T, \tau \le T \\
R (t-\tau), & t > T, \tau > T
\end{cases}
\end{equation*}
that can also be expressed as
\begin{equation}
S^{tlr}(\tau,t) = R[t-\max\{\tau,T\}]_+
\label{eq:transientlatencyrate}
\end{equation}
where $[x]_+ = \max \{0,x\}$ is the non-negative part of $x$. Eq.~\eqref{eq:transientlatencyrate} takes the form of a latency-rate function where the latency is transient, abbreviated by superscript $tlr$.
\subsubsection{Performance Bounds}
\label{sec:performancebounds}
To evaluate the influence of the sleep cycle $T$ on the performance of different types of traffic, we specify the traffic arrivals by a general statistical envelope function $\mathcal{A}^{\varepsilon}(t)$ that satisfies the sample path guarantee
\begin{equation}
\mathsf{P}[A(\tau,t) \le \mathcal{A}^{\varepsilon}(t-\tau),\, \forall \tau \in [0,t] ] \ge 1-\varepsilon
\label{eq:statisticalenvelopedefinition}
\end{equation}
for all $t \ge 0$. Above, $\varepsilon \in (0,1]$ is a probability of overflow. With Eq.~\eqref{eq:statisticalenvelopedefinition}, statistical performance bounds follow readily by substitution of $\mathcal{A}^{\varepsilon}(t-\tau)$ for $A(\tau,t)$, e.g., the backlog bound Eq.~\eqref{eq:backlogbound} yields
\begin{equation}
\mathsf{P} \biggl[B(t) \le \sup_{\tau \in [0,t]} \{\mathcal{A}^{\varepsilon}(t-\tau) - S(\tau,t)\}\biggr] \ge 1-\varepsilon .
\label{eq:statisticalbacklogboundhalf}
\end{equation}
In previous work~\cite{becker:tsc}, we used an established method~\cite{jiang:stochasticnetworkcalculus} for construction of $\mathcal{A}^{\varepsilon}(t)$ that is based on Chernoff's theorem and the union bound. We improve this envelope by applying a recent technique~\cite{jiang:noteonsnetcalc, poloczek:schedulingmartingales} that uses Doob's martingale inequality~\cite{doob:stochasticprocesses} instead. An envelope for processes with independent and identically distributed (iid) increments is
\begin{equation}
\mathcal{A}^{\varepsilon}(t) = \frac{1}{\theta} ( \ln \mathsf{M}_A(\theta,t) - \ln \varepsilon) ,
\label{eq:statisticalenvelope}
\end{equation}
where $\mathsf{M}_A(\theta,t) = \mathsf{E}[e^{\theta A(\tau,\tau+t)}] = (\mathsf{M}_A(\theta,1))^t$ for $\tau,t \ge 0$ is the moment generating function (MGF) of $A$ and $\theta > 0$ is a free parameter. We defer the derivation to the appendix.

For a numerical study, we consider the special case of a discrete time equivalent of a stationary Poisson arrival process that enables an exact solution for comparison. In each time-slot, the probability of a packet arrival is an independent Bernoulli trial with parameter $\alpha \in [0,1]$ such that the number of packet arrivals $N(t)$ in an interval of $t$ time-slots is binomial with the same parameter $\alpha$. The individual packet sizes $Y(i)$ with index $i = 1,2,\dots$ are iid geometric random variables with parameter $\beta \in (0,1]$. Parameter $\alpha$ has the interpretation of an average arrival rate and $1/\beta$ is the average size of packets. For the special case of a system with constant service rate $R=1$, $\alpha/\beta$ is the utilization and $\alpha < \beta$ is required for stability. The respective MGFs are $\mathsf{M}_N(\vartheta,t) = (\alpha e^{\vartheta} + 1-\alpha)^t$ and $\mathsf{M}_Y(\theta) = \beta e^{\theta}/(1-(1-\beta)e^{\theta})$ for $\theta \in [0,-\ln(1-\beta))$~\cite{ross:probability}. The cumulative arrival process is the doubly stochastic process $A(\tau,t) = \sum_{i=N(\tau)+1}^{N(t)} Y(i)$. It has MGF $\mathsf{M}_A(\theta,t-\tau) = \mathsf{M}_N(\ln\mathsf{M}_Y(\theta),t-\tau)$~\cite{ross:probability, fidler:netcalcguide} so that by insertion of $\mathsf{M}_N(\vartheta,t)$ and $\mathsf{M}_Y(\theta)$ we obtain
\begin{equation}
\mathsf{M}_A(\theta,t) = \left( \frac{\alpha \beta e^{\theta}}{1-(1-\beta)e^{\theta}} + 1-\alpha \right)^t.
\label{eq:poissonmgf}
\end{equation}
Clearly, $\mathsf{M}_A(\theta,t) = (\mathsf{M}_A(\theta,1))^t$ as $A(t)$ has iid increments.

In Fig.~\ref{fig:backlog_time_variant_vs_invariant_martingale}, we illustrate the progression of the backlog bound Eq.~\eqref{eq:statisticalbacklogboundhalf} over time, labeled time-variant service. The curve shows an evident transient overshoot. The parameters of the service process given by Eq.~\eqref{eq:transientlatencyrate} are $T=100$ and $R=1$. For the Poisson arrival process we use the MGF Eq.~\eqref{eq:poissonmgf} with parameters $\alpha=0.09$ and $\beta=0.3$ corresponding to a utilization of $0.3$. We choose $\varepsilon=10^{-9}$ and optimize the free parameter $\theta > 0$ in Eq.~\eqref{eq:statisticalenvelope} numerically.
\begin{figure}
\hspace{-10pt}
\subfigure[time-variant vs. time-invariant]{
\includegraphics[width=0.51\columnwidth]{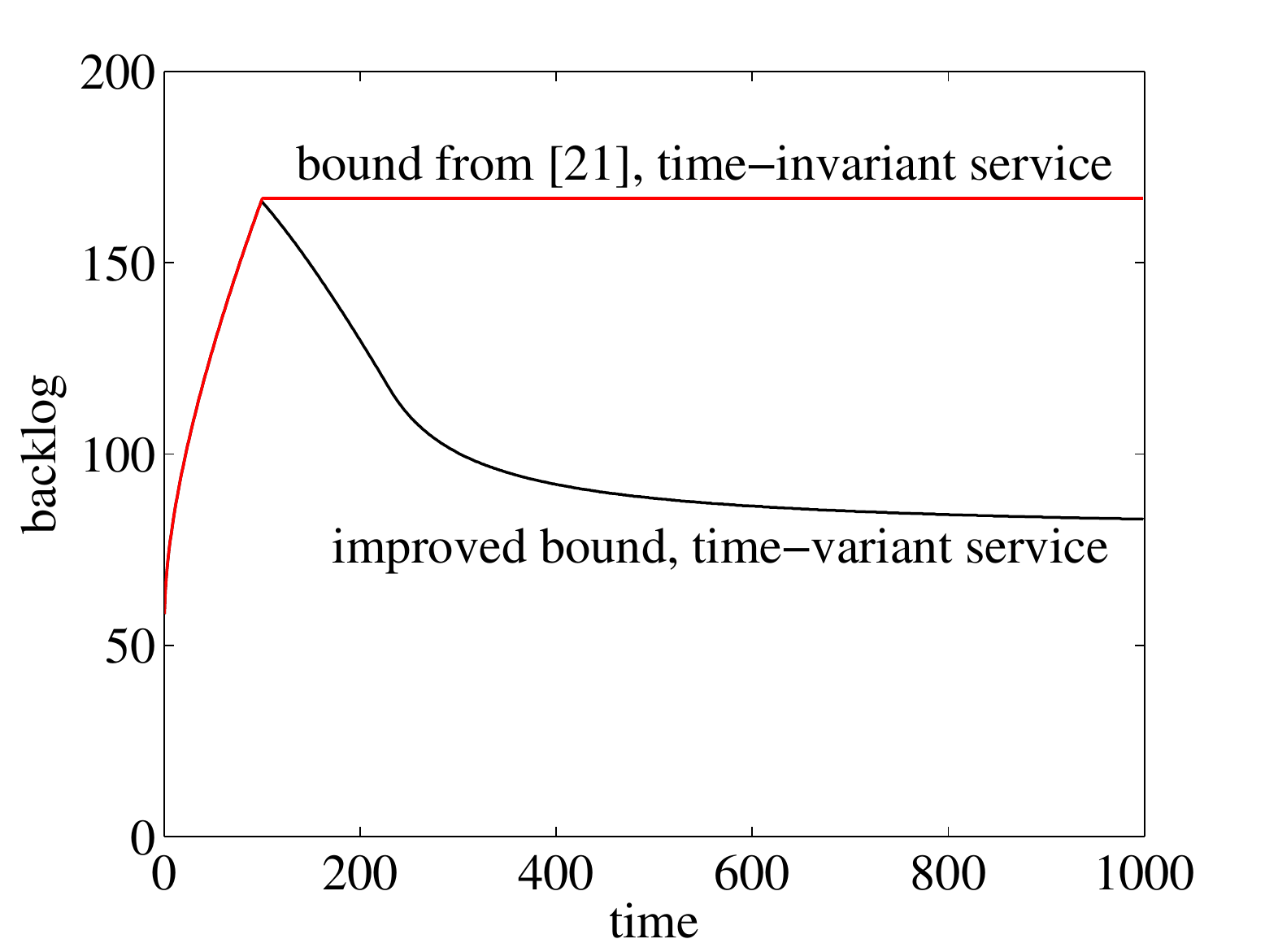}
\label{fig:backlog_time_variant_vs_invariant_martingale}
}
\hspace{-10pt}
\subfigure[time-variant bounds vs. quantile]{
\includegraphics[width=0.51\columnwidth]{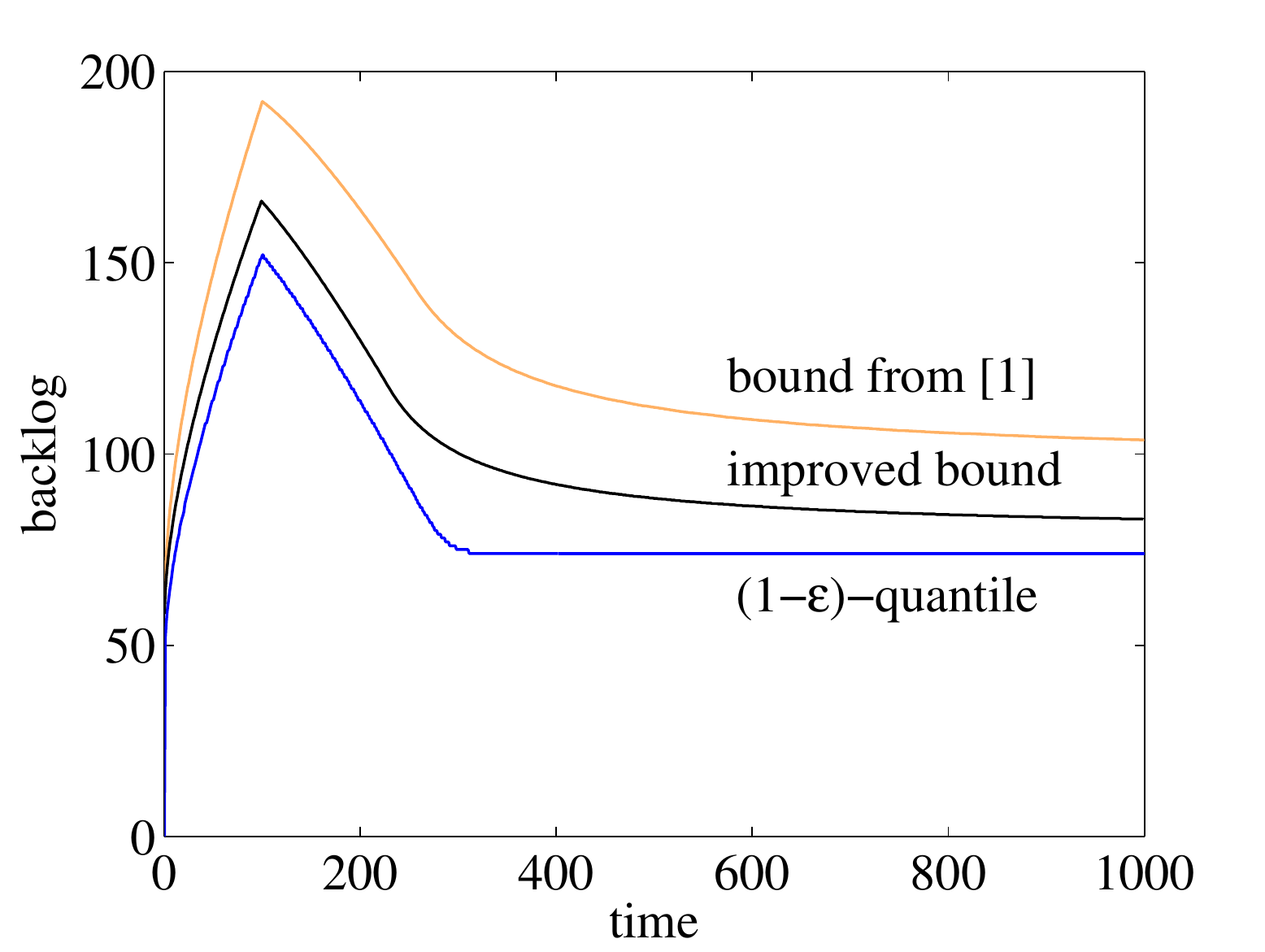}
\label{fig:backlog_martingale_chernoff_exact}
\hspace{-10pt}
}
\caption{Progression of the transient backlog over time. The time-variant service model correctly estimates the shape of the $(1-\varepsilon)$-quantile.}
\label{fig:backlog_qt}
\end{figure}
\subsubsection{Work-conserving Link with Vacations~\cite{chang:performanceguarantees}}
\label{sec:linkwithvacations}
A related model of a work-conserving link that takes deterministically bounded vacations whenever there is no data to transmit is proposed in~\cite[Th. 2.3.16]{chang:performanceguarantees}. The model is formulated in the deterministic network calculus where the service is expressed by the time-invariant, univariate function
\begin{equation}
S^{slr}(\tau,t) = R [t-\tau-T]_+ = S^{slr}(t-\tau)
\label{eq:stationarylatencyrate}
\end{equation}
that depends only on the width of the interval $(\tau,t]$ but not on its absolute position in time. Eq.~\eqref{eq:stationarylatencyrate} has the well-known type of a latency-rate function, however, the latency is stationary as it is not tied to any regeneration points as in Eq.~\eqref{eq:transientlatencyrate}.

Fig.~\ref{fig:backlog_time_variant_vs_invariant_martingale} shows that both, the time-variant and the time-invariant service model, reveal the same growth of the backlog bound until service starts at $T=100$. How the transient backlog is cleared after $T$ and eventually converges to a stationary backlog bound is, however, only explained by the time-variant model. To see why the time-invariant model cannot include this relaxation, note that the $\sup$ in Eq.~\eqref{eq:statisticalbacklogboundhalf} is generally non-decreasing in $t$ if $S$ is a univariate function.
\subsubsection{Exact Markov Model}
\label{sec:exactmarkovmodel}
To evaluate the tightness of the performance bounds, we take advantage of the memorylessness of the arrivals and set up a Markov model that enables an exact (albeit numerical) transient solution. While an exact solution is feasible for the special case of the discrete-time Poisson process, we note that the envelope approach extends to a variety of non-trivial arrival processes including self-similar, long-range dependent~\cite{rizk:fbm,liebeherr:heavytailed}, and heavy-tailed processes~\cite{liebeherr:heavytailed}. Martingale bounds are also available for Markov modulated~\cite{poloczek:schedulingmartingales}, and autoregressive processes~\cite{chang:performanceguarantees, poloczek:schedulingmartingales}.

The state of the Markov chain $K(t) \ge 0$ represents the number of arrivals that are in the system at $t$. For $t \le T$ the transition matrix $\mathbf{Q}(t)$ is composed of the probabilities $q_{i,i} = 1-\alpha$, $q_{i,i+1} = \alpha$, and all other $q_{i,j} = 0$. For $t > T$ the service starts so that $q_{i,i-1} = (1-\alpha)\beta$, $q_{i,i} = (1-\alpha)(1-\beta) + \alpha\beta$, $q_{i,i+1} = \alpha(1-\beta)$, and all other $q_{i,j} = 0$. The Markov chain starts in state $K(0) = 0$, i.e., the initial state distribution $\mathbf{P}(0)$ is the column vector $(1,0,0,\dots)$. The state distribution for $t > 0$ follows by repeated insertion of $\mathbf{P}(t) = \mathbf{Q}(t) \mathbf{P}(t-1)$.

It follows that the state distribution is binomial for $t \le T$, whereas for $t > T$ the distribution makes a transition and for $t \rightarrow \infty$ attains the geometric stationary state distribution~\cite{arita:discretemm1}
\begin{equation*}
\mathsf{P}[K(\infty)=k] = \frac{\beta - \alpha}{\beta (1-\alpha)} \left(\frac{\alpha(1-\beta)}{(1-\alpha)\beta}\right)^k .
\end{equation*}
The backlog distribution can be computed as
\begin{equation}
\mathsf{P}[B(t)=b] = \sum_{k=1}^{\infty} \mathsf{P}[B(t)=b|K(t)=k] \, \mathsf{P}[K(t)=k]
\label{eq:exactbacklog}
\end{equation}
for $b > 0$ and $\mathsf{P}[B(t)=0]=\mathsf{P}[K(t)=0]$ for $b=0$. Using the memorylessness of the geometric distribution, the conditional backlog in state $k$ is the sum of $k$ iid geometric random variables that is negative binomial, i.e., for $k,b > 0$
\begin{equation*}
\mathsf{P}[B(t)=b|K(t)=k] = \binom{b-1}{k-1} \beta^{k} (1-\beta)^{b-k} .
\end{equation*}

For comparison with Fig.~\ref{fig:backlog_time_variant_vs_invariant_martingale}, we depict in Fig.~\ref{fig:backlog_martingale_chernoff_exact} the $(1-\varepsilon)$-quantile of $B(t)$ as defined by Eq.~\eqref{eq:exactbacklog} as well as the previous bound~\cite[Eq.~(7)]{becker:tsc}. We observe that the bounds that are derived from the time-variant service model provide good estimates that recover the shape of the quantile well. The deviations are due to inequalities that are invoked in the derivation of the envelope functions, where the use of Doob's martingale inequality shows a clear improvement.

\begin{figure}
\hspace{-10pt}
\subfigure[$\alpha \in \{0.09,0.12,\dots,0.21\}$]{
\includegraphics[width=0.51\columnwidth]{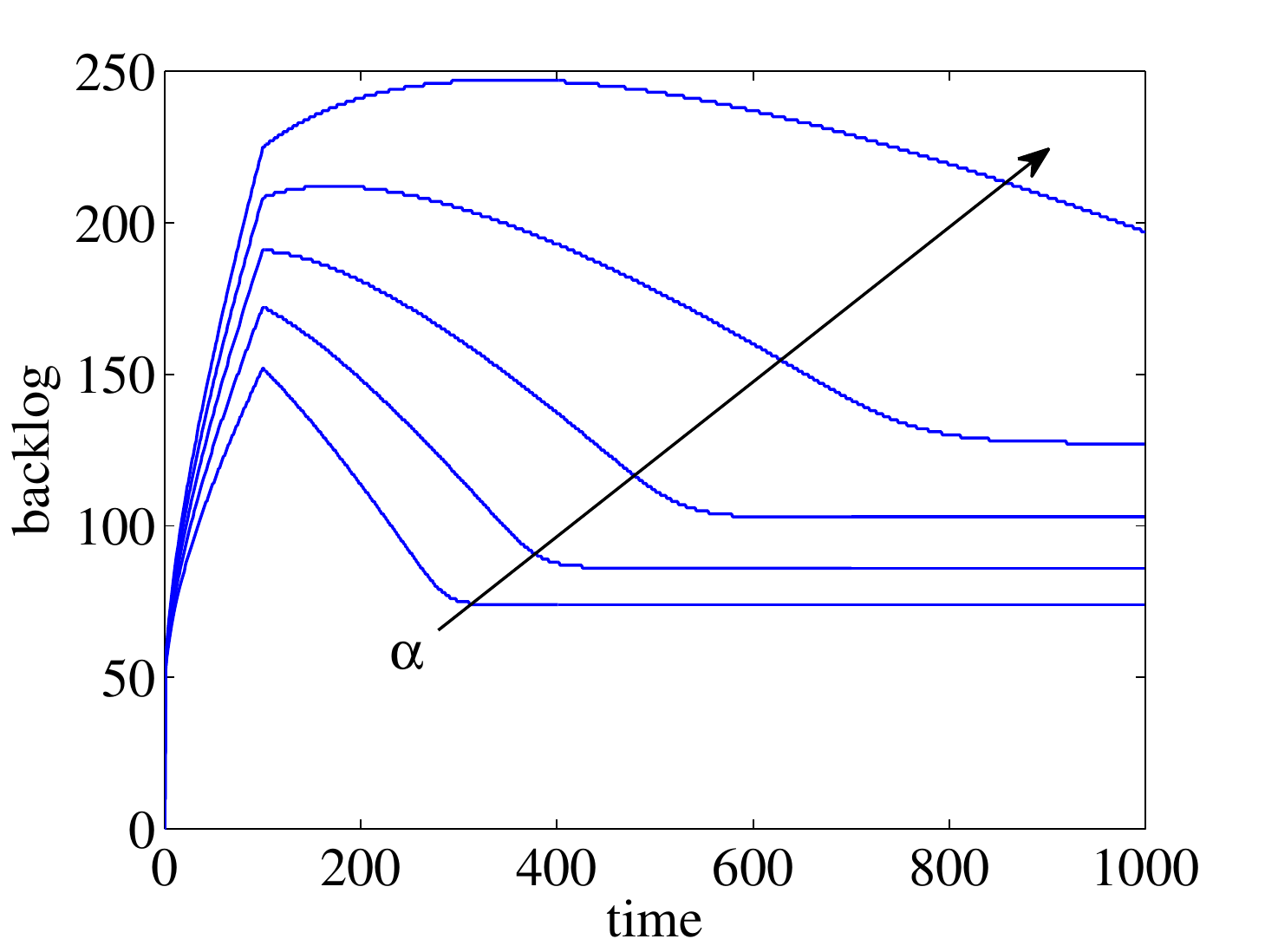}
\label{fig:load_qt}
}
\hspace{-10pt}
\subfigure[$T \in \{0,25,\dots,150\}$]{
\includegraphics[width=0.51\columnwidth]{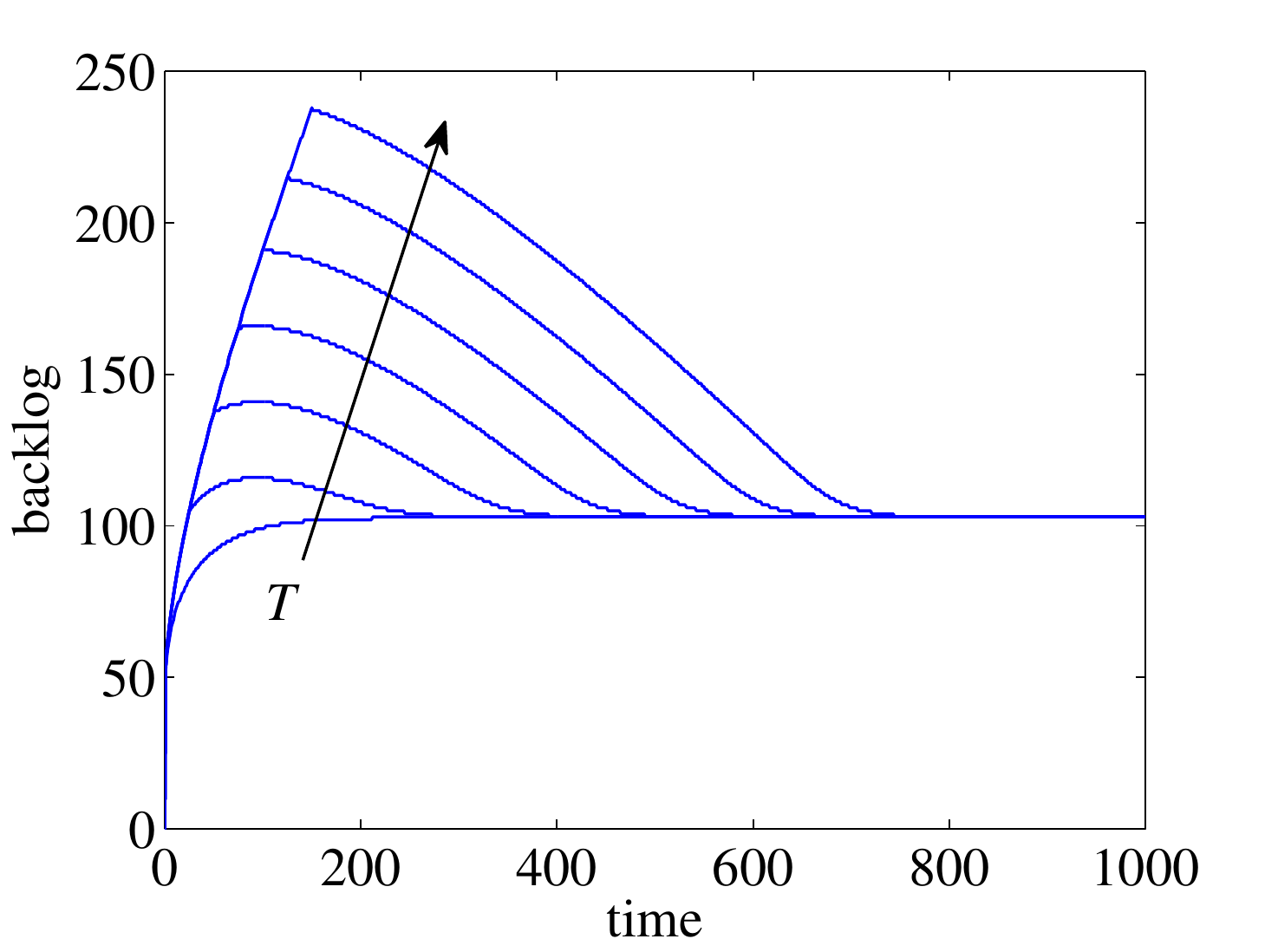}
\label{fig:convergence_qt}
\hspace{-10pt}
}
\caption{Impact of arrival rate $\alpha$ and sleep cycle $T$ on the backlog quantile.}
\label{fig:parameters_qt}
\end{figure}
Fig.~\ref{fig:parameters_qt} evaluates the impact of the arrival rate $\alpha$ and the sleep cycle $T$. The remaining parameters are as in Fig.~\ref{fig:backlog_qt}. The measures of interest~\cite{wang:transientatm} are the maximum overshoot compared to the steady-state and the relaxation time, i.e., the time that is required to reach the steady-state within a defined range. Fig.~\ref{fig:load_qt} shows that $\alpha$ has a significant impact on both quantities. Interestingly, if $\alpha$ is large, the maximum overshoot occurs after $T$, i.e., during the transition from binomial to geometric state distribution. The relaxation time reaches values that are larger than $T$ by an order of magnitude.
%
%
\subsection{Non-stationary Service Curves}
\label{sec:systemmodelrandom}
Based on the concept of time-variant, regenerative service process, we consider the service as a non-stationary random process and derive a basic stochastic characterization. We define a lower service envelope $\mathcal{S}^{\varepsilon}(\tau,t)$ that satisfies
\begin{equation}
\mathsf{P}[S(\tau,t) \ge \mathcal{S}^{\varepsilon}(\tau,t),\, \forall \tau \in [0,t]] \ge 1-\varepsilon ,
\label{eq:bivariateenvelope}
\end{equation}
for all $t \ge 0$ where $\varepsilon \in (0,1]$ is a probability of underflow. Compared to the state-of-the-art definition of univariate service and arrival envelope functions~\cite{burchard:endtoendstatisticalcalculus}, such as Eq.~\eqref{eq:statisticalenvelopedefinition}, the essential difference is that Eq.~\eqref{eq:bivariateenvelope} specifies a time-variant, bivariate envelope function. We emphasize that this generalization is crucial for modelling transient changes over time, as already shown for the deterministic case in Sec.~\ref{sec:deterministicsleeping} and~\ref{sec:linkwithvacations}. It brings further basic differences about, e.g., the convolution is commutative for univariate but not for bivariate functions.

A system with service envelope $\mathcal{S}^{\varepsilon}(\tau,t)$ provides the service guarantee
\begin{equation}
\mathsf{P} [D(t) \ge A \otimes \mathcal{S}^{\varepsilon}(t)] \ge 1-\varepsilon ,
\label{eq:servicecurve}
\end{equation}
for all $t \ge 0$. We refer to $\mathcal{S}^{\varepsilon}(\tau,t)$ as \emph{non-stationary service curve}. To see that Eq.~\eqref{eq:servicecurve} follows from Eq.~\eqref{eq:bivariateenvelope} for a system Eq.~\eqref{eq:serviceprocess}, add $A(\tau)$ to both sides of Eq.~\eqref{eq:bivariateenvelope} to obtain
\begin{equation}
\mathsf{P}[A(\tau) + S(\tau,t) \ge A(\tau) + \mathcal{S}^{\varepsilon}(\tau,t),\, \forall \tau \in [0,t]] \ge 1-\varepsilon .
\label{eq:bivariateenvelope2}
\end{equation}
Since Eq.~\eqref{eq:bivariateenvelope2} makes a sample path argument for all $\tau \in [0,t]$, it also holds that
\begin{equation*}
\mathsf{P}\biggl[\inf_{\tau \in [0,t]} \{A(\tau) + S(\tau,t)\} \ge \inf_{\tau \in [0,t]} \{A(\tau) + \mathcal{S}^{\varepsilon}(\tau,t)\}\biggr] \ge 1-\varepsilon
\end{equation*}
where we finally substitute $D(t) = \inf_{\tau \in [0,t]} \{A(\tau) + S(\tau,t)\}$ using Eq.~\eqref{eq:serviceprocess} to arrive at Eq.~\eqref{eq:servicecurve}.

For a given service process $S(\tau,t)$, a non-stationary service curve $\mathcal{S}^{\varepsilon}(\tau,t)$ can be derived as
\begin{equation}
\mathcal{S}^{\varepsilon}(\tau,t) = -\frac{1}{\theta(\tau,t)} ( \ln \mathsf{M}_S(-\theta,\tau,t) + \rho (t-\tau) - \ln(\rho\varepsilon)\!) ,
\label{eq:statisticalserviceenvelope}
\end{equation}
where $\mathsf{M}_S(-\theta,\tau,t) = \mathsf{E}[e^{-\theta S(\tau,t)}]$ is the negative MGF, respectively, Laplace transform and $\theta(\tau,t) > 0$ and $\rho \in (0,1/\varepsilon]$ are free parameters. The derivation is in the appendix.
%
%
\subsubsection{Random Sleep Scheduling}
\label{sec:randomsleepscheduling}
The concept of non-stationary service curve enables the stochastic analysis of systems. We consider a work-conserving system with random sleep scheduling and random service increments $Z(t)$ for $t \ge 0$. When entering sleep state, the system regenerates and wakes up after a random time $T \ge 0$, i.e., $Z(t)=0$ for $t \in [0,T]$. The service process is computed as $S(\tau,t) = \sum_{\upsilon=\tau+1}^t Z(\upsilon)$ for all $t > \tau \ge 0$ and $S(t,t)=0$ for all $t \ge 0$. To derive the MGF of $S(\tau,t)$, we first consider the number of usable time-slots in $(\tau,t]$, i.e., after time $T$
\begin{equation*}
U(\tau,t) = [t-\max\{\tau,T\}]_+ .
\end{equation*}
The MGF of $U(\tau,t)$ is composed of three terms
\begin{multline}
\mathsf{M}_U(\theta,\tau,t) = \\ e^{\theta (t-\tau)} \mathsf{P}[T \le \tau] + \sum_{\upsilon = \tau+1}^{t} e^{\theta(t-\upsilon)} \mathsf{P}[T = \upsilon] + \mathsf{P}[T > t] ,
\label{eq:slotsafterT}
\end{multline}
that correspond to the cases where the start of the service $T$ occurs before and including $\tau$, within $(\tau,t]$, and after $t$, respectively. Given the service increments $Z(t)$ for $t > T$ are iid with MGF $\mathsf{M}_Z(\theta)$, the MGF of the service process is
\begin{equation}
\mathsf{M}_S(\theta,\tau,t) \!=\! \mathsf{E}\bigl[(\mathsf{M}_Z(\theta))^{U(\tau,t)}\bigr] \!=\! \mathsf{M}_U(\ln \mathsf{M}_Z(\theta),\tau,t) .
\label{eq:servicemgf}
\end{equation}

For an implementation, we use memoryless processes, as solutions for this specific case may also be derived, e.g., from a Markov model. This enables us to compute certain reference results in Sec.~\ref{sec:estimation}. Semi-Markov models are also used in~\cite{yang:modelingumtspowersaving, zhou:performanceltedrx, bhamber:analyticLTEpowersavingburstytraffic, wu:performancedrx,zhou:ltedrxm2m} to analyze mean wake-up delays. We note that the service curve in Eq.~\eqref{eq:statisticalserviceenvelope} is, however, not limited to memoryless processes. In detail, we model $T$ as a geometric random variable with parameter $p$, where $\mathsf{P}[T = \upsilon] = p(1-p)^\upsilon$, and $Z(t)$ for $t > T$ following a basic wireless outage model~\cite{fidler:netcalcguide} as iid Bernoulli trials with parameter $q$.

Regarding Eq.~\eqref{eq:slotsafterT}, we have $\mathsf{P}[T \le \tau] = 1-(1-p)^{\tau+1}$ and
\begin{equation*}
\sum_{\upsilon = \tau+1}^{t} e^{\theta(t-\upsilon)} \mathsf{P}[T = \upsilon] = e^{\theta t} p \sum_{\upsilon = \tau+1}^{t} \bigl(e^{-\theta} (1-p)\bigr)^\upsilon ,
\end{equation*}
where we substitute $y = e^{-\theta} (1-p)$ and compute
\begin{equation*}
\sum_{\upsilon = \tau+1}^{t} y^{\upsilon} = \frac{y^{\tau+1}-y^{t+1}}{1-y}.
\end{equation*}
Having obtained a solution of Eq.~\eqref{eq:slotsafterT}, the MGF of the service process follows from Eq.~\eqref{eq:servicemgf} by insertion of $\mathsf{M}_Z(\theta) = q e^{\theta} + 1-q$ for the Bernoulli service increment process. Finally, the non-stationary service curve is computed from Eq.~\eqref{eq:statisticalserviceenvelope}.

Fig.~\ref{fig:servicecurves} illustrates how $\mathcal{S}^{\varepsilon}(\tau,t)$ identifies the characteristic features of sleep scheduling. The parameters are $p=0.1$ and $q=0.5$, i.e., the mean transient latency is $\mathsf{E}[T] = (1-p)/p = 9$ and the mean stationary service rate is $\mathsf{E}[Z] = q = 0.5$. The service curves are computed for $\varepsilon = 10^{-6}$, $\rho = 10^{-4}$, and $\theta(\tau,t)$ is optimized numerically.

\begin{figure}
\hspace{-10pt}
\subfigure[$\tau \in \{0,100,200,\dots\}$]{
\includegraphics[width=0.51\columnwidth]{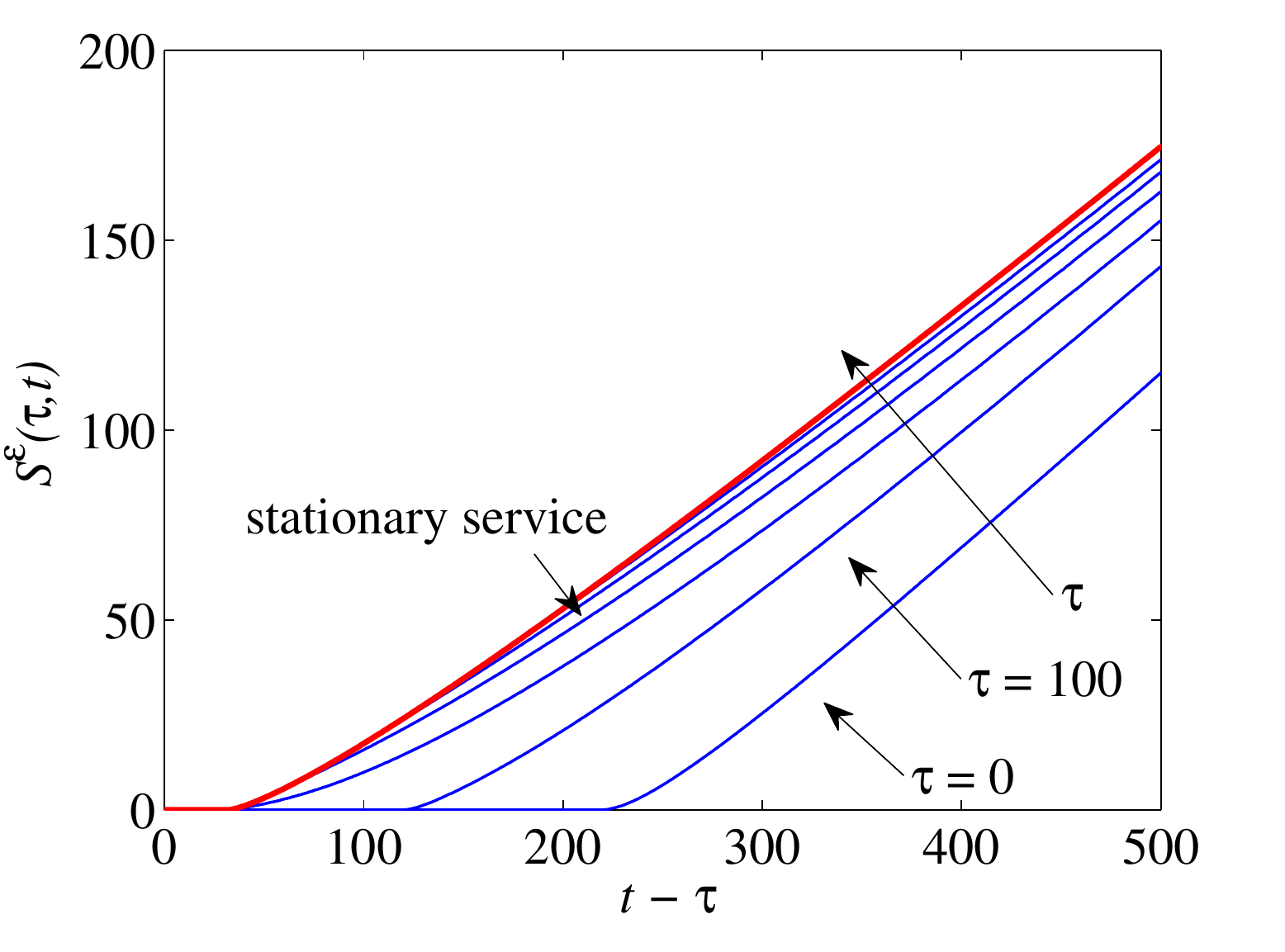}
\label{fig:servicecurvesforward}
}
\hspace{-10pt}
\subfigure[$t \in \{100,200,\dots\}$]{
\includegraphics[width=0.51\columnwidth]{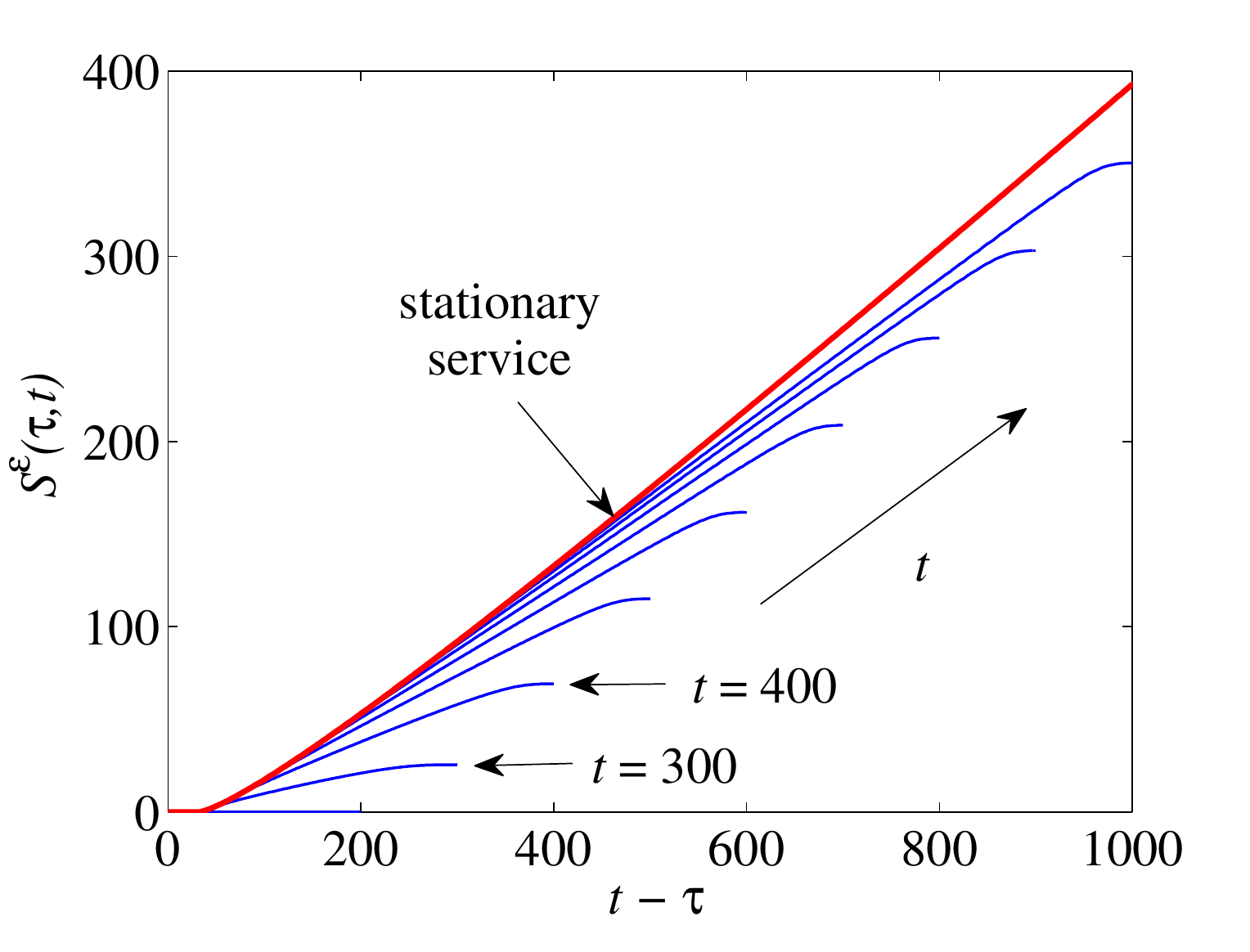}
\label{fig:servicecurvesbackwards}
\hspace{-10pt}
}
\caption{Non-stationary service curves of random sleep scheduling.}
\label{fig:servicecurves}
\end{figure}
In Fig.~\ref{fig:servicecurvesforward}, we show how the service in an interval of width $t-\tau$ increases with increasing distance $\tau$ from the last regeneration point. For small $\tau$ a significant impact of the initial transient phase is noticed. For large $\tau$ we observe that $\mathcal{S}^{\varepsilon}(\tau,t)$ converges towards a stationary service curve that is computed for the same service increment process, however, without sleep scheduling. Fig.~\ref{fig:servicecurvesbackwards} displays the service curves for fixed $t$ and variable $\tau$. We emphasize that the transient phase is reflected in the non-convex shape of the curves, rightwards where $\tau$ approaches zero. In contrast, the initial delay at the origin, that also applies to the stationary service curve, is caused by the outages of the Bernoulli service increment process. For small intervals $t-\tau$ the process results in a service of zero with non-negligible probability. For $t = 100$ and $t=200$ the effects bring about a service curve of zero.
\subsubsection{Performance Bounds}
The presentation of $\mathcal{S}^{\varepsilon}(\tau,t)$ in Fig.~\ref{fig:servicecurvesbackwards} conforms with the formulation of statistical performance bounds, where $t$ is fixed and all $\tau \in [0,t]$ are evaluated. A statistical backlog bound follows with Eqs.~\eqref{eq:servicecurve} and~\eqref{eq:statisticalenvelopedefinition} as
\begin{equation}
\mathsf{P} \Bigl[B(t) \le \sup_{\tau \in [0,t]} \{\mathcal{A}^{\varepsilon}(t-\tau) - \mathcal{S}^{\varepsilon}(\tau,t)\}\Bigr] \ge 1-2\varepsilon
\label{eq:statisticalbacklogbound}
\end{equation}
and a first-come first-served delay bound as
\begin{multline}
\!\!\!\!\!\mathsf{P} \Bigl[W(t) \le \inf \Bigl\{w \! \ge \! 0: \! \sup_{\tau \in [0,t]} \{\mathcal{A}^{\varepsilon}(t-\tau) - \mathcal{S}^{\varepsilon}(\tau,t+w)\} \le 0 \Bigr\}\Bigr] \\ \ge 1-2\varepsilon .
\label{eq:statisticaldelaybound}
\end{multline}
Intuitively, the backlog and delay bound are the maximal vertical and horizontal deviation of $\mathcal{A}^{\varepsilon}(t-\tau)$ and $\mathcal{S}^{\varepsilon}(\tau,t)$, respectively. Backlog and delay bounds for Poisson arrivals as in Sec.~\ref{sec:systemmodelregenerative} with parameters $\alpha = 0.06$, $\beta = 0.3$, and $\varepsilon = 10^{-6}$ are shown in Fig.~\ref{fig:backlogdelay} for different parameters of the service $p$ and $q$. The long-term utilization is $\alpha/(\beta q)$. The free parameters $\theta$ and $\rho$ are optimized numerically. Compared to the stationary case, the transient overshoot of sleep scheduling is considerable.
\begin{figure}
\hspace{-10pt}
\subfigure[backlog]{
\includegraphics[width=0.51\columnwidth]{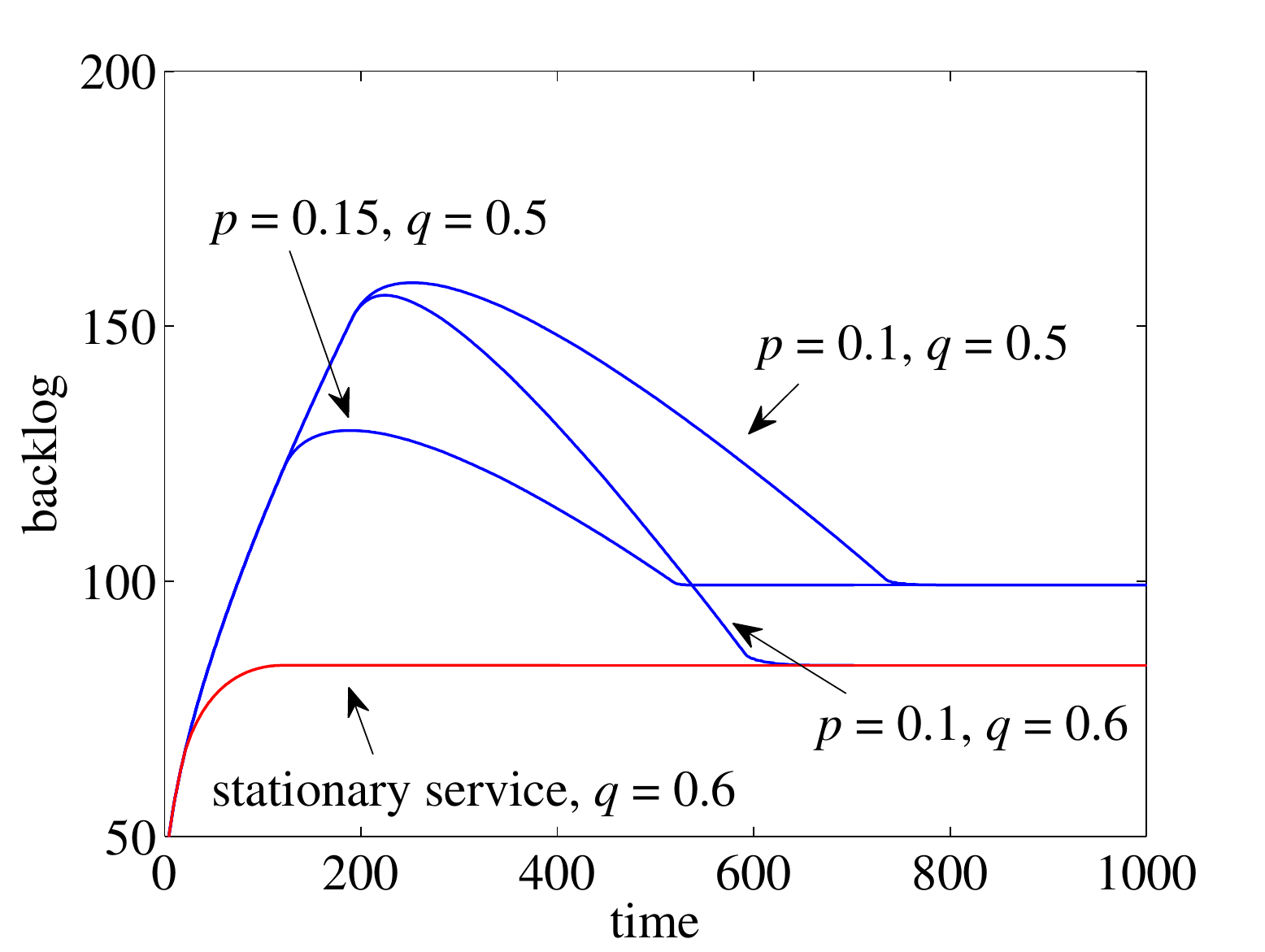}
}
\hspace{-10pt}
\subfigure[delay]{
\includegraphics[width=0.51\columnwidth]{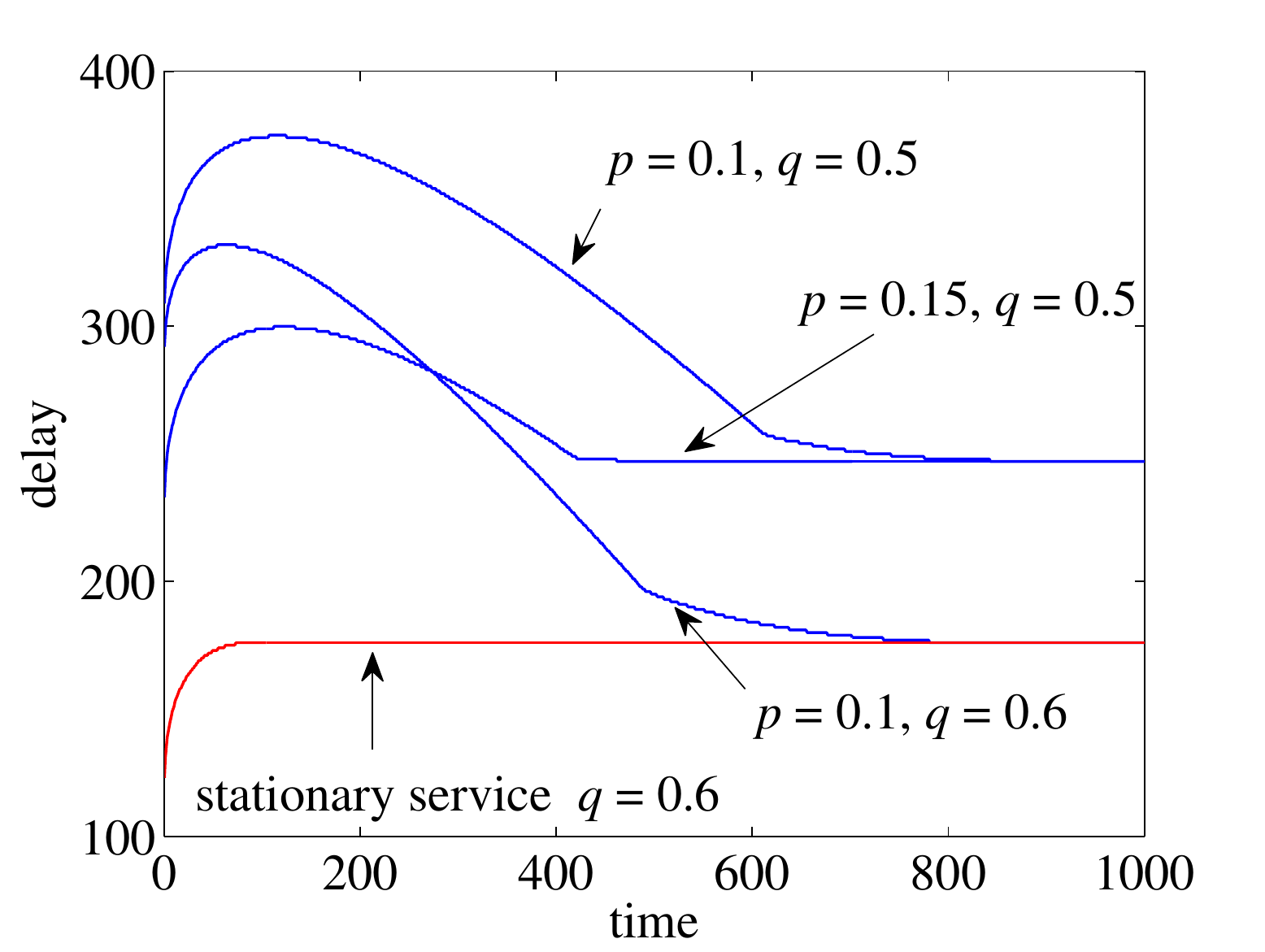}
\hspace{-10pt}
}
\caption{Transient backlog and delay of random sleep scheduling.}
\label{fig:backlogdelay}
\end{figure}
%
%
\section{Measurement-based Estimation Methods}
\label{sec:estimation}
Following the model-based approach in Sec.~\ref{sec:systemmodel}, we derived service curves of sleep scheduling. To identify a system's service curve without knowledge of its internals, we investigate measurement-based methods. We use these to evaluate full-fledged implementations of sleep scheduling in cellular networks, in Sec.~\ref{sec:measurements}.

The estimation of a system's service curve from measurements has received increasing attention in the past years~\cite{cetinkaya:egressadmissioncontrol2, valaee:adhocadmissioncontrol, alcuri:servicecurvemeasurement, bredel:netcalcroutermeasurements, hisakado:legendre, agharebparast:slopedomain, liebeherr:availbw, luebben:availbw2, luebben:availbwdiss}. The general approach can be viewed as system identification of a black box model~\cite{luebben:availbwdiss}, where the service curve of an unknown system is estimated from observations of its arrivals and departures. Usually, the black box model presumes only basic properties such as linearity and stationarity. The existing works can be classified accordingly to be based either on the assumption of a work-conserving system~\cite{cetinkaya:egressadmissioncontrol2, valaee:adhocadmissioncontrol, alcuri:servicecurvemeasurement, bredel:netcalcroutermeasurements} or a general min-plus linear system~\cite{hisakado:legendre, agharebparast:slopedomain, liebeherr:availbw, luebben:availbw2,luebben:availbwdiss} and either use a deterministic model of a time-invariant system~\cite{alcuri:servicecurvemeasurement, bredel:netcalcroutermeasurements, hisakado:legendre, agharebparast:slopedomain, liebeherr:availbw} or a stochastic model of a stationary system~\cite{cetinkaya:egressadmissioncontrol2, valaee:adhocadmissioncontrol, luebben:availbw2, luebben:availbwdiss}. Some of the works rather apply a gray box approach, where additional information about the system internals are used~\cite{alcuri:servicecurvemeasurement, bredel:netcalcroutermeasurements}, e.g., assuming a service curve of the latency-rate type, the latency and rate parameters are estimated from measurements. For a more detailed comparison see also~\cite{fidler:netcalcsurvey}. A closely related research area is available bandwidth estimation, e.g.,~\cite{melander:topp, strauss:spruce, jain:slops, ribeiro:pathchirp, nam:minimalbackloggingbwest} that seeks to estimate the long-term average unused capacity of a system, mostly assuming a deterministic fluid-flow model of a single link with fcfs multiplexing.

The different approaches can be further distinguished by using either passive measurements~\cite{cetinkaya:egressadmissioncontrol2, hisakado:legendre, liebeherr:availbw, alcuri:servicecurvemeasurement}, where traffic traces of the network's production traffic are used, or active probing~\cite{melander:topp, strauss:spruce, jain:slops, ribeiro:pathchirp, nam:minimalbackloggingbwest, valaee:adhocadmissioncontrol, agharebparast:slopedomain, liebeherr:availbw, bredel:netcalcroutermeasurements, luebben:availbw2, luebben:availbwdiss}, that injects artificial test traffic into the network. Active probing provides an important degree of freedom, as the shape of the probe traffic can be tailored to the specific purpose. Typical probes are packet pairs, i.e., two packets sent with a defined spacing~\cite{strauss:spruce}, packet trains, i.e., a sequence of packets sent at a constant rate~\cite{melander:topp, jain:slops, nam:minimalbackloggingbwest, agharebparast:slopedomain, liebeherr:availbw, bredel:netcalcroutermeasurements, luebben:availbw2, luebben:availbwdiss}, or packet chirps that are packet trains with an increasing rate~\cite{ribeiro:pathchirp, liebeherr:availbw}. The variety of probes that are used in practice raises the question of a potentially optimal probe. In this work, we state a condition for a minimal probe and show how it can be estimated.

While methods, such as rate scanning~\cite{luebben:availbw2}, are available for estimation of stationary service curves for the general class of min-plus linear systems with random service~\cite{luebben:availbw2, luebben:availbwdiss}, the measurement-based estimation of non-stationary service curves that are the subject of this work is not elaborated in the current literature. Like~\cite{hisakado:legendre, agharebparast:slopedomain, liebeherr:availbw, luebben:availbw2, luebben:availbwdiss},
we consider min-plus linear systems, where Eq.~\eqref{eq:serviceprocess} holds with equality, i.e.,
\begin{equation}
D(t) = \inf_{\tau \in [0,t]} \{ A(\tau) + S(\tau,t) \}
\label{eq:linearsystem}
\end{equation}
for all $t \ge 0$. The unknown $S(\tau,t)$ is a time-variant, random service process as in~\cite{luebben:availbw2, luebben:availbwdiss}. Compared to~\cite{luebben:availbw2, luebben:availbwdiss} we do, however, not assume stationarity of $S(\tau,t)$.

Based on Eq.~\eqref{eq:linearsystem}, the goal of service curve estimation can be phrased as an inversion problem, i.e., given measurements of $A(t)$ and $D(t)$, solve Eq.~\eqref{eq:linearsystem} for $S(\tau,t)$. Due to the infimum, the min-plus convolution has, however, no inverse operation in general and a solution can only be obtained for certain functions $A(t)$~\cite{liebeherr:availbw}. Finally, we seek to estimate a maximal non-stationary service curve $\mathcal{S}^{\varepsilon}(\tau,t)$, i.e., find a Pareto efficient function $\mathcal{S}^{\varepsilon}(\tau,t)$ that satisfies Eq.~\eqref{eq:servicecurve}.

We perform measurements assuming a regenerative service process as defined in Sec.~\ref{sec:systemmodelregenerative}, where repeated network probes can observe samples of the service process $S_i(\tau,t)$ at regeneration point $P_i$ as illustrated in Fig.~\ref{fig:sleepschedule}.

The remainder of this section is structured as follows: We adapt two known methods, rate scanning (Sec.~\ref{sec:ratescanning}) and burst response (Sec.~\ref{sec:burstresponse}), for estimation of non-stationary service curves. Fundamental limitations of these methods are identified that are explained by the non-convexity and the super-additivity of the service. We devise a new two-phase probing method (Sec.~\ref{sec:minimalprobing}) that first uses the burst response to determine the shape of a suitable probe, that is proven to be minimal under certain conditions. The probe is then used to obtain a service curve estimate with a defined accuracy.
%
%
\subsection{Rate Scanning}
\label{sec:ratescanning}
First, we consider the rate scanning method from~\cite{liebeherr:availbw, luebben:availbw2}. The method uses constant rate probes $A(t) = rt$ for a set of rates $r \in \mathbb{R}$ to excite the network for measurements. While~\cite{liebeherr:availbw} and~\cite{luebben:availbw2} consider deterministic and stationary service curves, respectively, we adapt the method to provide estimates of non-stationary service curves for transient phases.

For constant rate probes $A(t) = r t$ the backlog of a system follows from Eq.~\eqref{eq:linearsystem} as $B(t) = \sup_{\tau \in [0,t]} \{r (t-\tau) - S(\tau,t) \}$. Consequently, it holds that $B(t) \ge r (t-\tau) - S(\tau,t)$ for all $\tau \in [0,t]$. Solving for $S(\tau,t)$ provides the lower bound
\begin{equation}
S(\tau,t) \ge r (t-\tau) - B(t) ,
\label{ratescansamplepath}
\end{equation}
for all $\tau \in [0,t]$

Next, we use the definition of $(1-\xi)$-quantile
\begin{equation}
X^{\xi} = \inf \{x \ge 0: \mathsf{P}[X \le x] \ge 1-\xi \}
\label{eq:quantile}
\end{equation}
and denote $B^{\xi}(r,t)$ the backlog quantile at time $t$ as a function of $r$. By insertion of $B^{\xi}(r,t)$ into Eq.~\eqref{ratescansamplepath}, it holds that
\begin{equation*}
\mathsf{P}[S(\tau,t) \ge r(t-\tau) - B^{\xi}(r,t), \forall \tau \in [0,t]] \ge 1-\xi ,
\end{equation*}
which takes the form of Eq.~\eqref{eq:bivariateenvelope}. By application of the union bound, it follows for $t \ge \tau \ge 0$ that
\begin{equation}
\mathcal{S}_{rs}^{\varepsilon}(\tau,t) = \max_{r \in \mathbb{R}} \{r (t-\tau) - B^{\xi}(r,t)\}
\label{eq:ratescanestimate}
\end{equation}
is a non-stationary service curve as defined by Eq.~\eqref{eq:servicecurve} with probability $\varepsilon = \sum_{r \in \mathbb{R}} \xi$. Compared to~\cite{luebben:availbw2}, Eq.~\eqref{eq:ratescanestimate} uses the transient instead of the stationary backlog to estimate a non-stationary service curve.

In a practical implementation of rate scanning, the set of probing rates $\mathbb{R}$ has to be selected. Options are, e.g., linearly or geometrically spaced rates combined with suitable tests that determine the maximum probing rate~\cite{liebeherr:availbw,luebben:availbw2}. Then, estimates of the backlog quantile $B^{\xi}(r,t)$ for $t \ge 0$ are obtained from repeated measurements for each $r \in \mathbb{R}$. Finally, a service curve estimate $\mathcal{S}_{rs}^{\varepsilon}(\tau,t)$ follows from Eq.~\eqref{eq:ratescanestimate}.

For evaluation of the method, we consider a simulation of the random sleep scheduling model from Sec.~\ref{sec:randomsleepscheduling} with identical parameters $p=0.1$ and $q=0.5$. We perform rate scanning with ten rates $r \in \{0.05,0.1,\dots,0.5\}$. For each rate we obtain $10^5$ backlog samples from repeated experiments. From these, we extract an estimate of the backlog quantile $B^{\xi}(r,t)$ for $\xi = 10^{-4}$ so that $\varepsilon = \sum_{r\in\mathbb{R}} \xi = 10^{-3}$.

\begin{figure}
\hspace{-10pt}
\subfigure[Rate scanning]{
\includegraphics[width=0.51\columnwidth]{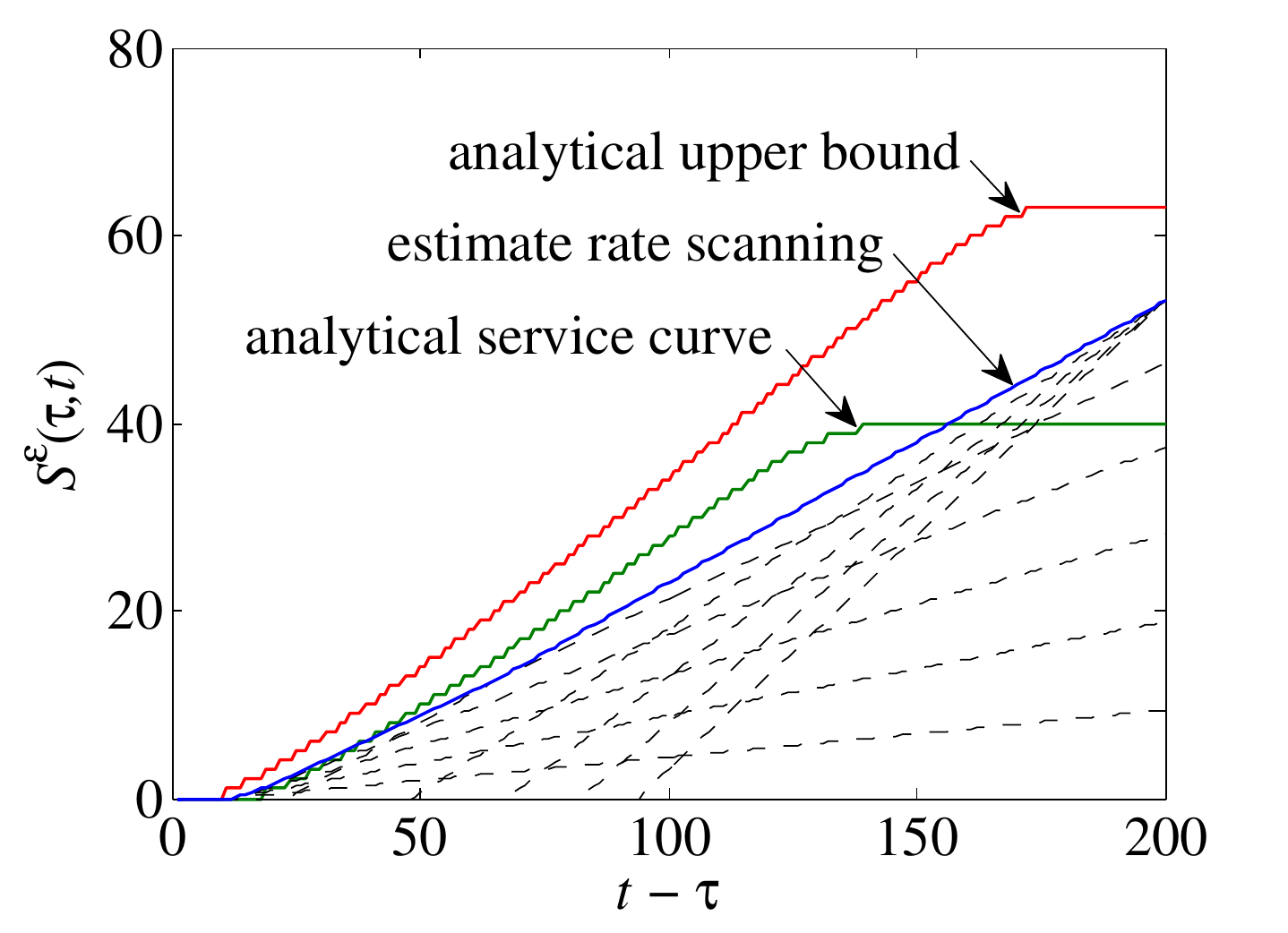}
\label{fig:ratescanning}
}
\hspace{-10pt}
\subfigure[Burst response]{
\includegraphics[width=0.51\columnwidth]{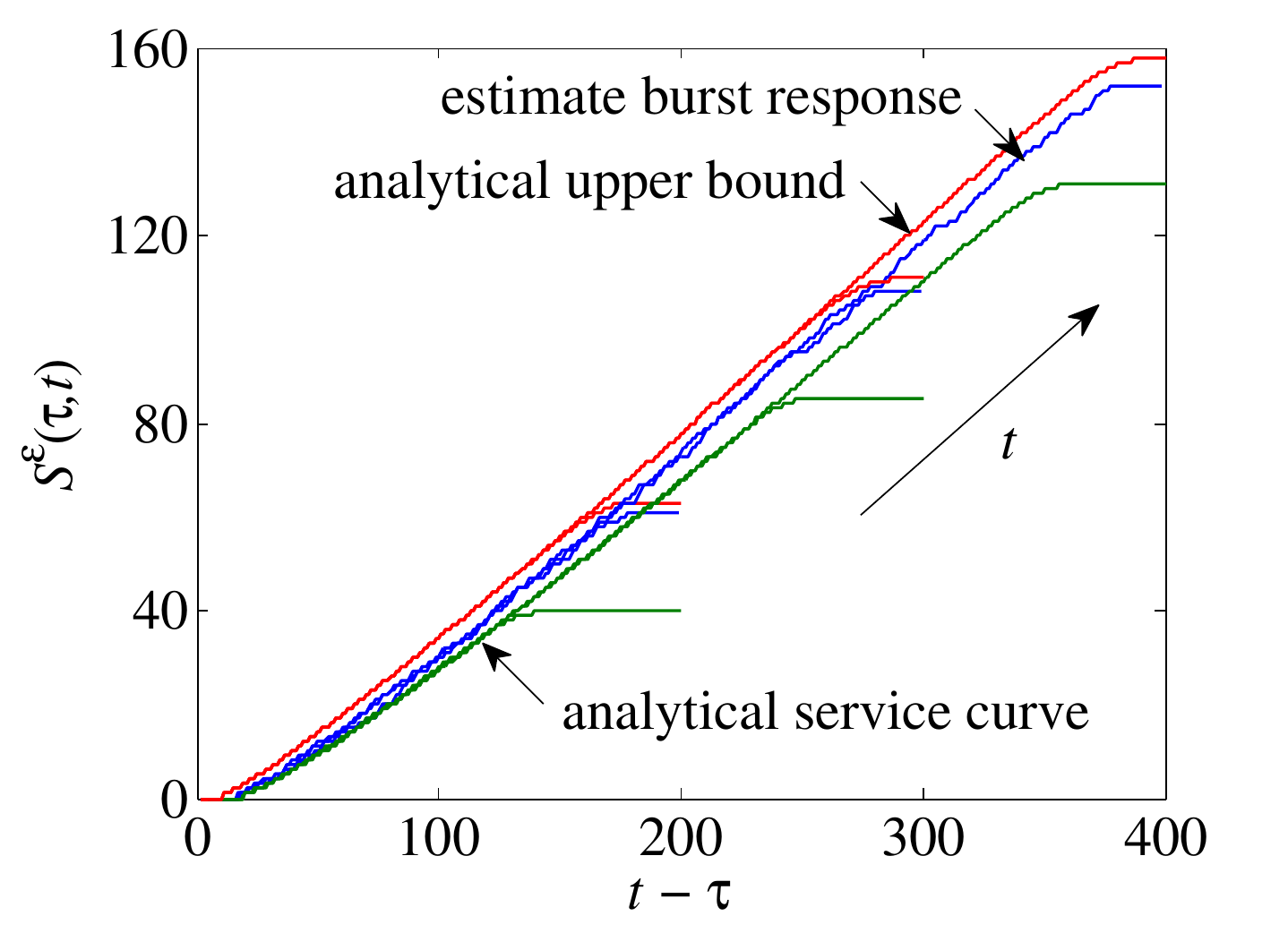}
\label{fig:burstprobing}
\hspace{-10pt}
}
\caption{Service curve estimates of random sleep scheduling. The estimate of rate scanning is the maximum of linear rate segments (dashed lines). By construction it can only recover a convex hull. Burst probing can estimate non-convex service curves and performs close to the analytical upper bound.}
\label{fig:servicecurveestimates}
\end{figure}
Fig.~\ref{fig:ratescanning} shows the linear segments gathered by each of the ten probing rates $r \in \mathbb{R}$ and the resulting estimate $\mathcal{S}_{rs}^{\varepsilon}(\tau,t)$ that is obtained as the maximum of the linear segments by Eq.~\eqref{eq:ratescanestimate}. We display bivariate functions for $t=200$ and vary $\tau \in [0,t]$ to consider the influence of the width of the interval $(\tau,t]$. As a reference, we include an analytical upper bound\footnote{For a Bernoulli service increment process with parameter $q$ that starts after a geometrically distributed time $T$, an upper bound of the service provided in $(\tau,t]$ is derived as the $(1-\varepsilon)$-quantile of a binomial distribution with parameters $t-\max\{\tau,T\}$ and $q$. We note that the first parameter, i.e., the number of trials $t-\max\{\tau,T\}$, is a random variable.}. Any function that exceeds the upper bound for some $\tau \in [0,t]$ violates the definition of service envelope Eq.~\eqref{eq:bivariateenvelope}. Further, we show an analytical service curve\footnote{The analytical service curve again applies the binomial distribution as in case of the upper bound but in addition makes a sample path argument using the union bound to satisfy Eq.~\eqref{eq:bivariateenvelope}.} that provides the lower guarantee specified by Eq.~\eqref{eq:servicecurve}. Both, the analytical upper bound as well as the analytical service curve exhibit a flat area in the upper right for large $t-\tau$, respectively, small $\tau$ that is due to the initial transient phase starting at $\tau=0$. In contrast, the service curve estimate of rate scanning cannot recover the non-convex part of the analytical results.

This shortcoming of rate scanning is fundamental as it is explained by the properties of Eq.~\eqref{eq:ratescanestimate} that has the form of a Legendre-Fenchel transform of the backlog~\cite{liebeherr:availbw}. The Legendre-Fenchel transform has a number of useful properties in the network calculus~\cite{agharebparast:slopedomain, hisakado:legendre, fidler:legendre}. It is also known as convex conjugate as the result is generally a convex function~\cite{rockafellar:convexanalysis}. For convex functions, the Legendre-Fenchel transform is its own inverse, whereas the bi-conjugate of a non-convex function can only recover the convex hull~\cite{rockafellar:convexanalysis}, as seen for the estimate of rate scanning in Fig.~\ref{fig:ratescanning}.
%
%
\subsection{Burst Response}
\label{sec:burstresponse}
\begin{figure*}
\begin{center}
\subfigure[Transient latency]{
\includegraphics[width=0.51\columnwidth]{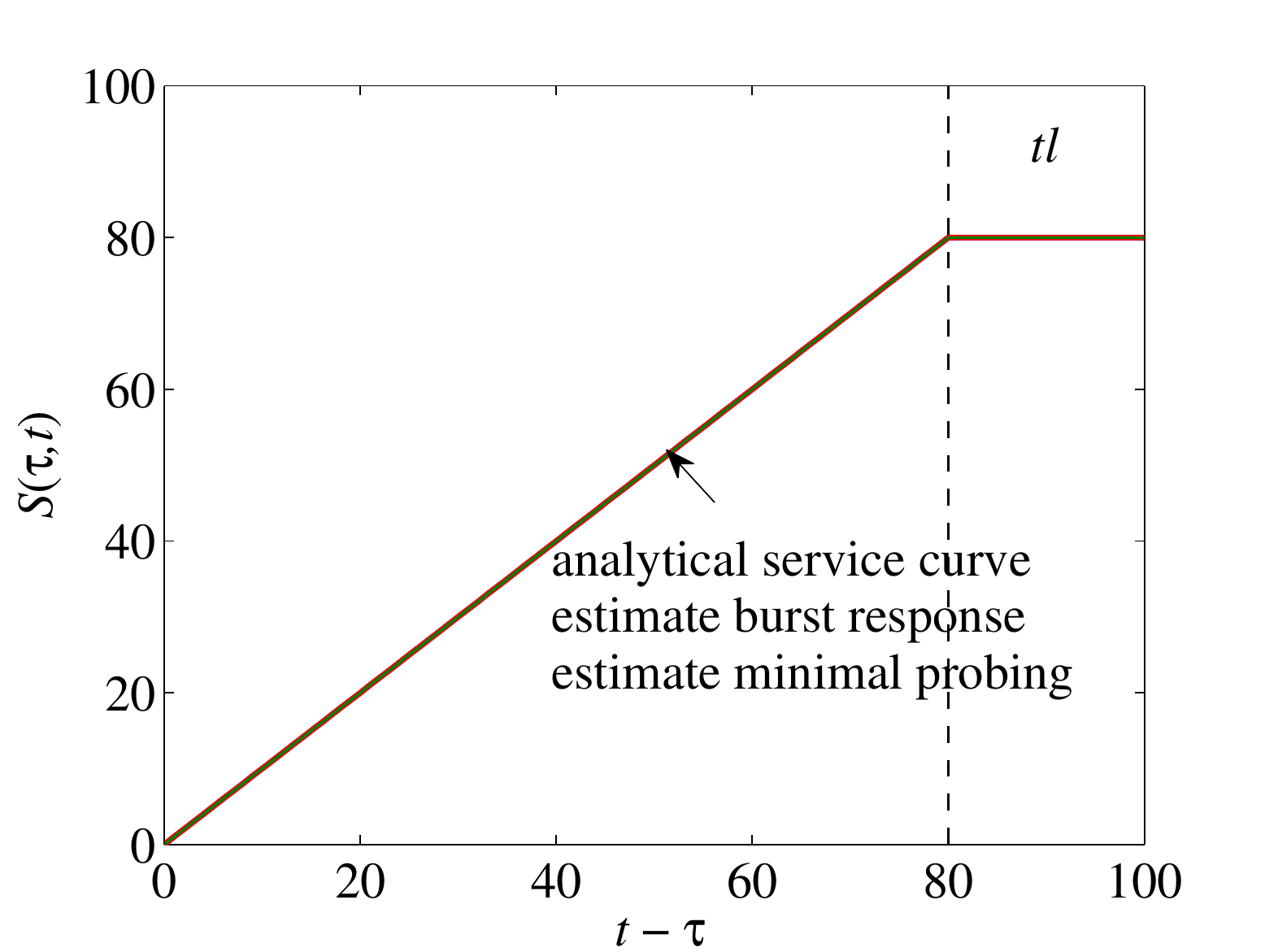}
\label{fig:transientlatencyrate}
}
\subfigure[Stationary latency]{
\includegraphics[width=0.51\columnwidth]{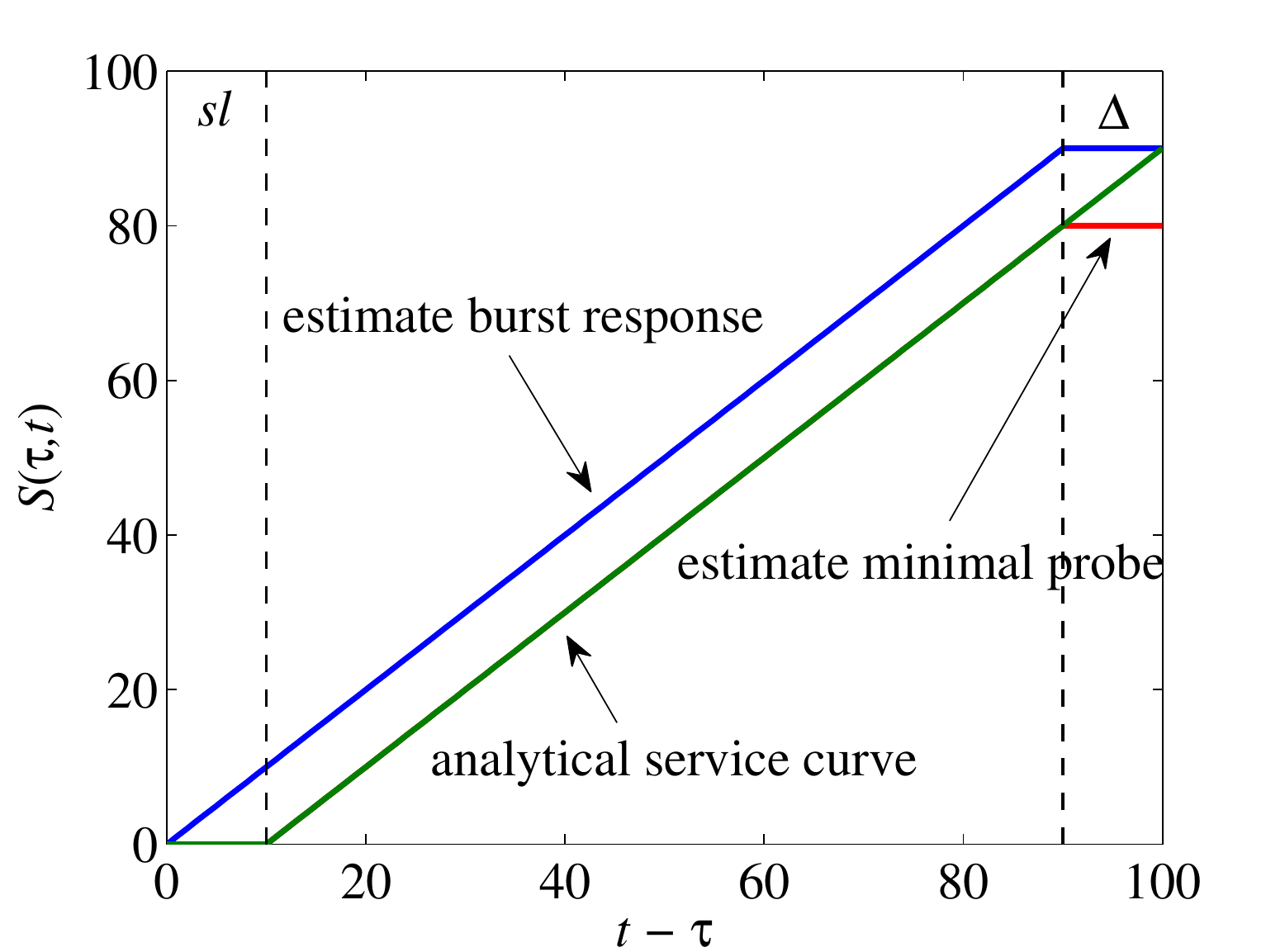}
\label{fig:stationarylatencyrate}
}
\subfigure[Transient and stationary latencies]{
\includegraphics[width=0.51\columnwidth]{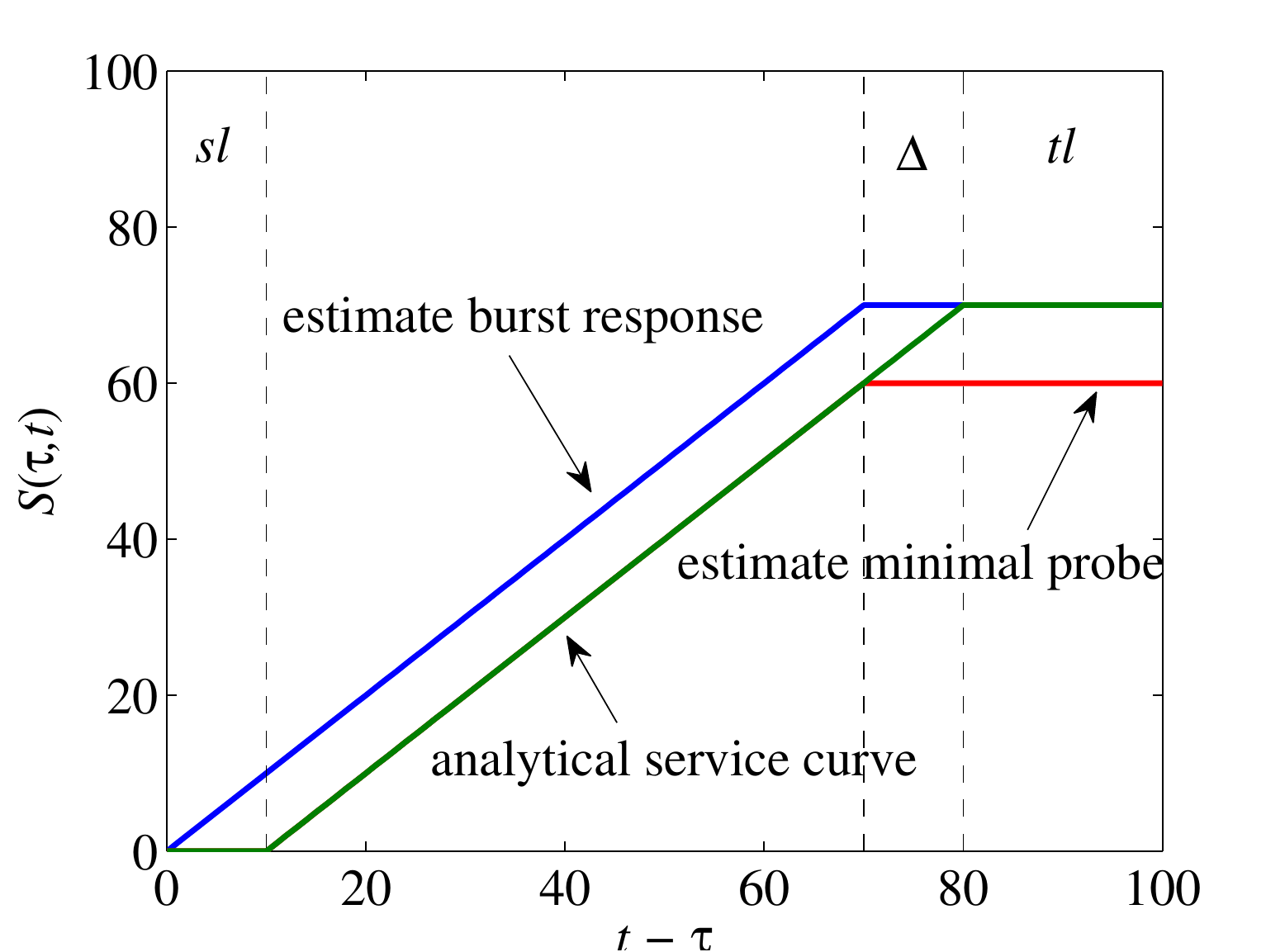}
\label{fig:transientstationarylatencyrate}
}
\caption{Service curve estimates of deterministic sleep scheduling. Latency-rate service curves with a transient latency, with a stationary latency, and with both are compared. The transient latency equals 20 and the stationary latency 10. In case of the transient latency solely, the latency-rate service curve is additive and burst probing recovers the exact result. In contrast, in case of a stationary latency, the service curve is super-additive and burst probing overestimates the service curve. Minimal probing (see Sec.~\ref{sec:minimalprobing}) provides a corresponding lower estimate that matches the service curve exactly in case of additivity.}
\label{fig:latencyrate}
\end{center}
\end{figure*}
Motivated by the min-plus systems theory of the network calculus~\cite{chang:performanceguarantees, leboudec:networkcalculus}, a canonical probe for system identification is the burst function that takes the role of the Dirac delta function in min-plus algebra. The burst function $\delta(t)$ is defined as
\begin{equation}
\delta(t) = \begin{cases} 0 & \text{for } t=0, \\ \infty & \text{for } t > 0. \end{cases}
\label{eq:burstfunction}
\end{equation}
The burst function is the neutral element of min-plus convolution, so that sending a burst probe $A(t) = \delta(t)$ reveals the service $S(0,t)$ for all $t \ge 0$ as the burst response of the system
\begin{equation}
D(t) = \inf_{\tau \in [0,t]} \{ \delta(\tau) + S(\tau,t) \} = S(0,t) .
\label{eq:burstresponse}
\end{equation}
For additive service processes as defined in~\cite[p.~6]{jiang:stochasticnetworkcalculus}, $S(\tau,t)$ can be readily obtained from Eq.~\eqref{eq:burstresponse} for all $t \ge \tau \ge 0$ as
\begin{equation}
S(\tau,t) = S(0,t) - S(0,\tau) .
\label{eq:additivity}
\end{equation}

For a stochastic analysis, we denote $\Omega$ the set of all causally possible sample paths $D_{\omega}(t)$. For each $\omega \in \Omega$ we use the additivity as specified by Eq.~\eqref{eq:additivity} to obtain the service process from the burst response given by Eq.~\eqref{eq:burstresponse} as
\begin{equation}
S_{\omega}(\tau,t) = D_{\omega}(t) - D_{\omega}(\tau)
\label{eq:servicesampleburstestimate}
\end{equation}
for all $\tau \in [0,t]$. We fix $t > 0$ and select a subset of the sample paths $\Psi_t \subseteq \Omega$ with probability $\mathsf{P}[\Psi_t] \ge 1-\varepsilon$. By definition of
\begin{equation}
\mathcal{S}_{br}^{\varepsilon}(\tau,t) = \inf_{\psi \in \Psi_t} \{S_{\psi}(\tau,t)\} ,
\label{eq:effectiveserviceburstrespones}
\end{equation}
it holds that $S_{\psi}(\tau,t) \ge \mathcal{S}_{br}^{\varepsilon}(\tau,t)$ for all $\tau \in [0,t]$ and all $\psi \in \Psi_t$. Further, we have $\mathsf{P}[\Psi_t] \ge 1-\varepsilon$ so that $\mathcal{S}_{br}^{\varepsilon}(\tau,t)$ satisfies Eq.~\eqref{eq:bivariateenvelope}. Hence, $\mathcal{S}_{br}^{\varepsilon}(\tau,t)$ is a non-stationary service curve that conforms to Eq.~\eqref{eq:servicecurve}.

In a practical probing scheme, we can only observe a finite set of sample paths $\Omega$ from repeated measurements of the departures $D_{\omega}(t)$ for $t \ge 0$. We fix $t > 0$ and estimate the service from Eq.~\eqref{eq:servicesampleburstestimate} for all $\tau \in [0,t]$ and $\omega \in \Omega$. To construct $\Psi_t$ we remove a set of minimal sample paths from $\Omega$. We define the minimal sample paths $\Phi$ to be the sample paths that attain the minimum
\begin{equation*}
S_{\min}(\tau,t) = \min_{\omega \in \Omega} \{S_{\omega}(\tau,t)\}
\end{equation*}
for $\tau \in [0,t]$ most frequently. Formally, that is $\Phi = \arg\max_{\omega \in \Omega} \{X_{\omega} \}$, where for all $\omega \in \Omega$
\begin{equation*}
X_{\omega} = \sum_{\tau=0}^{t-1} 1_{S_{\omega}(\tau,t) = S_{\min}(\tau,t)} .
\end{equation*}
The indicator function $1_{(\cdot)}$ is one if the argument is true and zero otherwise. We remove $\Phi$ to obtain $\Psi_t = \Omega \backslash \Phi$ and repeat the above steps for $\Psi_t$ as long as $\mathsf{P}[\Psi_t] \ge 1-\varepsilon$. Finally, we estimate the service curve from Eq.~\eqref{eq:effectiveserviceburstrespones} for all $\tau \in [0,t]$.

Fig.~\ref{fig:burstprobing} shows service curve estimates of burst probing for $t=200$, $300$, and $400$ together with analytical upper bounds and analytical reference service curves for the same system as Fig.~\ref{fig:ratescanning}. We observe that burst probing performs very well as it provides service curve estimates that are in between the analytical reference curves and closely track the analytical upper bound. Unlike rate scanning, burst probing recovers the flat, non-convex areas in the upper right of the analytical curves that are caused by the initial transient phase.

Aside from its good performance, burst probing has, however, limitations. Burst probes have been considered intrusive as they cause non-linear behavior of certain systems. An example are fcfs multiplexers, where a burst probe can preempt other traffic resulting in a too optimistic service estimate~\cite{liebeherr:availbw}.

Furthermore, we discover that the intuitive assumption of additive service processes~\cite{jiang:stochasticnetworkcalculus}, that is used as a basis of the method in Eq.~\eqref{eq:additivity}, is not justifiable in general. While sub-additivity is less an issue, as~\eqref{eq:additivity} provides a conservative estimate of $S(\tau,t)$ that yields a valid lower service curve in this case, we find that there are relevant systems, where super-additivity applies and~\eqref{eq:additivity} overestimates $S(\tau,t)$.

Formally, a function $f(s,t)$ is super-additive if $f(s,u) \ge f(s,t) + f(t,u)$ for all $u \ge t \ge s \ge 0$. Trivially, an additive function $f(s,u) = f(s,t) + f(t,u)$ for all $u \ge t \ge s \ge 0$ is also super-additive. To quantify the magnitude of super-additivity, we define the maximal deviation of a super-additive function $f(s,t)$ from additivity as
\begin{equation}
\Delta(s,u) := f(s,u) - \inf_{t \in [s,u]} \{f(s,t) + f(t,u)\} .
\label{eq:deviationfromadditivity}
\end{equation}

With respect to Eq.~\eqref{eq:additivity}, super-additivity implies that estimating $S(\tau,t)$ as $S (\tau,t) = S(0,t) - S(0,\tau)$ is too optimistic as in general only $S(\tau,t) \le S(0,t)-S(0,\tau)$ holds.
\subsubsection{Super-additive Service Processes}
\label{sec:superadditiveprocesses}
To provide insights on the occurrence and the impact of non-additivity, we first consider the instructive example of deterministic latency-rate functions as specified by $S^{tlr}(\tau,t)$ in Eq.~\eqref{eq:transientlatencyrate} for the case of a transient latency and $S^{slr}(\tau,t)$ in Eq.~\eqref{eq:stationarylatencyrate} for a stationary latency, respectively. In case of a transient latency, it can be verified that $S^{tlr}(\tau,t)$ is additive. For a stationary latency, however, we find that $S^{slr}(\tau,t)$ is super-additive but not generally additive. An example is $S^{slr}(0,\tau) + S^{slr}(\tau,t) = R (t-2T)$ for $t \ge \tau+T$ and $\tau \ge T$, whereas $S^{slr}(0,t) = R (t-T)$. The example also achieves the maximal deviation from additivity $\Delta(\tau,t) = RT$ for $t \ge \tau + 2T$.

For illustration, Fig.~\ref{fig:latencyrate} shows different types of latency-rate functions, labelled analytical service curve: with a transient latency of $T=20$ according to Eq.~\eqref{eq:transientlatencyrate} in Fig.~\ref{fig:transientlatencyrate}, as it applies for deterministic sleep scheduling; a stationary latency of $T=10$ according to Eq.~\eqref{eq:stationarylatencyrate} in Fig.~\ref{fig:stationarylatencyrate}, for example to model a constant propagation delay; and for a combined system with both a transient and a stationary latency in Fig.~\ref{fig:transientstationarylatencyrate}. The rate is $R=1$. In all cases, we fix $t=100$ and show the service curves for an increasing interval $t-\tau$.

As can be seen from the flat end of the curve in the region marked with {\it tl} in Fig.~\ref{fig:transientlatencyrate}, the transient latency has an effect only for large intervals $t-\tau > 80$, respectively, $\tau < 20$. The stationary latency, on the other hand, applies for all $\tau$ resulting in the right shift of the curve in Fig.~\ref{fig:stationarylatencyrate} marked with {\it sl}. Finally, Fig.~\ref{fig:transientstationarylatencyrate} shows both effects.

We also present estimates of the service curves as obtained from the burst response. Clearly, the burst response recovers the transient latency-rate function in Fig.~\ref{fig:transientlatencyrate} exactly. It overestimates, however, the stationary latency-rate function in Fig.~\ref{fig:stationarylatencyrate} and similarly in Fig.~\ref{fig:transientstationarylatencyrate}. The effect is due to the super-additivity induced by the stationary latency. In this case, the assumption of additivity used by Eq.~\eqref{eq:additivity} falsely assigns the stationary latency to large intervals $t-\tau$ only, as seen in the region marked with $\Delta$. The figures further include estimates obtained by a method that we introduce as minimal probing in Sec.~\ref{sec:minimalprobing}, where we also provide details on the curves.
\subsubsection{Super-additivity of $\min$ and $\otimes$}
\label{sec:superadditivity}
A second notable exception of additive service processes are networks of systems, where the end-to-end network service process $S^{net}(\tau,t)$ is derived by min-plus convolution of the service processes of the individual systems $S^i(\tau,t)$ for $i = 1,2,\dots n$, see Eq.~\eqref{eq:networkserviceprocess}. For additive $S^i$, the following Lem.~\ref{lem:additivity} proves that $S^{net}$ is super-additive in general and additive only in certain special cases. The result is significant as it refutes the assumption of additive service processes for the large class of tandem systems. As before, it implies that the burst probing method is too optimistic.

In Fig.~\ref{fig:superadditivity}, we evaluate the network service process $S^{net}(\tau,t)$ for a tandem of $n=2,3,4$ systems with random sleep scheduling as in~Sec.~\ref{sec:systemmodelrandom}, i.e., the service processes $S^i(\tau,t)$ are additive. We consider $10^5$ sample paths and show the distribution of the relative deviation of $S^{net}(0,400)$ from additivity, i.e., $\Delta(\tau,t)/S^{net}(\tau,t)$. Fig.~\ref{fig:superadditivity} confirms significant deviations from additivity.

The following lemmas formalize the results and show that the $\otimes$ and $\min$ operators in general maintain only super-additivity but not additivity. The proofs are in the appendix.
\begin{lem}[Super-additivity of $\otimes$]
\label{lem:additivity}
Given two bivariate functions $f(s,t)$ and $g(s,t)$ for $t \ge s \ge 0$ where $f(t,t), g(t,t) = 0$ for all $t \ge 0$. Define $h(s,t) = f \otimes g(s,t)$.
\begin{enumerate}
\renewcommand{\theenumi}{\roman{enumi}}
\item If $f$ and $g$ are super-additive, then $h$ is super-additive.
\item If $f$ and $g$ are additive and univariate, then $h$ is additive.
\end{enumerate}
\end{lem}
From Lem.~\ref{lem:additivity} the convolution of two additive bivariate functions $f(s,t)$ and $g(s,t)$ for $t \ge s \ge 0$ is super-additive, however, in general it is not additive. As a counterexample consider the additive functions $f(s,t) = t-s$ and $g(s,t) = 2(\lfloor t/2 \rfloor - \lfloor s/2 \rfloor)$. It follows that $h(s,t) = f \otimes g(s,t) = 2 \lfloor t/2 \rfloor - s$ is not additive. Additivity is, however, achieved in case of univariate functions, where Lem.~\ref{lem:additivity} ii) extends a known result for min-plus convolution of sub-additive univariate functions in~\cite[p. 142]{leboudec:networkcalculus}.
\begin{lem}[Super-additivity of $\min$]
\label{lem:additivitymin}
Given two super-additive bivariate functions $f(s,t)$ and $g(s,t)$ for $t \ge s \ge 0$. The minimum $h(s,t) = \min \{f(s,t), g(s,t)\}$ is super-additive.
\end{lem}
The minimum of two additive bivariate functions $f(s,t)$ and $g(s,t)$ for $t \ge s \ge 0$ is super-additive, as a special case of Lem.~\ref{lem:additivitymin}, but in general not additive. A counterexample is $f(s,t) = t-s$ and $g(s,t) = 2(\lfloor t/2 \rfloor - \lfloor s/2 \rfloor)$. Clearly $f$ and $g$ are additive, however, $h = \min \{f,g\} = 2 \lfloor t/2 \rfloor - s$ is not.
\begin{figure}
\hspace{-10pt}
\subfigure[Deviation from additivity]{
\includegraphics[width=0.51\columnwidth]{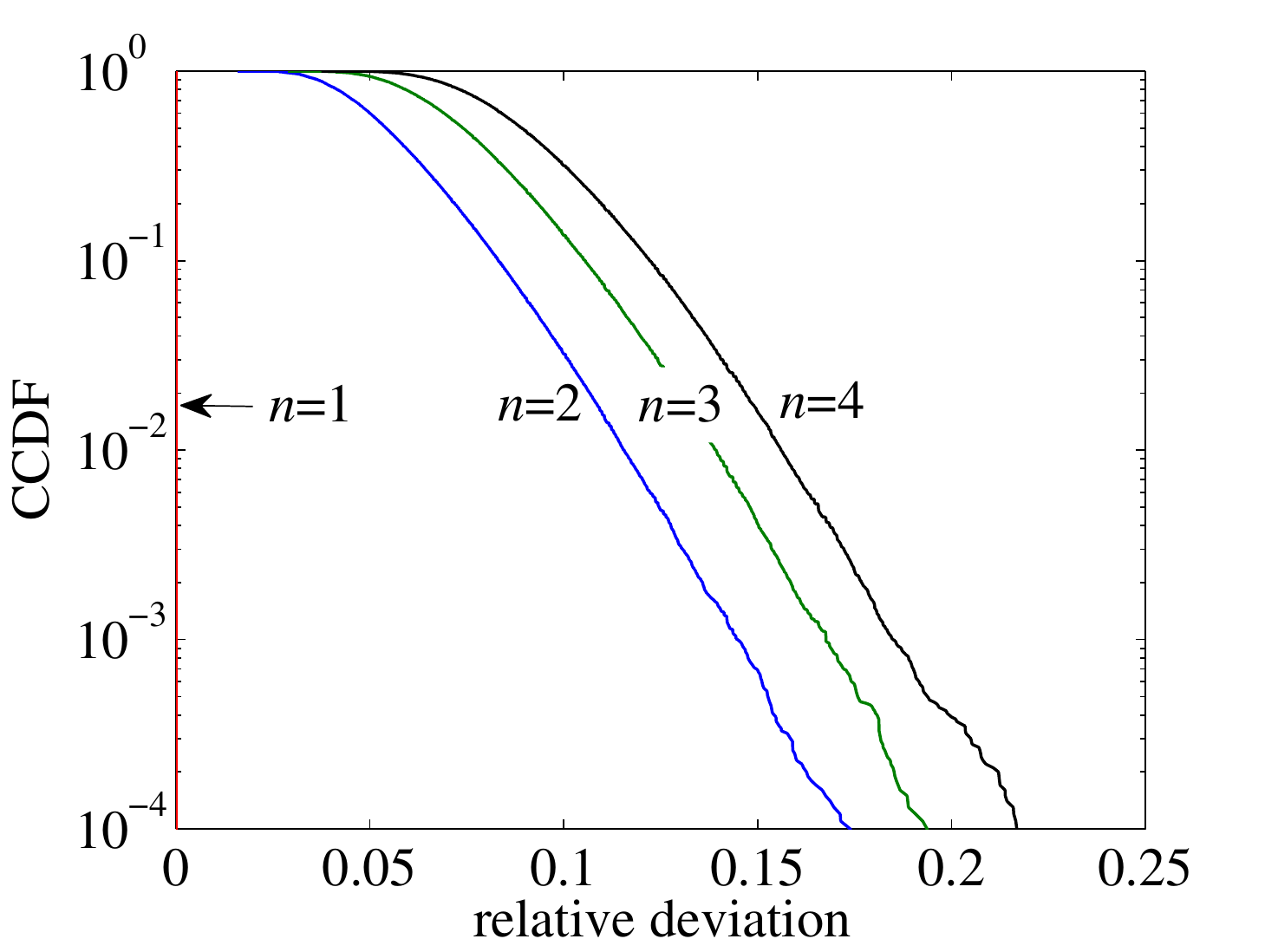}
\label{fig:superadditivity}
}
\hspace{-10pt}
\subfigure[Accuracy of minimal probing]{
\includegraphics[width=0.51\columnwidth]{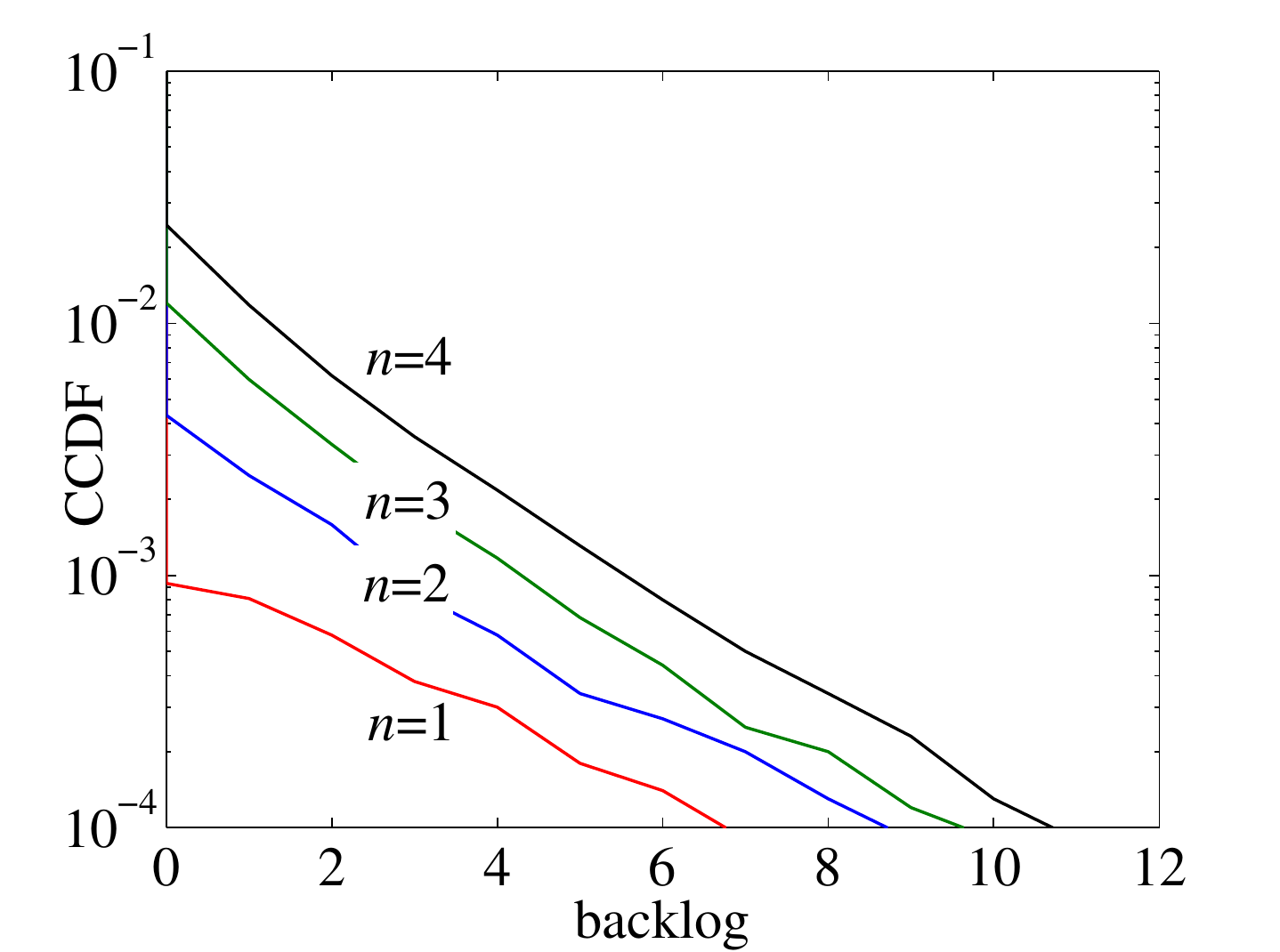}
\label{fig:minimalprobing}
\hspace{-10pt}
}
\caption{Network of $n$ systems with random sleep scheduling in series. (a) The network service process deviates from additivity. (b) Minimal probing achieves small backlogs, corresponding to a high accuracy of the estimate.}
\label{fig:deviationsfromadditivity}
\end{figure}
%
%
\subsection{Minimal Probing}
\label{sec:minimalprobing}
Given the limitations identified above, we devise a new probing method to characterize transient service processes that are neither convex nor additive. The method comprises two phases. In the first step, a minimal probe as well as an upper bound of the service are estimated using the burst response as in Sec.~\ref{sec:burstresponse}. In the second step, the minimal probe is used to estimate a non-stationary service curve with a defined accuracy. We show that the minimal probe reveals the service of the system, whereas any smaller or larger probe estimates only a lower bound. The importance of the probe traffic intensity is also discussed in~\cite{liebeherr:availbw}. Estimation methods for stationary systems that propose techniques for minimal backlogging are~\cite{valaee:adhocadmissioncontrol, nam:minimalbackloggingbwest}, where the transmission of each probe packet triggers the generation of a new one.
\subsubsection{Estimation using Arbitrary Probes}
First, we consider how to obtain a service curve estimate from an arbitrary probe $A(\tau)$. The estimate will then be used to derive conditions for the shape of a minimal probe. From Eq.~\eqref{eq:linearsystem} it follows that $D(t) \le A(\tau) + S(\tau,t)$ for all $\tau \in [0,t]$ so that
\begin{equation}
S(\tau,t) \ge D(t) - A(\tau),
\label{eq:estimategeneral}
\end{equation}
for all $\tau \in [0,t]$. An equivalent expression using the backlog is $S(\tau,t) \ge A(\tau,t)-B(t)$ for all $\tau \in [0,t]$. By insertion of the backlog quantile it follows that
\begin{equation}
\mathcal{S}^{\varepsilon}(\tau,t) = A(\tau,t) - B^{\varepsilon}(t)
\label{eq:estimateminimalprobing}
\end{equation}
satisfies Eq.~\eqref{eq:bivariateenvelope}, i.e., $\mathcal{S}^{\varepsilon}(\tau,t)$ is a non-stationary service curve as defined by Eq.~\eqref{eq:servicecurve}. To estimate $\mathcal{S}^{\varepsilon}(\tau,t)$, the probing procedure uses a defined probe $A(\tau,t)$ to obtain $B^{\varepsilon}(t)$ from repeated backlog measurements. In the special case $A(\tau,t) = r (t-\tau)$ the rate scanning method can be recovered.
\subsubsection{Definition of Minimal Probe}
So far, we did not constrain the shape of the probe and defining a suitable probe is non-trivial. Intuitively, a probe that is too small will provide little information about the service as the observed departures are mostly limited by the arrivals. A too large probe, on the other hand, will deteriorate the estimate, e.g., in the extreme case of a burst probe $A(\tau) = \delta(\tau)$, the lower bound in Eq.~\eqref{eq:estimategeneral} will only be useful for $\tau=0$ but not for $\tau > 0$. The same restriction applied to burst probing earlier.
The following lemma defines a necessary and sufficient condition for a minimal probe to obtain the true service from Eq.~\eqref{eq:estimategeneral}.
\begin{lem}[Minimal Probe]
\label{lem:minimalprobe}
Fix $t>0$ and define the minimal probe
\begin{equation}
A_{mp}(\tau) = S(0,t) - S(\tau,t) ,
\label{eq:minimalprobe}
\end{equation}
for $\tau \in [0,t]$. Eq.~\eqref{eq:estimategeneral} holds with equality if and only if $A(\tau) = A_{mp}(\tau)$ for all $\tau \in [0,t]$.
\end{lem}
\begin{proof}
By substitution of Eq.~\eqref{eq:minimalprobe} into Eq.~\eqref{eq:linearsystem} we have $D(t) = S(0,t)$ and finally by insertion into Eq.~\eqref{eq:estimategeneral} $D(t)-A_{mp}(\tau) = S(\tau,t)$ for all $\tau \in [0,t]$, which proves that Eq.~\eqref{eq:minimalprobe} is sufficient. To see that Eq.~\eqref{eq:minimalprobe} is necessary, consider any other probe $A(\tau)$ expressed as $A(\tau) = A_{mp}(\tau) \pm f(\tau)$ where $f(0) = 0$ and $\exists \tau \in (0,t]: f(\tau) \neq 0$. It follows by the same steps that $D(t)-A(\tau) = S(\tau,t) + \inf_{\upsilon \in [0,t]} \{\pm f(\upsilon) \} \mp f(\tau)$ implying that $\exists \tau \in [0,t]: D(t)-A(\tau) < S(\tau,t)$.
\end{proof}
While Lem.~\ref{lem:minimalprobe} proves the optimality of $A_{mp}(\tau)$, it depends on the unknown service and cannot be constructed a priori. To gather information on $A_{mp}(\tau)$, we initially consider the system's burst response $D(\tau) = \delta \otimes S(\tau) = S(0,\tau)$ and estimate Eq.~\eqref{eq:minimalprobe} by $\widetilde{A}_{mp}(\tau) = S(0,\tau)$ for $\tau \in [0,t]$. We use tilde to indicate that $\widetilde{A}_{mp}(\tau)$ is an approximation of $A_{mp}(\tau)$ that is exact only if $S(\tau,t)$ is additive. In the stochastic case, we employ $\mathcal{S}_{br}^{\varepsilon}(\tau,t)$ from Eq.~\eqref{eq:effectiveserviceburstrespones} to estimate
\begin{equation}
\widetilde{A}_{mp}(\tau) = \mathcal{S}_{br}^{\varepsilon}(0,t) - \mathcal{S}_{br}^{\varepsilon}(\tau,t) ,
\label{eq:minimalprobestochastic}
\end{equation}
for $\tau \in [0,t]$. We note that additivity of $\mathcal{S}_{br}^{\varepsilon}(\tau,t)$ cannot be assumed by construction of Eq.~\eqref{eq:effectiveserviceburstrespones}, see Lem.~\ref{lem:additivitymin}.
\subsubsection{Accuracy of the Estimates}
\label{sec:accuracy}
We show that the backlog at the end of the probe $\widetilde{A}_{mp}(\tau)$ is a measure of the estimation accuracy. Further, we verify that the backlog is bounded by the deviation of a super-additive service process from additivity.

First we investigate a time-variant, deterministic system, where Eq.~\eqref{eq:estimategeneral} provides a lower estimate of the service. The estimate of the minimal probe $\widetilde{A}_{mp}(\tau) = S(0,\tau)$ is obtained as the system's burst response from Eq.~\eqref{eq:burstresponse}. By insertion into Eq.~\eqref{eq:estimategeneral}, a lower service estimate is
\begin{equation*}
S(\tau,t) \ge \inf_{\upsilon \in [0,t]} \{S(0,\upsilon) + S(\upsilon,t)\} - S(0,\tau),
\end{equation*}
where we used Eq.~\eqref{eq:linearsystem} to compute $D(t)$ for $\widetilde{A}_{mp}(\tau)$. On the other hand, the burst response provides the upper estimate $S(\tau,t) \le S(0,t) - S(0,\tau)$ if $S(\tau,t)$ is super-additive. Taking the difference of the upper and the lower estimate, the service process is bounded in an interval of width
\begin{equation*}
\Delta(0,t) = S(0,t) -  \inf_{\upsilon \in [0,t]} \{S(0,\upsilon) + S(\upsilon,t)\}
\end{equation*}
that is the maximum deviation of $S(0,t)$ from additivity as defined in Eq.~\eqref{eq:deviationfromadditivity}. Further, since $S(0,t) = \widetilde{A}_{mp}(t)$ we have
\begin{equation}
\Delta(0,t) = \widetilde{A}_{mp}(t) - D(t) = B(t),
\label{eq:estimateaccuracydeterministic}
\end{equation}
i.e., the estimation accuracy is determined by the backlog at the end of the probe. If $S(\tau,t)$ is additive, $\Delta(0,t)=0$, i.e., both the lower and the upper estimate recover $S(\tau,t)$ exactly. Further, the backlog $B(\tau)$ that is caused by $\widetilde{A}_{mp}(\tau)$ is zero during the entire probe if $S(\tau,t)$ is additive.

In the stochastic case, we obtain a non-stationary service curve estimate from Eq.~\eqref{eq:estimateminimalprobing} as $\mathcal{S}_{mp}^{\varepsilon}(\tau,t) = \widetilde{A}_{mp}(\tau,t) - B^{\varepsilon}(t)$ using the probe defined in Eq.~\eqref{eq:minimalprobestochastic}. Further, we have by use of Eq.~\eqref{eq:minimalprobestochastic} that $\widetilde{A}_{mp}(\tau,t) = \widetilde{A}_{mp}(t) - \widetilde{A}_{mp}(\tau) = \mathcal{S}_{br}^{\varepsilon}(\tau,t)$ since $\mathcal{S}_{br}^{\varepsilon}(t,t) = 0$ by definition of Eqs.~\eqref{eq:servicesampleburstestimate} and~\eqref{eq:effectiveserviceburstrespones}. By insertion into Eq.~\eqref{eq:estimateminimalprobing} it holds for all $\tau \in [0,t]$ that
\begin{equation}
\mathcal{S}_{mp}^{\varepsilon}(\tau,t) = \mathcal{S}_{br}^{\varepsilon}(\tau,t) - B^{\varepsilon}(t) .
\label{eq:estimateaccuracy}
\end{equation}
We conclude that $B^{\varepsilon}(t)$ observed by minimal probing at the end of the probe is a measure of accuracy that separates the generally conservative estimate of minimal probing from the possibly too optimistic estimate of burst probing.

To investigate $B^{\varepsilon}(t)$, we consider the backlog expression $B(t) = \sup_{\tau \in [0,t]} \{ A(\tau,t) - S(\tau,t) \}$ that follows from Eq.~\eqref{eq:linearsystem}. By insertion of the probe defined in Eq.~\eqref{eq:minimalprobestochastic} we have $B(t) = \sup_{\tau \in [0,t]} \{ \mathcal{S}_{br}^{\varepsilon}(\tau,t) - S(\tau,t) \}$. Next, we substitute $\mathcal{S}_{br}^{\varepsilon}(\tau,t) = \inf_{\psi \in \Psi_t} \{S_{\psi}(0,t) - S_{\psi}(0,\tau)\}$ from Eq.~\eqref{eq:effectiveserviceburstrespones} so that for any sample path $\varphi \in \Psi_t$ it holds that
\begin{align*}
B_{\varphi}(t) &= \sup_{\tau \in [0,t]} \Bigl\{ \inf_{\psi \in \Psi_t} \{S_{\psi}(0,t) - S_{\psi}(0,\tau)\} - S_{\varphi}(\tau,t) \Bigr\} \\
&\le \sup_{\tau \in [0,t]} \{ S_{\varphi}(0,t) - S_{\varphi}(0,\tau) - S_{\varphi}(\tau,t) \} \\
&= S_{\varphi}(0,t) - \inf_{\tau \in [0,t]} \{S_{\varphi}(0,\tau) + S_{\varphi}(\tau,t) \} = \Delta_{\varphi}(0,t) ,
\end{align*}
i.e., $B_{\varphi}(t)$ is bounded by the maximal deviation of $S_{\varphi}(0,t)$ from additivity, defined as $\Delta_{\varphi}(0,t)$ in Eq.~\eqref{eq:deviationfromadditivity}. If $S(\tau,t)$ is additive, it follows that $B_{\varphi}(t) = 0$ for all $\varphi \in \Psi_t$. Since $\mathsf{P}[\Psi_t] \ge 1-\varepsilon$, it holds that $B^{\varepsilon}(t) = 0$ and $\mathcal{S}_{mp}^{\varepsilon}(\tau,t)$ recovers $\mathcal{S}_{br}^{\varepsilon}(\tau,t)$ exactly. The same applies if $S(\tau,t)$ is sub-additive.

For a first example, we consider the latency-rate service curves of deterministic sleep scheduling in Fig.~\ref{fig:latencyrate}, where we also include the results of minimal probing. For the case of a transient latency depicted in Fig.~\ref{fig:transientlatencyrate}, minimal probing recovers the estimate of burst probing, which confirms that the estimate is exact. This is due to the additivity of the service.

In case of a stationary latency, as in Figs.~\ref{fig:stationarylatencyrate} and~\ref{fig:transientstationarylatencyrate}, the estimates of minimal probing and burst probing differ by a constant offset that is equal to the backlog at the end of the minimal probe $B(t)=10$, as derived by Eq.~\eqref{eq:estimateaccuracydeterministic}. The deviation indicates that the estimate from burst probing is overly optimistic. This is a result of the super-additivity of the service that has a maximal deviation from additivity defined by Eq.~\eqref{eq:deviationfromadditivity} of $\Delta(0,t) = R T = B(t)$.

Unlike burst probing, minimal probing correctly identifies the stationary latency in the region marked with {\it sl} in Figs.~\ref{fig:stationarylatencyrate} and~\ref{fig:transientstationarylatencyrate}. It matches the analytical service curve up to the region marked with $\Delta$, where it provides a conservative estimate with an accuracy of $\Delta(0,t)$. Finally, it reproduces the flat area in the region marked with {\it tl} in Figs.~\ref{fig:transientlatencyrate} and~\ref{fig:transientstationarylatencyrate} that is due to the transient latency.

\begin{figure}
\hspace{-10pt}
\subfigure[Transient latency]{
\includegraphics[width=0.51\columnwidth]{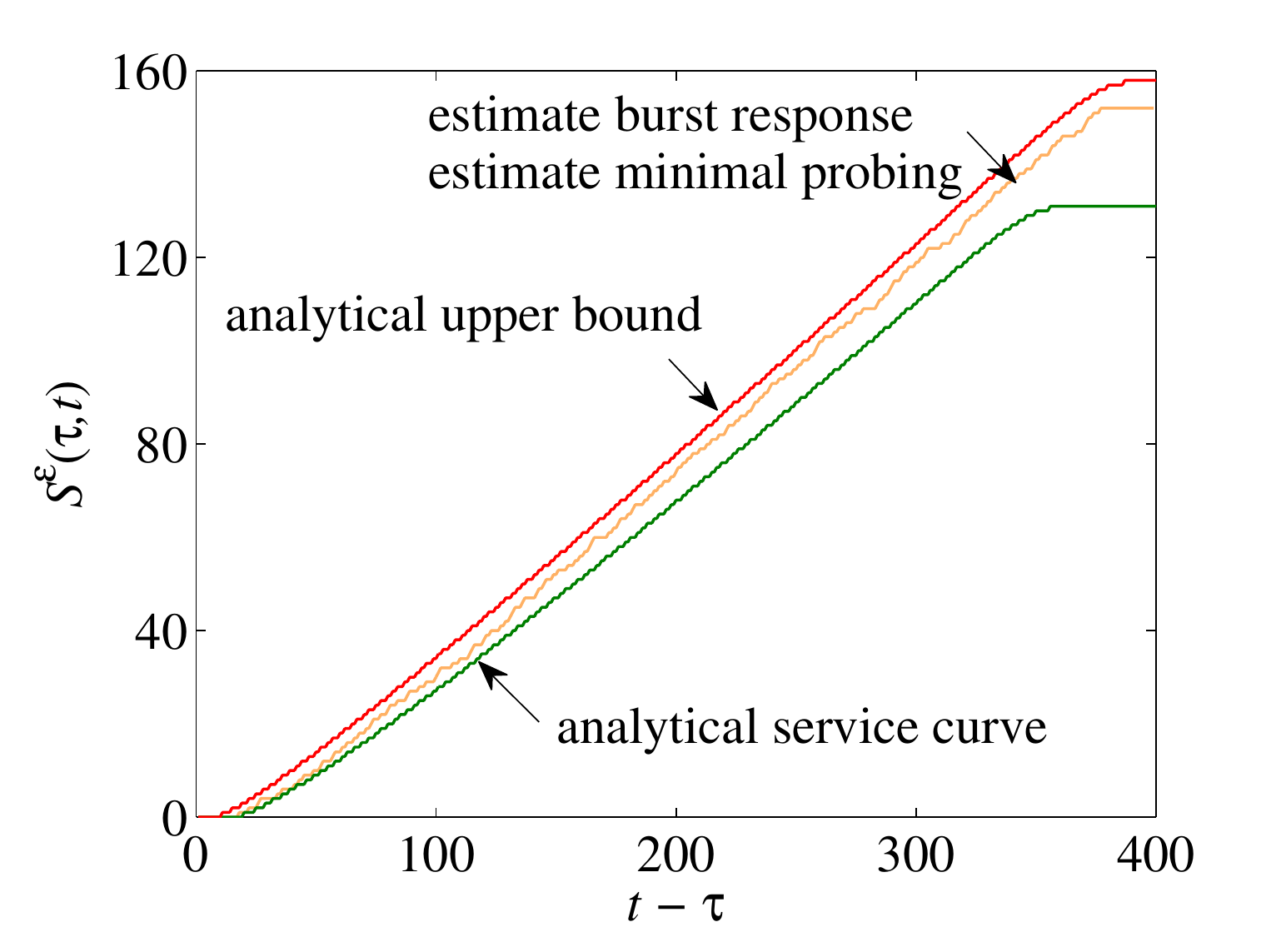}
\label{fig:minimalprobingtransient}
}
\hspace{-10pt}
\subfigure[Transient and stationary latency]{
\includegraphics[width=0.51\columnwidth]{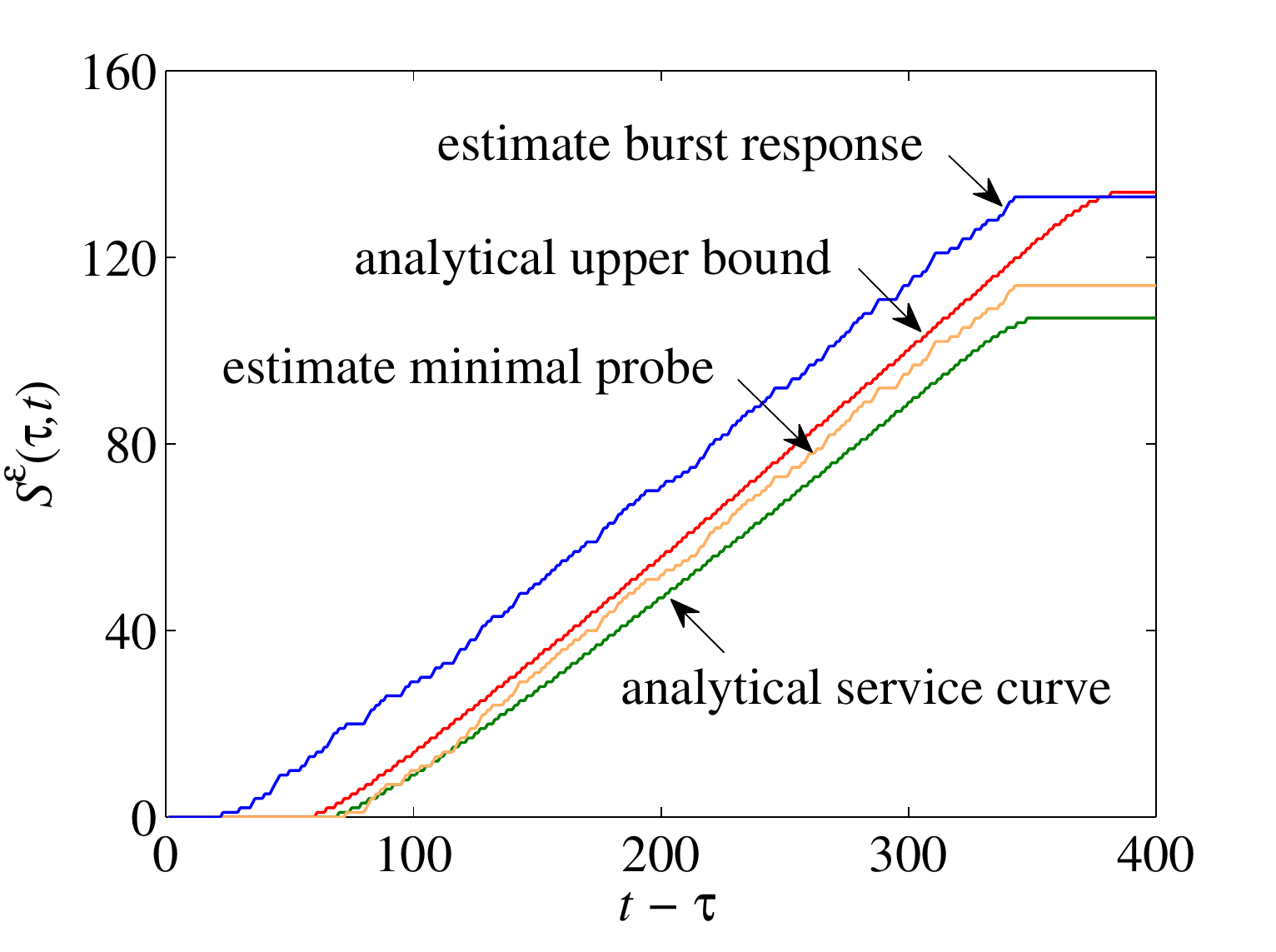}
\label{fig:minimalprobingtransientstationary}
\hspace{-10pt}
}
\caption{Service curve estimates of random sleep scheduling plus a stationary latency. The estimate of minimal probing stays between the analytical curves, whereas burst probing exceeds the upper bound in case of a stationary latency.}
\label{fig:sc_minimalprobing}
\end{figure}
Next, we regard random sleep scheduling as in Fig.~\ref{fig:servicecurveestimates} where we find that the estimates of minimal probing and burst probing in Fig.~\ref{fig:minimalprobingtransient} match. As before, we include as a reference an analytical upper bound as well as an analytical service curve that provides the lower guarantee specified by Eq.~\eqref{eq:servicecurve}. Further, in Fig.~\ref{fig:minimalprobingtransientstationary} we include a stationary latency of 50. In this case, the estimates of burst probing and minimal probing differ: Both identify the correct rate, i.e., the slope, as well as the transient latency, expressed by the flat area in the upper right. In the lower left, a service of zero is caused by outages of the Bernoulli increment process as well as by the stationary latency that is recovered, however, only by minimal probing. Due to the super-additivity of the service process, the estimate of burst probing is too optimistic and exceeds the analytical upper bound. Minimal probing, on the other hand, provides a valid service curve that resides between the analytical reference curves.

Finally, Fig.~\ref{fig:minimalprobing} quantifies the distribution of $B(t)$ that is observed by minimal probing at $t=400$ for a network of $n = 1 \dots 4$ systems in series, each with random sleep scheduling as in Fig.~\ref{fig:servicecurveestimates}. The network service process is additive for $n=1$, but not for $n > 1$ as confirmed by Fig.~\ref{fig:superadditivity}. Since in general, it is not known whether the service process is additive or not, estimates of burst probing are not reliable. Using minimal probing, we can either if $B^{\varepsilon}(t) = 0$ confirm the estimate of burst probing or otherwise obtain a conservative estimate with a defined accuracy that is given by $B^{\varepsilon}(t)$. In Fig.~\ref{fig:minimalprobing} we observe that $B^{\varepsilon}(t) > 0$ for $\varepsilon = 10^{-3}$ and $n > 1$, i.e., the estimate from burst probing is not confirmed. $B^{\varepsilon}(t)$ is, however, small, i.e, minimal probing is accurate.
%
%
\section{Estimation of Cellular Sleep Scheduling}
\label{sec:measurements}
We use the minimal probing method from Sec.~\ref{sec:minimalprobing} to estimate service curves of cellular networks with sleep scheduling, specifically DRX. Compared to state-of-the-art steady-state solutions of semi-Markov DRX models~\cite{yang:modelingumtspowersaving, zhou:performanceltedrx, bhamber:analyticLTEpowersavingburstytraffic, wu:performancedrx, zhou:ltedrxm2m}, our non-stationary service curve approach provides transient performance measures, such as the transient overshoot and corresponding relaxation time. Further, the service curve applies for arbitrary traffic arrivals, including TCP traffic, and is not limited to Markovian systems. While we investigate sleep scheduling of the mobile that is woken up for transmission of pending uplink data, we expect that similar studies can be performed for the downlink as well as for green cellular networks with base station sleeping~\cite{tabassum:basestationsleeping}.

In the following, we show our cellular measurement setup (Sec.~\ref{sec:measurementsetup}). We present estimates of service curves that explain characteristics of the cellular data service, including transient delays caused by sleep scheduling, capacity limits, and service outages, e.g., due to the radio channel (Sec.~\ref{sec:transientlteservice}). We conclude this section with a comparison of service curves obtained for different technologies, i.e., 2G (EDGE), 3G (HSPA), and 4G (LTE) (Sec.~\ref{sec:comparisonEDGEHSPA}), as well as for different times of day showing typical diurnal behavior (Sec.~\ref{sec:diurnalcharacteristics}).
%
%
\subsection{Measurement Setup}
\label{sec:measurementsetup}
\begin{figure}
\centering
\includegraphics[width=0.8\columnwidth]{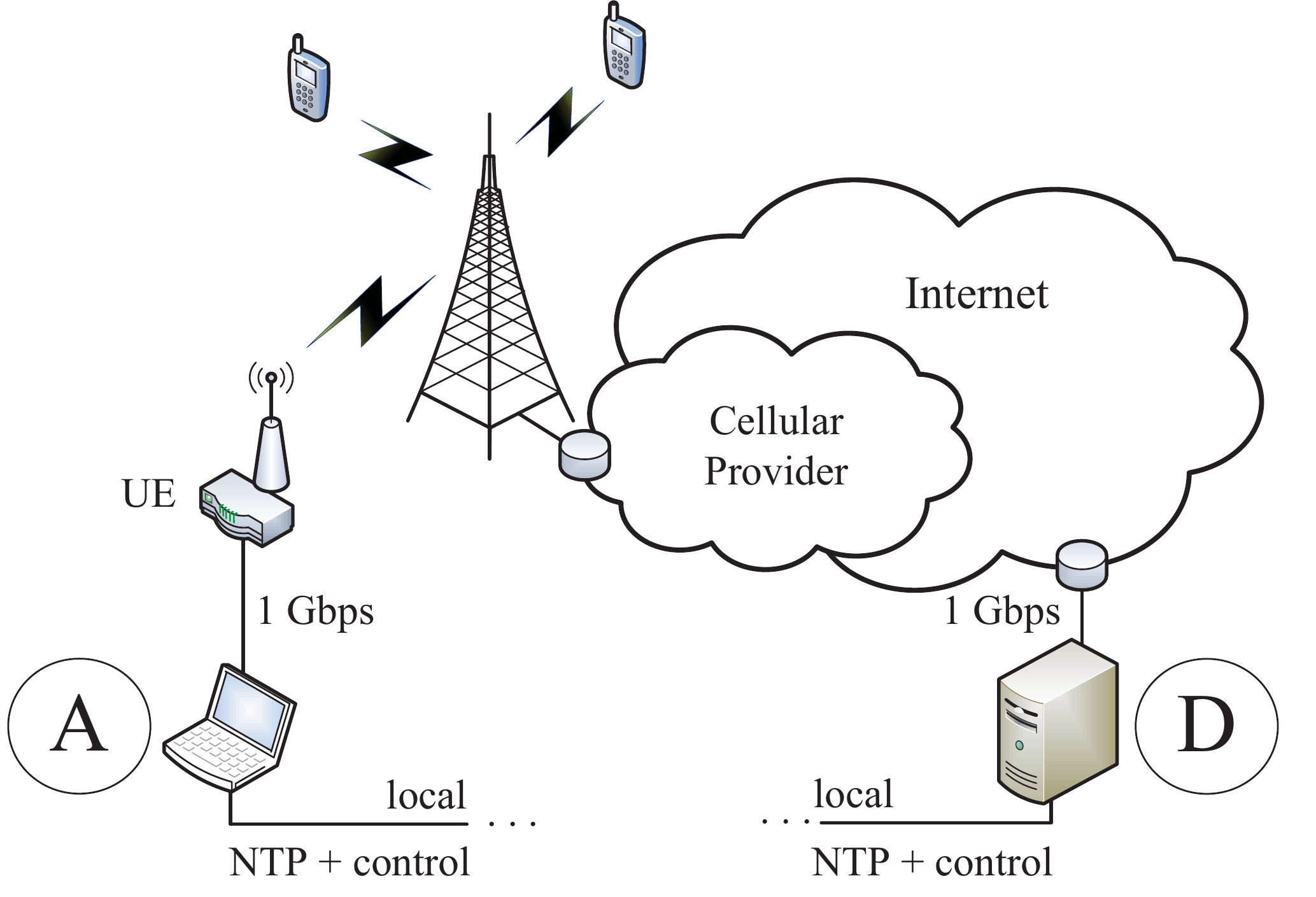}
\caption{The measurement setup comprises a cellular data connection from client (A) to server (D) for estimation, and a separated local control network.}
\label{fig:testbed}
\end{figure}

The main components of our measurement setup are displayed in Fig.~\ref{fig:testbed}. We operate in our lab a cellular client (A) and a wired server (D), that are connected to the Internet by a major German commercial cellular provider and the German National Research and Education Network, respectively. For the cellular connection, the user equipment (UE) is a stationary category 3 Teldat RS232j-4G modem for EDGE and LTE and a Teltonika HSPA+ RUT500 modem for HSPA. The nominal uplink rates as stated by the network provider are 220 kbps for EDGE, 5,76 Mbps for HSPA, and 50 Mbps for LTE. The wired connections are 1 Gbps Ethernet links.

In addition to the cellular data connection, we maintain a separated local control network for the client and the server that permits time synchronization in the order of a few ms using the Network Time Protocol (NTP). Further, it enables operating the client and the server remotely. To automate the measurements we use the tool \textit{sshlauncher}\footnote{https://github.com/bozakov/sshlauncher}, that facilitates repeated execution of distributed network experiments. We repeat each measurement 100 times to obtain a statistical basis.

We use the UDP traffic generator \textit{rude\&crude}\footnote{http://rude.sourceforge.net/} for transmission of probe traffic by the client. To measure the probe arrivals $A(t)$ and departures $D(t)$, packet traces are captured by \textit{libpcap} at the client and the server, respectively. We choose packets of 1400 Bytes size for HSPA and LTE and 500 Bytes for EDGE to accommodate the low uplink data rate.
%
%
\subsection{Transient Service of LTE}
\label{sec:transientlteservice}
We evaluate the performance impact of the transient phase that occurs if an uplink transmission is requested while the LTE UE is dormant, i.e., the radio resource control protocol is in idle state, see Fig.~\ref{fig:sleepschedule}. After each measurement, a sufficiently long pause ensures that the UE enters idle state again. For the network under observation, this amounts to 10.5 sec~\cite{becker:lte}.
\subsubsection{Transient Overshoot and Relaxation Time}\label{sec:backlogcellular}
\begin{figure}
\hspace{-10pt}
\subfigure[Mean]{
\includegraphics[width=0.51\columnwidth]{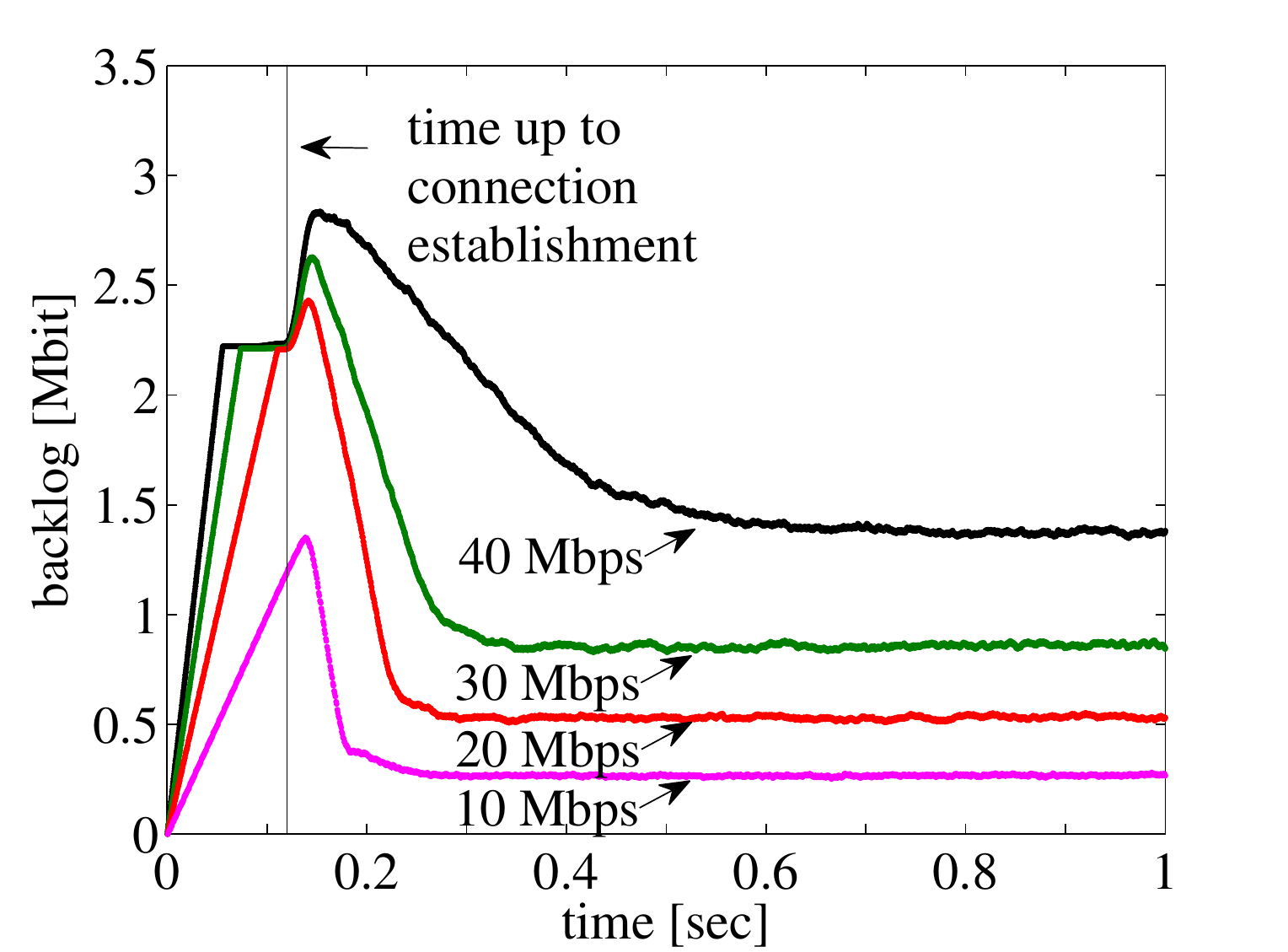}
\label{fig:lte_mean}
}
\hspace{-10pt}
\subfigure[0.95-Quantile]{
\includegraphics[width=0.51\columnwidth]{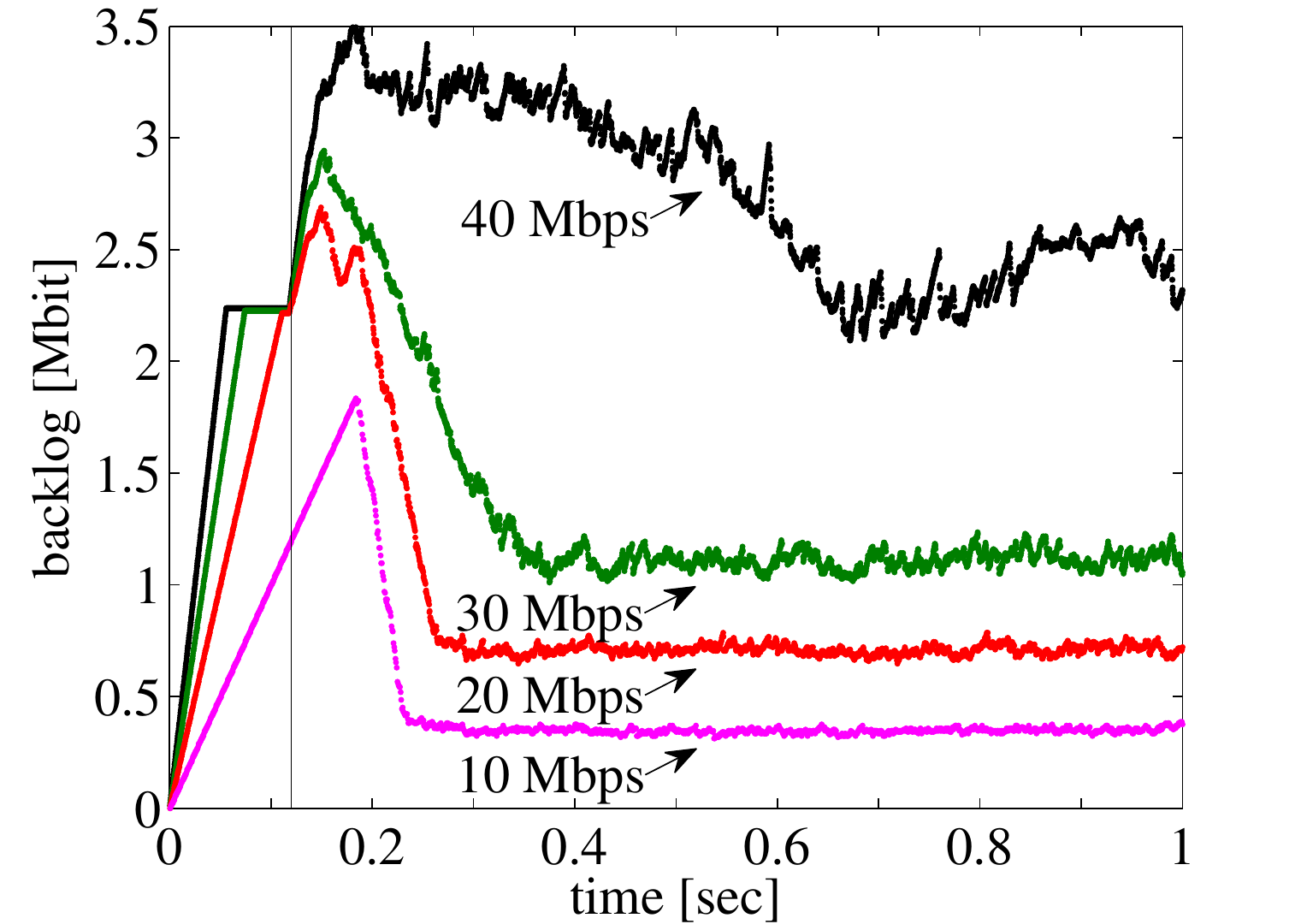}
\label{fig:lte_quantile}
\hspace{-10pt}
}
\caption{Transient backlog of LTE for CBR traffic.}
\label{fig:backlog_lte}
\end{figure}
First, we consider the backlog process that is induced by constant bit rate (CBR) traffic with rates from 10 to 40 Mbps. Corresponding mean backlogs and 0.95-quantiles are depicted in Fig.~\ref{fig:backlog_lte}. All curves show the typical transient overshoot and approach stationarity after a certain relaxation time, similar to the model in Fig.~\ref{fig:load_qt}. Next, we investigate the different phases in detail.

At startup, connection establishment is triggered and data are queued in a buffer at the UE. For CBR traffic the backlog increases linearly with a slope that corresponds to the traffic rate, e.g., after 120 ms, marked by a vertical line in Fig.~\ref{fig:backlog_lte}, the backlog amounts to 1.2 Mb for a rate of 10 Mbps.

While we do not observe packet loss at 10 Mbps, packet loss is measured at rates of 20 Mbps and above and occurs deterministically once the backlog reaches 2.2 Mb, respectively, 200 packets. To see the deterministic behavior, note that both the mean and the 0.95-quantile show an identical plateau of 2.2 Mb that extends until 120 ms. Experiments with different packet sizes substantiate a buffer limit in the UE of about 200 packets, as packet loss occurs regardless of the packet size once the backlog reaches 200 packets.

Interestingly, the backlog starts to grow again after 120 ms, that is the time to establish the connection~\cite{becker:lte}. The increase of the backlog, computed as $B(t) = A(t) - D(t)$, is caused by packets that have been served from the buffer but are still in transmission while new packets enter the buffer at the offered traffic rate. For the mean backlog, the effect lasts for about 25 ms, that corresponds to the one-way delay (OWD) observed in~\cite{becker:lte}, until packets depart from the network. Regarding the 0.95-quantile, the effect extends to up to 50 ms.

Afterwards, the backlog is depleted at a rate that is determined as the difference of the service rate and the traffic arrival rate. For traffic rates that are close to the capacity limit of about 45 Mbps, this causes significantly prolonged relaxation times until the backlog overshoot is cleared. Further, we observe an increasing volatility for higher traffic rates, as can be seen from the 0.95-quantile that grows more strongly than the mean.

Eventually, the backlog approaches stationarity, where it is mostly caused by packets in transmission. Hence, the stationary mean backlog can be approximately determined as the product of the OWD of 25 ms and the traffic rate, e.g., for a rate of 10 Mbps a mean backlog of 0.25 Mb applies. The 0.95-quantile depicts larger backlogs that correspond to delays of up to 50 ms.
%
%
\subsubsection{Non-stationary Service Curves}
\label{sec:servicecellular}
\begin{figure}
\hspace{-10pt}
\subfigure[0.95-Quantile of the backlog]{
\includegraphics[width=0.51\columnwidth]{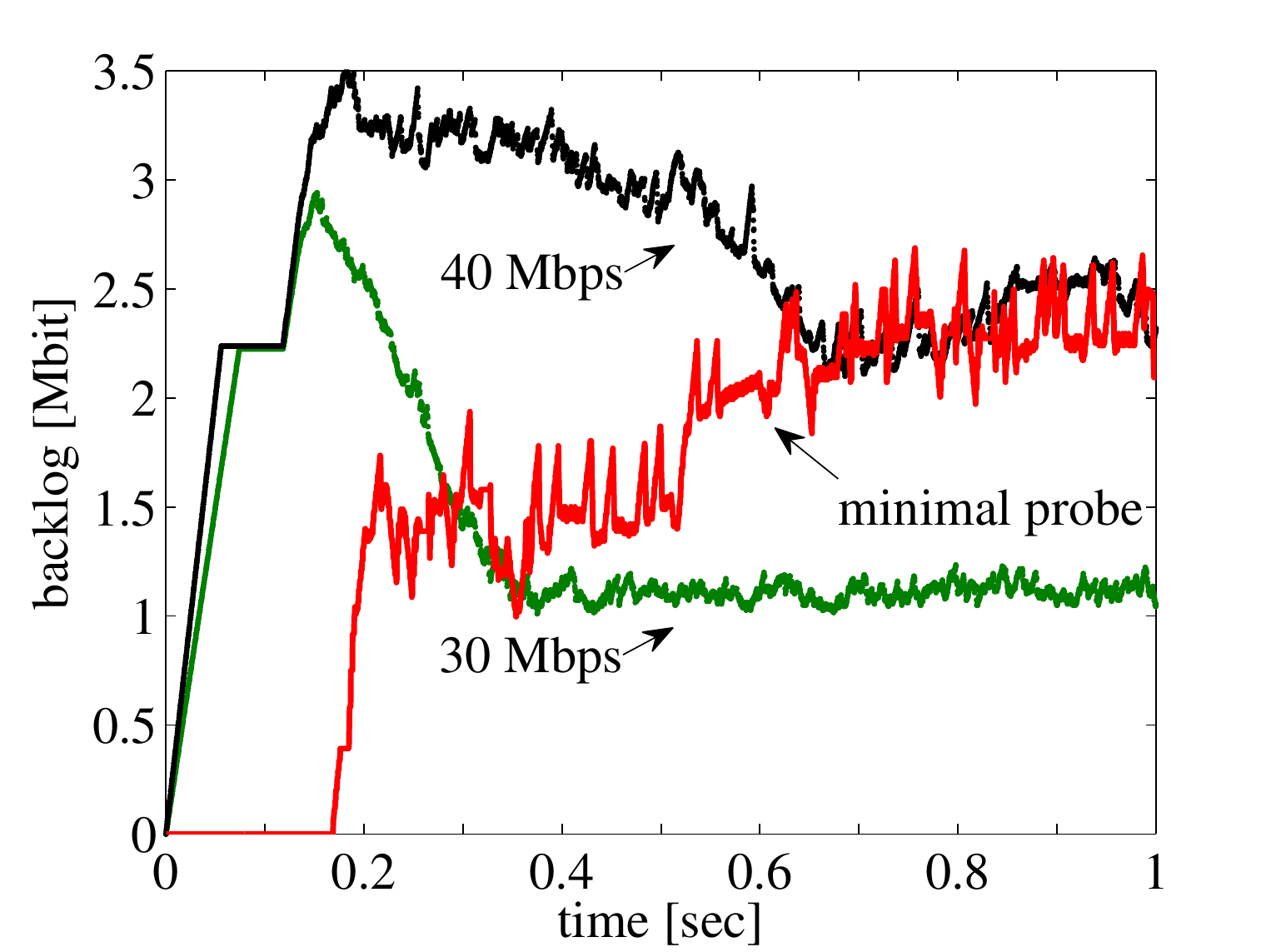}
\label{fig:minimalprobe_backlog_lte}
}
\hspace{-10pt}
\subfigure[Service curve estimates, $\varepsilon=0.05$]{
\includegraphics[width=0.51\columnwidth]{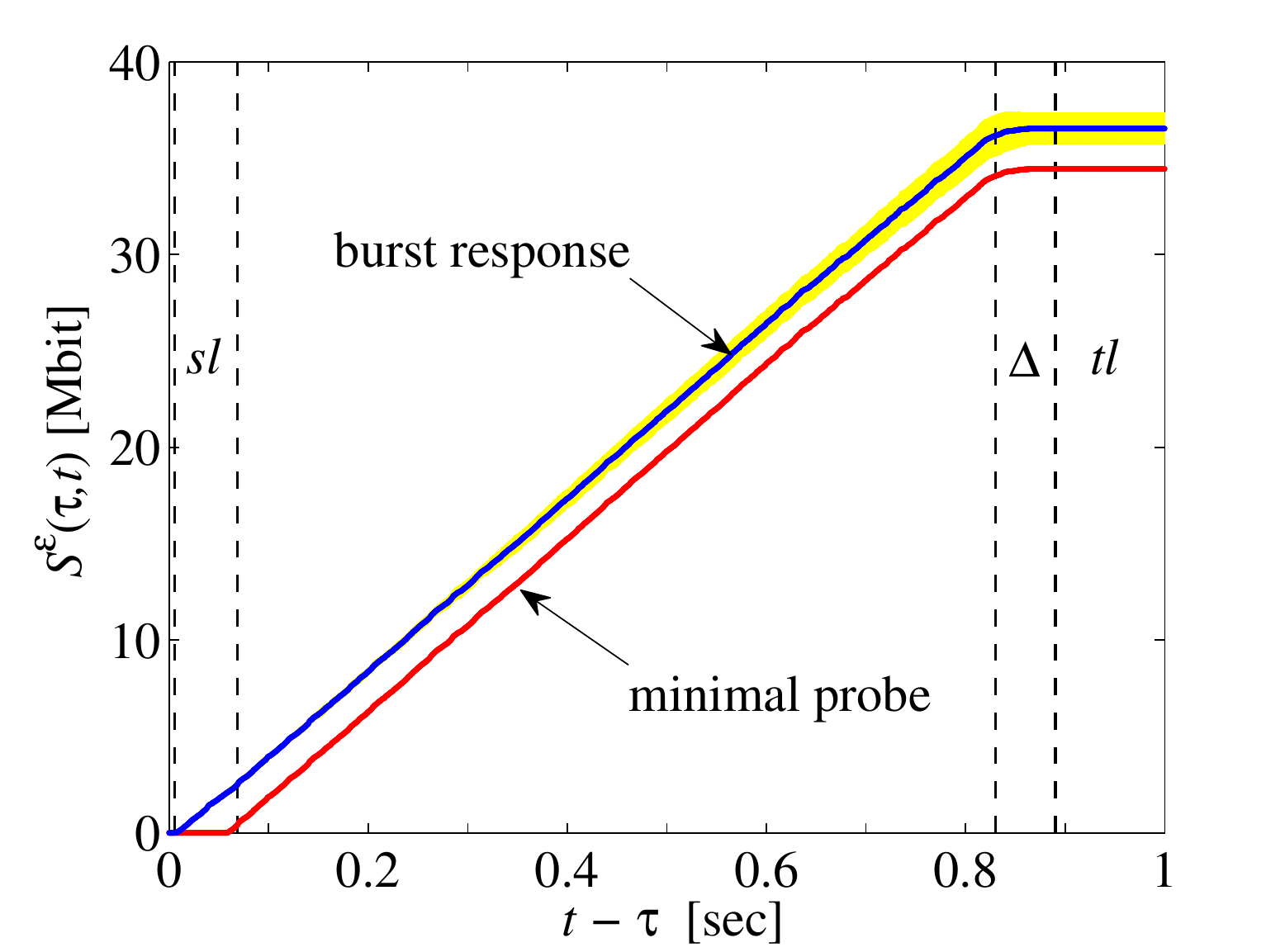}
\label{fig:servicecurveestimates_lte}
\hspace{-10pt}
}
\caption{Backlog of minimal probing and LTE service curve estimates.}
\label{fig:servicecurve_lte}
\end{figure}
Next, we identify the transient service of the LTE network by a non-stationary service curve, that is obtained by minimal probing as defined in Sec.~\ref{sec:minimalprobing}. Compared to the transient backlogs shown for individual CBR traffic arrivals in Fig.~\ref{fig:backlog_lte}, the service curve is a single characteristic function of the system that provides transient performance measures for any type of traffic arrivals.

The measurement method operates in two phases: first, an estimate of the minimal probe $\widetilde{A}_{mp}(\tau)$ is obtained from the burst response; secondly, $\widetilde{A}_{mp}(\tau)$ is used to estimate a service curve $\mathcal{S}^{\varepsilon}_{mp}(\tau,t)$. The accuracy of $\mathcal{S}^{\varepsilon}_{mp}(\tau,t)$ is determined as the backlog $B^{\varepsilon}(t)$ that is induced by $\widetilde{A}_{mp}(\tau)$. The accuracy can also be visualized using a service curve estimate $\mathcal{S}^{\varepsilon}_{br}(\tau,t)$ from the burst response, where $\mathcal{S}^{\varepsilon}_{br}(\tau,t) = \mathcal{S}^{\varepsilon}_{mp}(\tau,t) + B^{\varepsilon}(t)$ from Eq.~\eqref{eq:estimateaccuracy} serves as the upper bound.

In practice, we emulate the burst function $\delta(t)$ by sending packets at a high rate that exceeds the nominal capacity of 50 Mbps for a duration of $t$. We choose $t=1$ sec and measure 100 sample paths of the burst response to obtain $\mathcal{S}^{\varepsilon}_{br}(\tau,t)$ for $\varepsilon = 0.05$ from Eq.~\eqref{eq:effectiveserviceburstrespones}. An estimate of the minimal probe $\widetilde{A}_{mp}(\tau)$ follows from $\mathcal{S}^{\varepsilon}_{br}(\tau,t)$ by Eq.~\eqref{eq:minimalprobestochastic}. We add a single packet to the start of the probe that acts as a trigger to initiate the wake-up procedure of the UE. We make use of the minimal probe to measure the backlog $B(t)$ at the end of the probe. After collecting 100 backlog samples, we select the 0.95-quantile $B^\varepsilon(t)$. A service curve $\mathcal{S}^{\varepsilon}_{mp}(\tau,t)$ for $\varepsilon = 0.05$ is estimated from Eq.~\eqref{eq:estimateminimalprobing} by insertion of $\widetilde{A}_{mp}(\tau)$ and $B^\varepsilon(t)$.

The advantage of the minimal probe is that it is adapted to the system's transient service characteristics. As a consequence, minimal probing does not sent further packets during connection establishment, such that the transient overshoot and the following relaxation time are eliminated. The effect is shown in Fig.~\ref{fig:minimalprobe_backlog_lte}, where a representative sample path of the backlog of minimal probing is compared to the 0.95-backlog quantile of 30 Mbps and 40 Mbps CBR traffic, respectively. After the initial waiting time, the minimal probe has an average rate of 44 Mbps that results in a 0.95-backlog quantile of about 2.1~Mb at the end of the probe.

In Fig.~\ref{fig:servicecurveestimates_lte}, we present the mean of ten estimates of $\mathcal{S}^{\varepsilon}_{mp}(\tau,t)$ and $\mathcal{S}^{\varepsilon}_{br}(\tau,t)$ obtained by minimal probing and the burst response, respectively. We also include the 0.95-confidence interval of $\mathcal{S}^{\varepsilon}_{br}(\tau,t)$, depicted as a yellow area, that confirms stable estimates. For clarity, we omit the confidence interval of $\mathcal{S}^{\varepsilon}_{mp}(\tau,t)$, as it provides little additional information. Further, we note that the estimate $\mathcal{S}^{\varepsilon}_{mp}(\tau,t)$ shows a good accuracy, determined as the backlog $B^{\varepsilon}(t)$ of the minimal probe, that separates the upper bound $\mathcal{S}^{\varepsilon}_{br}(\tau,t)$ from the lower service estimate $\mathcal{S}^{\varepsilon}_{mp}(\tau,t)$. The deviation is a consequence of the super-additivity of the service, see Sec.~\ref{sec:accuracy}.

Notably, the service curve estimates show the same distinct features as illustrated in Figs.~\ref{fig:transientstationarylatencyrate} and~\ref{fig:sc_minimalprobing}, previously:
\begin{itemize}
\item {\bf Service outages:} For intervals $t-\tau \le 8$ ms, both service curve estimates $\mathcal{S}^{\varepsilon}_{mp}(\tau,t)$ and $\mathcal{S}^{\varepsilon}_{br}(\tau,t)$ are equal to zero, indicating service outages on short time-scales, e.g., due to the characteristics of the radio channel.
\item {\bf Stationary latency:} The region marked with {\it sl} expresses a stationary latency of about 50 ms. As in Fig.~\ref{fig:transientstationarylatencyrate}, $\mathcal{S}^{\varepsilon}_{mp}(\tau,t)$ identifies this region correctly, whereas $\mathcal{S}^{\varepsilon}_{br}(\tau,t)$ overestimates the service and attributes the stationary latency to the region marked with $\Delta$. The effect is due to the super-additivity of the service that is caused by the stationary latency, see Sec.~\ref{sec:superadditiveprocesses}.
\item {\bf Transient latency:} The region marked with {\it tl} shows a transient latency of about 120 ms that is due to sleep scheduling.
\item {\bf Capacity limit:} The upward segment at the center has a slope of 44 Mbps. It denotes the effective capacity limit with respect to $\varepsilon$. The almost constant slope evidences a stable transmission rate for intervals of $t-\tau \ge 58$ ms.
\end{itemize}
%
%
\subsection{Comparison with HSPA and EDGE}
\label{sec:comparisonEDGEHSPA}
We compare the performance of the LTE DRX mode with the preceding technologies HSPA and EDGE. For evaluation, we estimate non-stationary service curves using the measurement method as in Sec.~\ref{sec:servicecellular}. Taking into account smaller capacities and higher latencies, we reduce the probe traffic rate and extend the measurement duration accordingly.

\begin{figure*}
\begin{center}
\subfigure[LTE]{
\includegraphics[width=0.51\columnwidth]{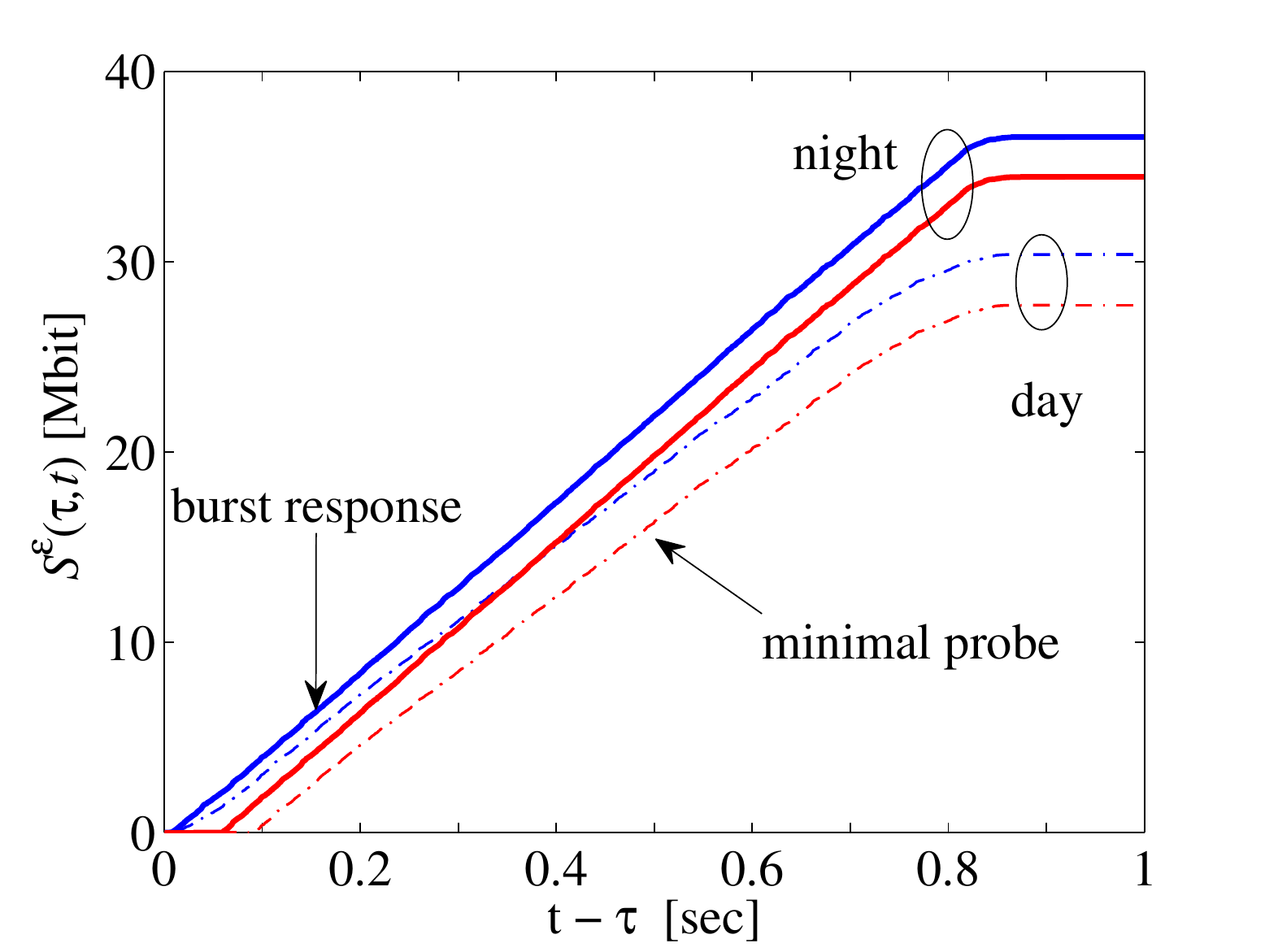}
\label{fig:lte}
}
\subfigure[HSPA]{
\includegraphics[width=0.51\columnwidth]{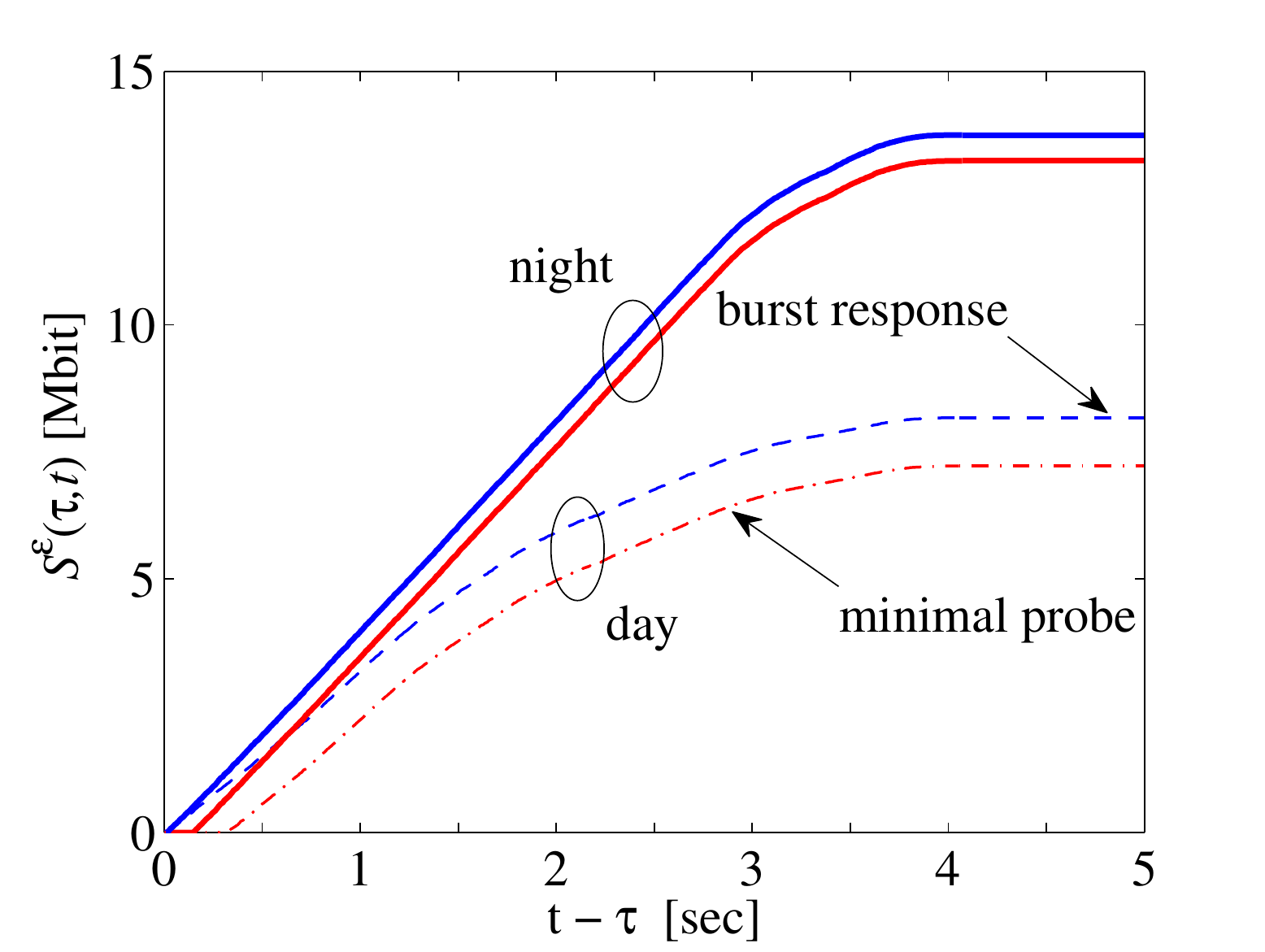}
\label{fig:hspa}
}
\subfigure[EDGE]{
\includegraphics[width=0.51\columnwidth]{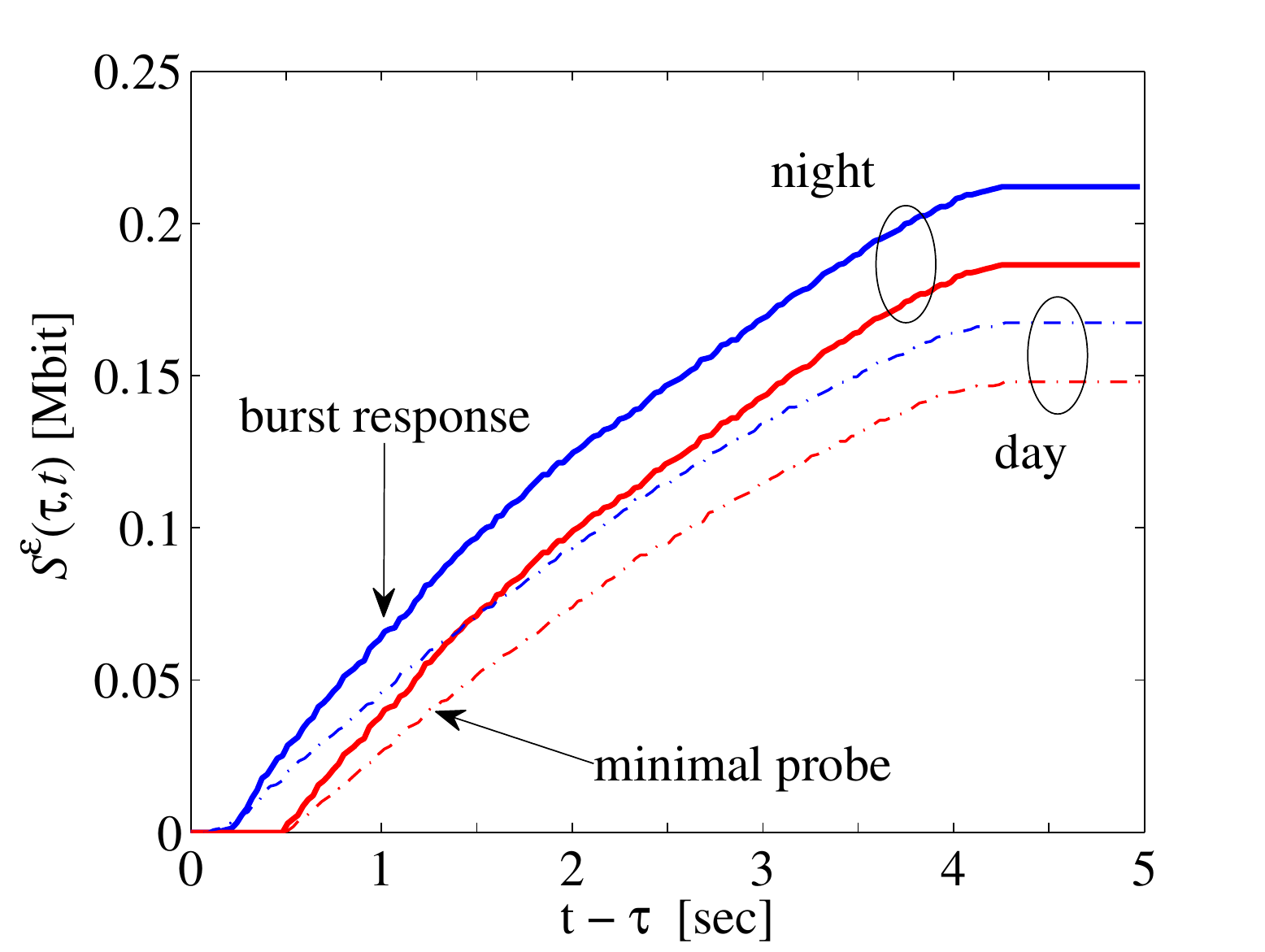}
\label{fig:edge}
}
\caption{Service curve estimates of LTE, HSPA, and EDGE. Solid lines show estimates obtained during the night, and dashed lines during the day, respectively.}
\label{fig:technologyevolution}
\end{center}
\end{figure*}
Fig.~\ref{fig:technologyevolution} evaluates service curve estimates of LTE, HSPA, and EDGE that are obtained for each technology during the day and during the night respectively. First, we consider only the solid curves that apply for the night.

Compared to LTE, the service curves of HSPA and EDGE exhibit the same shape, however, with less pronounced segments, particularly for EDGE. Further, the estimates $\mathcal{S}^{\varepsilon}_{mp}(\tau,t)$ of minimal probing and $\mathcal{S}^{\varepsilon}_{br}(\tau,t)$ of burst probing show a larger relative difference for EDGE, indicating a lower accuracy. For HSPA the accuracy is good, like in case of LTE.

Next, we pay closer attention to the service curve estimates of HSPA in Fig.~\ref{fig:hspa}. Following the same reasoning as in Sec.~\ref{sec:servicecellular}, we find service outages on short time-scales that result in a service of zero for intervals $t-\tau \le 15$~ms, a stationary latency of 130 ms, a transient latency of approximately 1 sec, and a capacity limit of close to 4 Mbps. Further, we notice that after the transient latency, the capacity limit is approached more slowly than in case of LTE, expressed by the bent segment in the region $3 \le t-\tau \le 3.9$~sec. The segment emphasizes once more the advantage of minimal probing that features a correspondingly reduced rate during this time.

In case of EDGE, the service curve estimates show an even stronger bend than in case of HSPA that reveals less explicit parameters. The estimates indicate service outages of up to 200 ms, a stationary delay of 300 ms, a transient delay of 500 ms, and an effective capacity limit of about 70 kbps.

\begin{figure}[b]
\hspace{-10pt}
\subfigure[HSPA]{
\includegraphics[width=0.51\columnwidth]{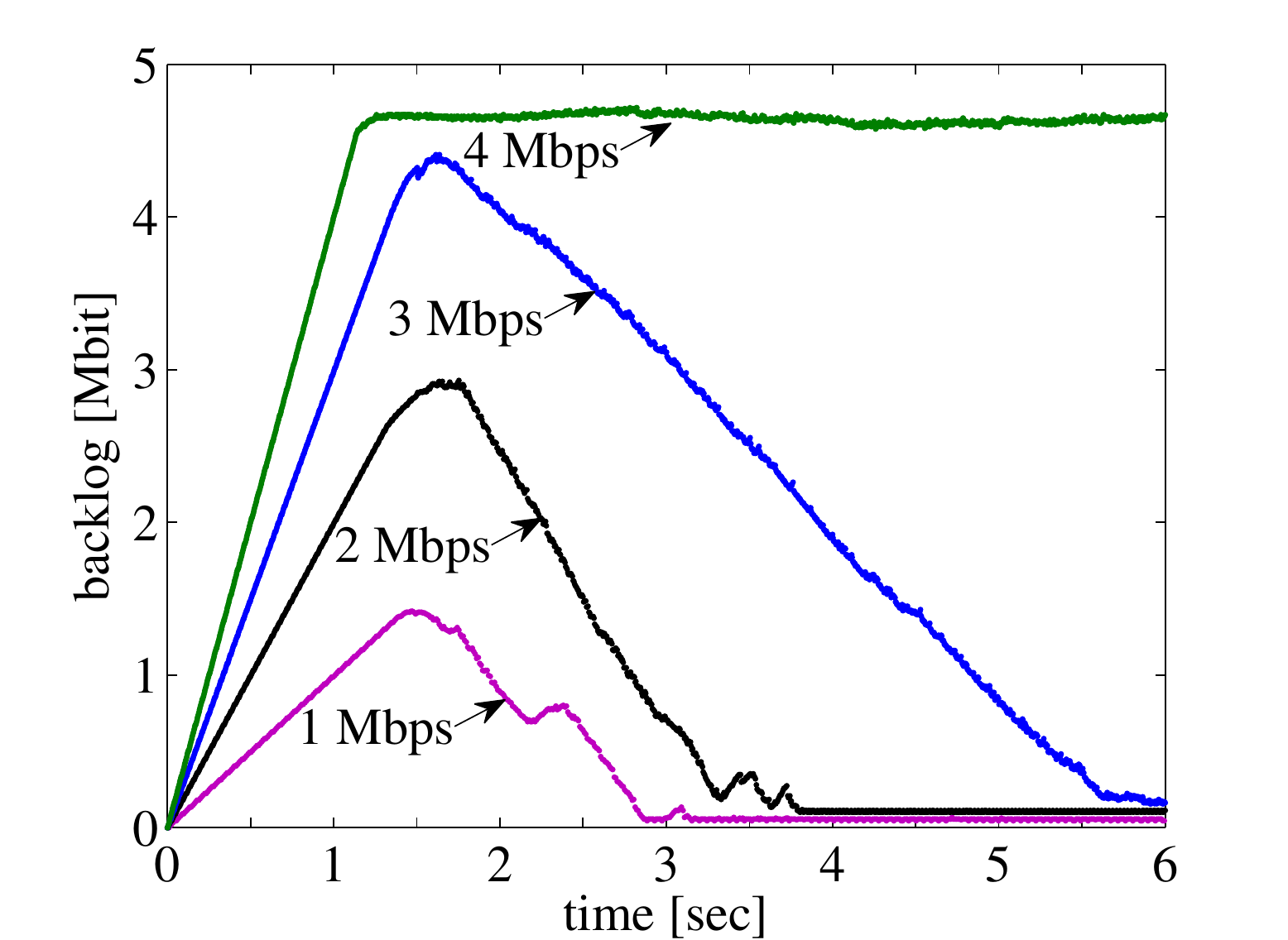}
\label{fig:hspa_quantile}
}
\hspace{-10pt}
\subfigure[EDGE]{
\includegraphics[width=0.51\columnwidth]{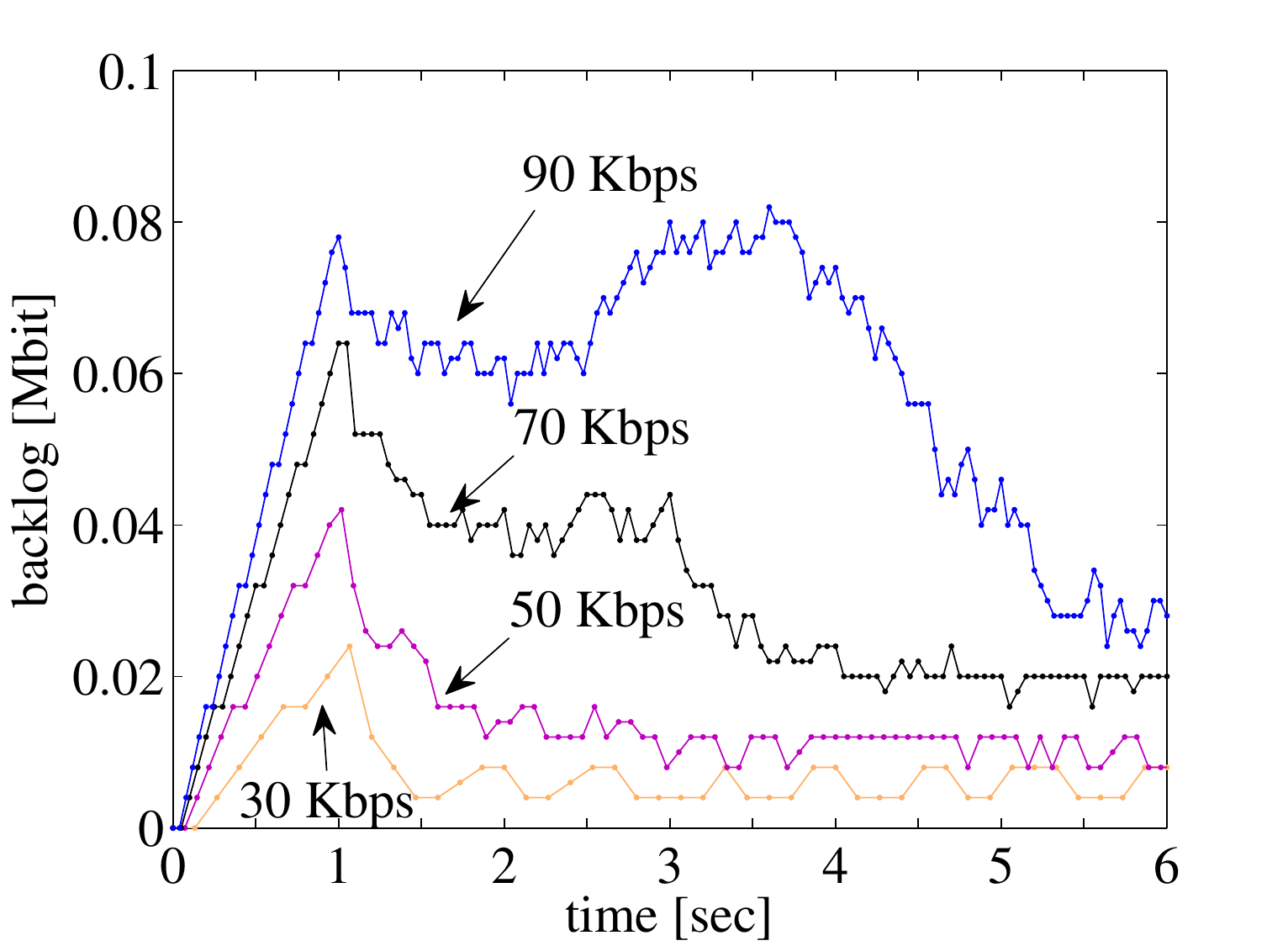}
\label{fig:gprs_quantile}
\hspace{-10pt}
}
\caption{Transient 0.95-backlog quantiles of HSPA and EDGE for CBR traffic.}
\label{fig:backlog_hspa_edge}
\end{figure}
Corresponding 0.95-quantiles of the backlog for CBR traffic with different rates are shown in Fig.~\ref{fig:backlog_hspa_edge}. Clearly, the long transient delay causes a backlog overshoot that can be an order of magnitude larger than the stationary backlog. Compared to LTE, the backlog of HSPA shows a less sharp peak that has a rounded top in the region around 1.5 sec. The effect is due to the slow ramp-up of the transmission rate after connection establishment. The backlog is cleared after a relaxation time of, e.g., about 5 sec for a traffic rate of 3 Mbps.
%
%
\subsection{Diurnal Characteristics}
\label{sec:diurnalcharacteristics}
We also include service curve estimates that are obtained during the day in Fig.~\ref{fig:technologyevolution}, plotted as dashed lines. Compared to the night, all daytime measurements show the same general trend, i.e., a reduction of the service, that is most pronounced for HSPA. The effect may be attributed to the utilization of the network by other users that is known to follow a diurnal pattern. Further, we notice a greater volatility of the estimates at daytime, except for EDGE, that is reflected by a reduced accuracy, i.e., a larger deviation $B^{\varepsilon}(t)$ of the lower estimate $\mathcal{S}^{\varepsilon}_{mp}(\tau,t)$ from the upper bound $\mathcal{S}^{\varepsilon}_{br}(\tau,t)$. Notably, the transient plus stationary latency, expressed by the flat area in the upper right of the curves, is, however, only marginally affected by the time of measurement.

For an example, we consider the case of LTE where the effective capacity limit is reduced from 44 Mbps at night to 38 Mbps during the day while $B^{\varepsilon}(t)$ increases from 2.1~Mb to 2.7~Mb. An analysis of the backlog for selected traffic rates, see Fig. \ref{fig:backlog_lte_day_night}, confirms the observations. While there is little difference between day and night for a traffic rate of 30 Mbps, we notice a significant increase of the backlog for 40 Mbps.
\begin{figure}[b]
\hspace{-10pt}
\subfigure[Backlog 30 Mbps]{
\includegraphics[width=0.51\columnwidth]{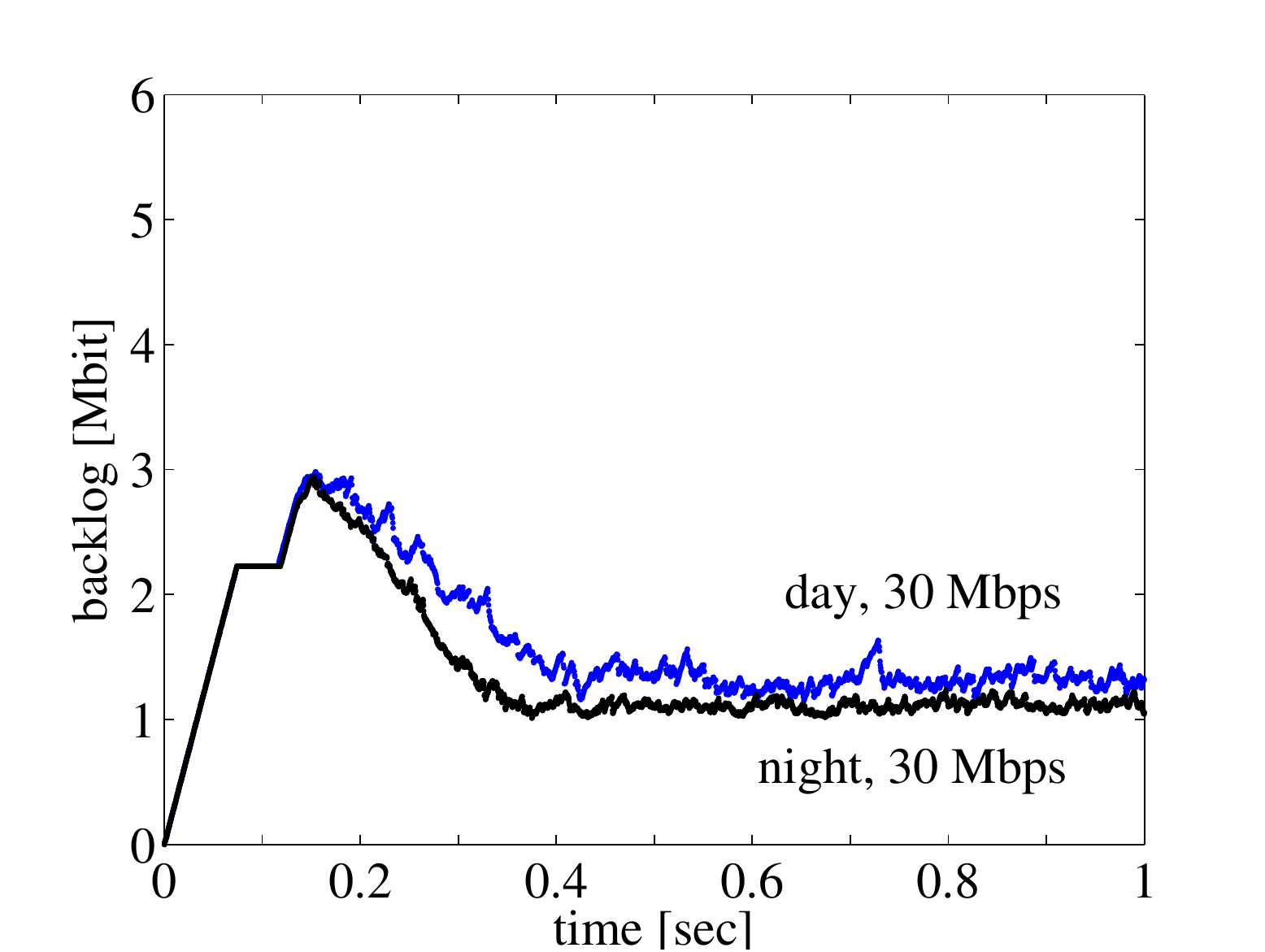}
\label{fig:lte_day_night_30mbps}
}
\hspace{-10pt}
\subfigure[Backlog 40 Mbps]{
\includegraphics[width=0.51\columnwidth]{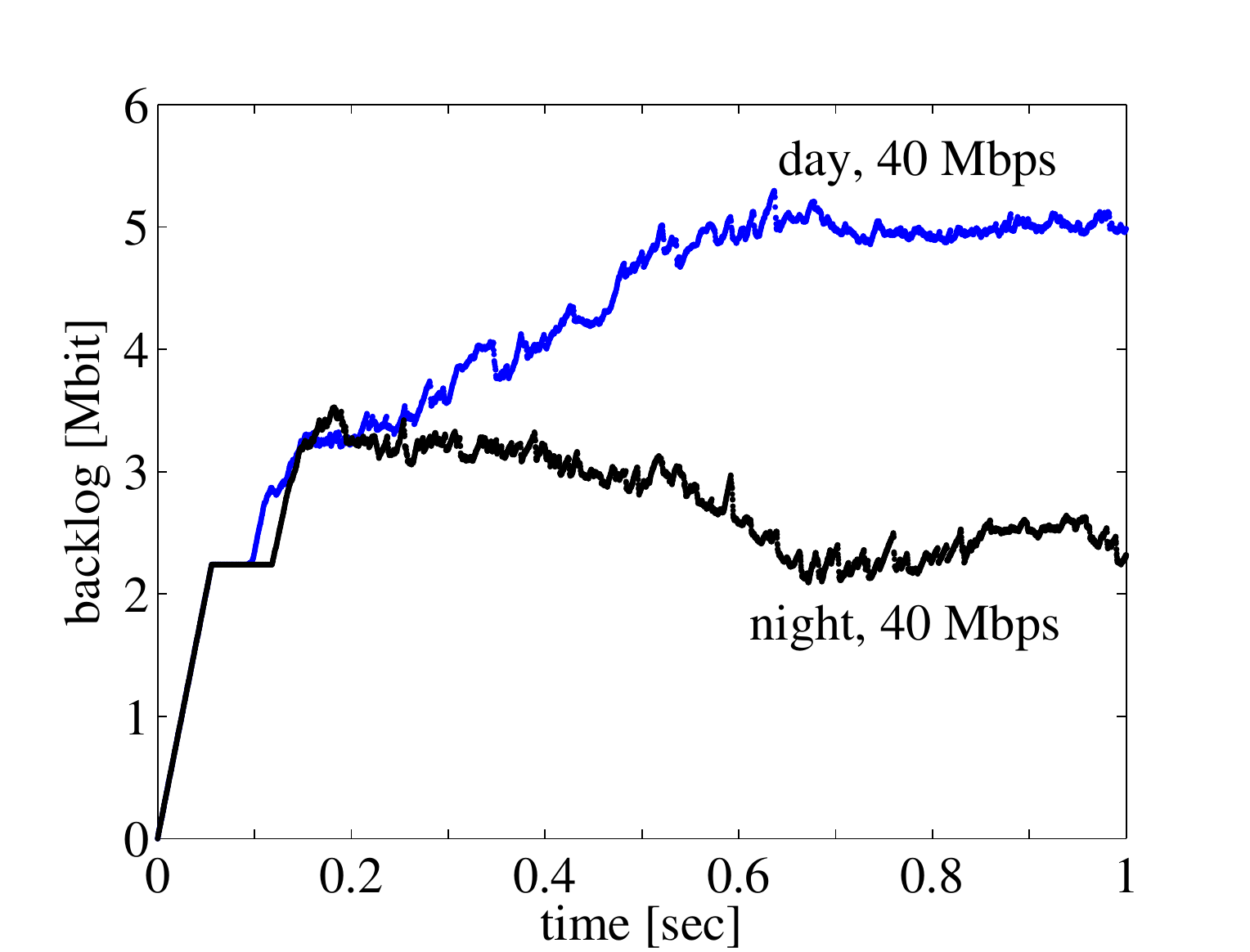}
\label{fig:lte_day_night_40mbps}
\hspace{-10pt}
}
\caption{Transient 0.95-backlog quantiles of LTE during day and night.}
\label{fig:backlog_lte_day_night}
\end{figure}
%
%
\section{Conclusions}
\label{sec:conclusion}
This work contributed a notion of non-stationary service curves that enables the analysis of transient phases. We contributed models of systems with sleep scheduling and provided insights into their transient behavior. Further, we considered measurement-based methods for identification of systems, using a black-box model. We discovered that existing probing methods cannot accurately estimate non-stationary service curves. The difficulties were related to the non-convexity and super-additivity of the service. We devised a novel two-phase probing method. First a minimal probe is estimated that is adapted to the network. In a second step, the minimal probe is used to obtain a service curve estimate with a defined accuracy. Taking advantage of the estimation method, we reported results from a comprehensive measurement study of cellular networks with sleep scheduling, including EDGE, HSPA, and LTE. The service curve estimates showed characteristic features of the cellular data service that explained the observation of significant transient overshoots and long relaxation times. While delays in the range of seconds have been measured for EDGE and HSPA, it has been found that LTE achieves an improvement by an order of magnitude. We believe that the general non-stationary service curve approach and the modelling and measurement-based identification methods lay a foundation for transient analysis that has the potential to provide insights into a variety of relevant other systems.
%
%
\section*{Appendix}
\subsubsection{Derivation of Eq.~\eqref{eq:statisticalenvelope}}
We consider the parameterized envelope
\begin{equation*}
\mathcal{A}^{\varepsilon}(t) = \rho_A t + \sigma_A .
\end{equation*}
with rate parameter $\rho_A > 0$ and burstiness parameter $\sigma_A \ge 0$. To obtain $\rho_A$ and $\sigma_A$, we insert $\mathcal{A}^{\varepsilon}(t)$ into Eq.~\eqref{eq:statisticalenvelopedefinition} and derive
\begin{align*}
1 - &\mathsf{P}[A(\tau,t) \le \rho_A (t-\tau) + \sigma_A,\, \forall \tau \in [0,t] ] \\
=&\mathsf{P}\biggl[ \max_{\tau \in [0,t-1]} \{ A(\tau,t) - \rho_A (t-\tau) \} > \sigma_A \biggr] \\
=&\mathsf{P}\biggl[ \max_{\tau \in [1,t]} \bigl\{ e^{\theta (A(t-\tau,t) - \rho_A \tau)} \bigr\} > e^{\theta\sigma_A} \biggr] ,
\end{align*}
for $\theta > 0$, where we used that $A(t,t)=0$. Next, we fix $t > 0$ and consider the process $U(\tau) = e^{\theta (A(t-\tau,t) - \rho_A \tau)}$, for $\tau \in [0,t]$. It follows that $U(\tau+1) = U(\tau) e^{\theta (A(t-\tau-1,t-\tau) - \rho_A)}$ and using the independence the conditional expectation becomes
\begin{align*}
&\mathsf{E}[U(\tau+1)|U(\tau),U(\tau-1),\dots, U(1)] \\
=& U(\tau) \mathsf{E}[e^{\theta(A(t-\tau-1,t-\tau)}] e^{-\theta \rho_A} .
\end{align*}
Since for iid increments $\mathsf{E}[e^{\theta(A(t-\tau-1,t-\tau)}] = \mathsf{M}_A(\theta,1)$, it follows that $\mathsf{E}[U(\tau+1)|U(\tau),U(\tau-1),\dots U(1)] = U(\tau)$ is a martingale if $e^{\theta \rho_A} = \mathsf{M}_A(\theta,1)$. Hence, we obtain parameter $\rho_A = \ln \mathsf{M}_A(\theta,1)/\theta$ so that $\rho_A t = \ln \mathsf{M}_A(\theta,t)/\theta$. By application of Doob's martingale inequality~\cite[Theorem 3.2, p. 314]{doob:stochasticprocesses} and the reformulation from~\cite{jiang:noteonsnetcalc} we have for non-negative martingales $U(\tau)$ for $\tau \ge 1$ that
\begin{equation*}
x \mathsf{P} \left[\max_{\tau \in [1,t]} \{ U(\tau \} \ge x \right] \le \mathsf{E}[U(1)] .
\end{equation*}
With $\mathsf{E}[U(1)] = 1$ and $x = e^{\theta \sigma_A}$ it follows that
\begin{equation*}
\mathsf{P} \left[\max_{\tau \in [0,t]} \{ A(\tau,t) - \rho_A(\theta) (t-\tau)\} > \sigma_A \right] \le e^{-\theta \sigma_A} .
\end{equation*}
Finally, we let $\varepsilon = e^{-\theta \sigma_A}$ to obtain $\sigma_A = -\ln\varepsilon/\theta$.
\subsubsection{Derivation of Eq.~\eqref{eq:statisticalserviceenvelope}}
The derivation of Eq.~\eqref{eq:statisticalserviceenvelope} extends~\cite{Fidler:2014:CDE} to non-stationary processes. It employs basic steps from the stochastic network calculus~\cite{li:effectivebandwidthcalculus2, ciucu:networkservicecurvescaling2}. We use the complementary formulation of Eq.~\eqref{eq:bivariateenvelope} to define
\begin{equation*}
\xi := \mathsf{P}[\exists \tau \in [0,t]:  S(\tau,t) < \mathcal{S}^{\varepsilon}(\tau,t)] \le \varepsilon .
\end{equation*}
We prove that $\mathcal{S}^{\varepsilon}(\tau,t)$ defined in Eq.~\eqref{eq:statisticalserviceenvelope} satisfies Eq.~\eqref{eq:bivariateenvelope} by showing that $\xi \le \varepsilon$. Using the union bound and Chernoff's lower bound $\mathsf{P}[X \le x] \le e^{\theta x} \mathsf{M}_X(-\theta)$ for $\theta \ge 0$ it holds that
\begin{equation*}
\xi \le \! \sum_{\tau=0}^{t-1} \mathsf{P}[S(\tau,t) < \mathcal{S}^{\varepsilon}(\tau,t)] \le \! \sum_{\tau=0}^{t-1} e^{\theta(\tau,t) \mathcal{S}^{\varepsilon}(\tau,t)} \mathsf{M}_S(-\theta,\tau,t) ,
\end{equation*}
where $\theta(\tau,t) \ge 0$ is a set of free parameters. The case where $\tau=t$ is omitted since $S(t,t) = 0$ and $\mathcal{S}^{\varepsilon}(t,t) \le 0$ by definition. By insertion of $\mathcal{S}^{\varepsilon}(\tau,t)$ from Eq.~\eqref{eq:statisticalserviceenvelope} it follows that
\begin{equation*}
\xi \le \rho\varepsilon \! \sum_{\tau=0}^{t-1} \! e^{-\rho(t-\tau)} = \rho\varepsilon \! \sum_{\upsilon=1}^{t} \! e^{-\rho\upsilon} \le \rho\varepsilon \! \int_{0}^{\infty} \!\!\! e^{-\rho y} dy = \varepsilon
\end{equation*}
where each summand is bounded by $e^{-\rho\upsilon} \le \int_{\upsilon-1}^{\upsilon} e^{-\rho y} dy$ since $e^{-\rho\upsilon}$ is decreasing. Finally, letting $t \rightarrow \infty$ and solving the integral completes the proof that $\xi \le \varepsilon$ for all $t \ge 0$.
\subsubsection{Proof of Lem.~\ref{lem:additivity}}
By definition of $h$ we have
\begin{align}
&h(s,t) + h(t,u) = f \otimes g (s,t) + f \otimes g (t,u) \nonumber \\
&= \inf_{\tau \in [s,t]} \inf_{\upsilon \in [t,u]} \{ f(s,\tau) + f(t,\upsilon) + g(\tau,t) + g(\upsilon,u) \}.
\label{eq:doubleinfadditivity}
\end{align}

i) Given $f$ and $g$ are super-additive. From Eq.~\eqref{eq:doubleinfadditivity} we have
\begin{align}
&h(s,t) + h(t,u) \nonumber \\
&\le \inf_{\tau \in [s,t]} \inf_{\upsilon \in [t,u]} \{ f(s,\upsilon) - f(\tau,t) + g(\tau,t) + g(\upsilon,u) \} \nonumber \\
&= \inf_{\upsilon \in [t,u]} \{ f(s,\upsilon) + g(\upsilon,u) \} + \inf_{\tau \in [s,t]} \{ g(\tau,t) - f(\tau,t) \} \nonumber \\
&\le \inf_{\upsilon \in [t,u]} \{ f(s,\upsilon) + g(\upsilon,u) \}.
\label{eq:subadditivitypart1}
\end{align}
In the first line, we estimated $f(s,\tau) + f(\tau,t) + f(t,\upsilon) \le f(s,\upsilon)$ due to the super-additivity of $f$. In the second line, we rearranged the infima, and in the third line, we estimated $\inf_{\tau \in [s,t]} \{ g(\tau,t) - f(\tau,t) \} \le g(t,t) - f(t,t) = 0$ since $f(t,t),g(t,t)=0$ for all $t \ge 0$. Similarly, using the super-additivity of $g$ we derive from Eq.~\eqref{eq:doubleinfadditivity} that
\begin{equation}
h(s,t) + h(t,u) \le \inf_{\tau \in [s,t]} \{ f(s,\tau) + g(\tau,u) \} .
\label{eq:subadditivitypart2}
\end{equation}
Combining Eq.~\eqref{eq:subadditivitypart1} and Eq.~\eqref{eq:subadditivitypart2} we obtain
\begin{equation*}
h(s,t) + h(t,u) \le \inf_{\tau \in [s,u]} \{ f(s,\tau) + g(\tau,u) \} = h(s,u) ,
\end{equation*}
which proves the super-additivity of $h$.

ii) For the special case of additive univariate functions $f(s,t) = f(t-s)$ and $g(s,t) = g(t-s)$, that depend only on the difference $t-s$ and not on the absolute values of $s$ and $t$, it follows that $h(s,t) = f \otimes g(t-s) = h(t-s)$ is also univariate. Using the additivity of $f$ and $g$, Eq.~\eqref{eq:doubleinfadditivity} yields that
\begin{align*}
& h(t-s) + h(u-t) \\
&= \inf_{\tau \in [s,t]} \inf_{\upsilon \in [t,u]} \{ f(\tau-s+\upsilon-t) + g(t-\tau+u-\upsilon) \} \\
&= \inf_{\varsigma \in [s+t,t+u]} \{ f(\varsigma-s-t) + g(t+u-\varsigma) \} \\
&= \inf_{\varsigma \in [0,u-s]} \{ f(\varsigma) + g(u-s-\varsigma)\} = h(u-s)
\end{align*}
which proves the additivity of $h$.\qed
\subsubsection{Proof of Lem.~\ref{lem:additivitymin}}
By definition of $h$ we have
\begin{align*}
&h(s,t) + h(t,u) \\
&= \min \{f(s,t), g(s,t)\} + \min \{f(t,u), g(t,u)\} \\
&\le \min \{f(s,u), g(s,u), f(s,t)+g(t,u), g(s,t)+f(t,u)\} \\
&\le h(s,u).
\end{align*}
In the second line we used the super-additivity of $f$ and $g$ and in the third line $\min\{f(s,u),g(s,u)\}=h(s,u)$ and $\min\{h(s,u),x\} \le h(s,u)$ for any $x$.\qed
\pagebreak
\bibliographystyle{IEEEtran}
\bibliography{IEEEabrv,IEEEfidler}

\begin{thebibliography}{10}
\providecommand{\url}[1]{#1}
\csname url@samestyle\endcsname
\providecommand{\newblock}{\relax}
\providecommand{\bibinfo}[2]{#2}
\providecommand{\BIBentrySTDinterwordspacing}{\spaceskip=0pt\relax}
\providecommand{\BIBentryALTinterwordstretchfactor}{4}
\providecommand{\BIBentryALTinterwordspacing}{\spaceskip=\fontdimen2\font plus
\BIBentryALTinterwordstretchfactor\fontdimen3\font minus
  \fontdimen4\font\relax}
\providecommand{\BIBforeignlanguage}[2]{{%
\expandafter\ifx\csname l@#1\endcsname\relax
\typeout{** WARNING: IEEEtran.bst: No hyphenation pattern has been}%
\typeout{** loaded for the language `#1'. Using the pattern for}%
\typeout{** the default language instead.}%
\else
\language=\csname l@#1\endcsname
\fi
#2}}
\providecommand{\BIBdecl}{\relax}
\BIBdecl

\bibitem{becker:tsc}
N.~Becker and M.~Fidler, ``A non-stationary service curve model for performance
  analysis of transient phases,'' in \emph{Proc. of {ITC 27}}, Sep. 2015.

\bibitem{mellia:shortlivedtcp}
M.~Mellia, I.~Stoica, and H.~Zhang, ``{TCP} model for short lived flows,''
  \emph{{IEEE} Commun. Lett.}, vol.~6, no.~2, pp. 85--87, Feb. 2002.

\bibitem{3gpp:MAC:rel8}
{3GPP specification TS 36.321}, ``{Evolved Universal Terrestrial Radio Access
  (E-UTRA); Medium Access Control (MAC) protocol specification},'' Mar. 2012,
  release 8, version 8.12.

\bibitem{perrucci:impact2g3g}
G.~P. Perrucci, F.~H. Fitzek, G.~Sasso, W.~Kellerer, and J.~Widmer, ``On the
  impact of {2G} and {3G} network usage for mobile phones' battery life,'' in
  \emph{Proc. of European Wireless}, 2009, pp. 255--259.

\bibitem{huang:closeexamination}
J.~Huang, F.~Qian, A.~Gerber, Z.~M. Mao, S.~Sen, and O.~Spatscheck, ``A close
  examination of performance and power characteristics of {4G} {LTE}
  networks,'' in \emph{Proc. of {ACM} {MobiSys}}, 2012, pp. 225--238.

\bibitem{becker:lte}
N.~Becker, A.~Rizk, and M.~Fidler, ``A measurement study on the
  application-level performance of {LTE},'' in \emph{Proc. of {IFIP}
  {Networking}}, Jun. 2014.

\bibitem{ra:energydelaytradeoff}
M.-R. Ra, J.~Paek, A.~B. Sharma, R.~Govindan, M.~H. Krieger, and M.~J. Neely,
  ``Energy-delay tradeoffs in smartphone applications,'' in \emph{Proc. of
  {ACM} {MobiSys}}, 2010, pp. 255--270.

\bibitem{yang:dynamicpowersaving3g}
S.-R. Yang, ``Dynamic power saving mechanism for {3G} {UMTS} system,''
  \emph{Mobile Networks and Applications}, vol.~12, no.~1, pp. 5--14, 2007.

\bibitem{yang:modelingumtspowersaving}
S.-R. Yang, S.-Y. Yan, and H.-N. Hung, ``Modeling {UMTS} power saving with
  bursty packet data traffic,'' \emph{{IEEE} Trans. Mobile Comput.}, vol.~6,
  no.~12, pp. 1398--1409, Dec. 2007.

\bibitem{zhou:performanceltedrx}
L.~Zhou, H.~Xu, H.~Tian, Y.~Gao, L.~Du, and L.~Chen, ``Performance analysis of
  power saving mechanism with adjustable {DRX} cycles in {3GPP LTE},'' in
  \emph{Proc. of IEEE VTC}, Sep. 2008.

\bibitem{bhamber:analyticLTEpowersavingburstytraffic}
R.~S. Bhamber, S.~Fowler, C.~Braimiotis, and A.~Mellouk, ``Analytic analysis of
  {LTE/LTE}-advanced power saving and delay with bursty traffic,'' in
  \emph{Proc. of {IEEE} {ICC}}, 2013, pp. 2964--2968.

\bibitem{wu:performancedrx}
J.~Wu, T.~Zhang, and Z.~Zeng, ``Performance analysis of discontinuous reception
  mechanism with web traffic in {LTE} networks,'' in \emph{Proc. of {IEEE}
  {PIMRC}}, 2013, pp. 1676--1681.

\bibitem{zhou:ltedrxm2m}
K.~Zhou, N.~Nikaein, and T.~Spyropoulos, ``{LTE/LTE-A} discontinuous reception
  modeling for machine type communications,'' \emph{IEEE Wireless Commun.
  Lett.}, vol.~2, no.~1, pp. 102--105, 2013.

\bibitem{ross:probabilitymodels}
S.~M. Ross, \emph{Introduction to Probability Models}.\hskip 1em plus 0.5em
  minus 0.4em\relax Academic Press, 2007.

\bibitem{zhang:transientmm1}
J.~Zhang and E.~J. Coyle, ``The transient solution of time-dependent {M/M/1}
  queues,'' \emph{{IEEE} Trans. Inf. Theory}, vol.~37, no.~6, pp. 1690--1696,
  Nov. 1991.

\bibitem{wang:transientatm}
C.-Y. Wang, D.~Logothetis, K.~S. Trivedi, and I.~Viniotis, ``Transient behavior
  of {ATM} networks under overloads,'' in \emph{Proc. of {IEEE} {INFOCOM}},
  Mar. 1996, pp. 978--985.

\bibitem{souza1998algorithm}
E.~{de Souza e Silva} and H.~R. Gail, ``An algorithm to calculate transient
  distributions of cumulative rate and impulse based reward,'' \emph{Stochastic
  models}, vol.~14, no.~3, pp. 509--536, 1998.

\bibitem{horvath2012transient}
A.~Horv{\'a}th, M.~Paolieri, L.~Ridi, and E.~Vicario, ``Transient analysis of
  non-{M}arkovian models using stochastic state classes,'' \emph{Performance
  Evaluation}, vol.~69, no.~7, pp. 315--335, 2012.

\bibitem{cruz:networkdelaycalculus}
R.~L. Cruz, ``A calculus for network delay, part {I} and {II}: Network elements
  in isolation and network analysis,'' \emph{{IEEE} Trans. Inf. Theory},
  vol.~37, no.~1, pp. 114--141, Jan. 1991.

\bibitem{leboudec:networkcalculus}
J.-Y. {Le Boudec} and P.~Thiran, \emph{Network Calculus A Theory of
  Deterministic Queuing Systems for the {I}nternet}.\hskip 1em plus 0.5em minus
  0.4em\relax Springer-Verlag, 2001.

\bibitem{chang:performanceguarantees}
C.-S. Chang, \emph{Performance Guarantees in Communication Networks}.\hskip 1em
  plus 0.5em minus 0.4em\relax Springer-Verlag, 2000.

\bibitem{burchard:endtoendstatisticalcalculus}
A.~Burchard, J.~Liebeherr, and S.~Patek, ``A min-plus calculus for end-to-end
  statistical service guarantees,'' \emph{{IEEE} Trans. Inf. Theory}, vol.~52,
  no.~9, pp. 4105--4114, Aug. 2006.

\bibitem{li:effectivebandwidthcalculus2}
C.~Li, A.~Burchard, and J.~Liebeherr, ``A network calculus with effective
  bandwidth,'' \emph{{IEEE/ACM} Trans. Netw.}, vol.~15, no.~6, pp. 1442--1453,
  Dec. 2007.

\bibitem{ciucu:networkservicecurvescaling2}
F.~Ciucu, A.~Burchard, and J.~Liebeherr, ``Scaling properties of statistical
  end-to-end bounds in the network calculus,'' \emph{{IEEE/ACM} Trans. Netw.},
  vol.~14, no.~6, pp. 2300--2312, Jun. 2006.

\bibitem{fidler:momentcalculus}
M.~Fidler, ``An end-to-end probabilistic network calculus with moment
  generating functions,'' in \emph{Proc. of {IWQoS}}, Jun. 2006, pp. 261--270.

\bibitem{jiang:stochasticnetworkcalculus}
Y.~Jiang and Y.~Liu, \emph{Stochastic Network Calculus}.\hskip 1em plus 0.5em
  minus 0.4em\relax Springer-Verlag, Sep. 2008.

\bibitem{fidler:netcalcsurvey}
M.~Fidler, ``A survey of deterministic and stochastic service curve models in
  the network calculus,'' \emph{{IEEE} Commun. Surveys Tuts.}, vol.~12, no.~1,
  pp. 59--86, 2010.

\bibitem{ciucu:goodvalue}
F.~Ciucu and J.~Schmitt, ``Perspectives on network calculus - no free lunch but
  still good value,'' in \emph{Proc. of {ACM} {SIGCOMM}}, Aug. 2012, pp.
  311--322.

\bibitem{fidler:netcalcguide}
M.~Fidler and A.~Rizk, ``A guide to the stochastic network calculus,''
  \emph{{IEEE} Commun. Surveys Tuts.}, vol.~17, no.~1, pp. 92--105, Mar. 2015.

\bibitem{chang:dynamicserviceguarantees}
C.-S. Chang, R.~L. Cruz, J.-Y. {Le Boudec}, and P.~Thiran, ``A min, + system
  theory for constrained traffic regulation and dynamic service guarantees,''
  \emph{{IEEE/ACM} Trans. Netw.}, vol.~10, no.~6, pp. 805--817, 2002.

\bibitem{agrawal:timevaryingservice}
R.~Agrawal, F.~Baccelli, and R.~Rajan, ``An algebra for queueing networks with
  time-varying service and its application to the analysis of {I}ntegrated
  {S}ervice networks,'' \emph{Mathematics of Operation Research}, vol.~29,
  no.~3, pp. 559--591, Aug. 2004.

\bibitem{chang:timevaryingfiltering}
C.-S. Chang and R.~L. Cruz, ``A time varying filtering theory for constrained
  traffic regulation and dynamic service guarantees,'' in \emph{Proc. of {IEEE}
  {INFOCOM}}, Mar. 1999, pp. 63--70.

\bibitem{jiang:noteonsnetcalc}
Y.~Jiang, ``A note on applying stochastic network calculus,'' Tech. Rep., 2010.

\bibitem{poloczek:schedulingmartingales}
F.~Poloczek and F.~Ciucu, ``Scheduling analysis with martingales,'' in
  \emph{Proc. of {IFIP Performance}}, Sep. 2014, pp. 56--72.

\bibitem{doob:stochasticprocesses}
J.~L. Doob, \emph{Stochastic Processes}.\hskip 1em plus 0.5em minus 0.4em\relax
  Wiley, 1953.

\bibitem{ross:probability}
S.~Ross, \emph{A First Course in Probability}.\hskip 1em plus 0.5em minus
  0.4em\relax Prentice Hall, 2002.

\bibitem{rizk:fbm}
A.~Rizk and M.~Fidler, ``Non-asymptotic end-to-end performance bounds for
  networks with long range dependent {FBM} cross traffic.'' \emph{Computer
  Networks}, vol.~56, no.~1, pp. 127--141, Jan. 2012.

\bibitem{liebeherr:heavytailed}
J.~Liebeherr, A.~Burchard, and F.~Ciucu, ``Delay bounds in communication
  networks with heavy-tailed and self-similar traffic,'' \emph{{IEEE} Trans.
  Inf. Theory}, vol.~58, no.~2, pp. 1010--1024, Feb. 2012.

\bibitem{arita:discretemm1}
C.~Arita and D.~Yanagisawa, ``Exclusive queueing process with discrete time,''
  Tech. Rep. arXiv:1008.4651v2, Nov. 2010.

\bibitem{cetinkaya:egressadmissioncontrol2}
C.~Cetinkaya, V.~Kanodia, and E.~W. Knightly, ``Scalable services via egress
  admission control,'' \emph{{IEEE} Trans. Multimedia}, vol.~3, no.~1, pp.
  69--81, Mar. 2001.

\bibitem{valaee:adhocadmissioncontrol}
S.~Valaee and B.~Li, ``Distributed call admission control for ad hoc
  networks,'' in \emph{Proc. of {IEEE} {VTC}}, Sep. 2002, pp. 1244--1248.

\bibitem{alcuri:servicecurvemeasurement}
L.~Alcuri, G.~Barbera, and G.~{D'Acquisto}, ``Service curve estimation by
  measurement: An input output analysis of a softswitch model,'' in \emph{Proc.
  of {QoS-IP}}, Feb. 2005, pp. 49--60.

\bibitem{bredel:netcalcroutermeasurements}
M.~Bredel, Z.~Bozakov, and Y.~Jiang, ``Analyzing router performance using
  network calculus with external measurements,'' in \emph{Proc. of {IEEE}
  {IWQoS}}, Jun. 2010.

\bibitem{hisakado:legendre}
T.~Hisakado, K.~Okumura, V.~Vukadinovic, and L.~Trajkovic, ``Characterization
  of a simple communication network using {L}egendre transform,'' in
  \emph{Proc. of {ISCAS}}, May 2003, pp. 738--741.

\bibitem{agharebparast:slopedomain}
F.~Agharebparast and V.~C.~M. Leung, ``Slope domain modeling and analysis of
  data communication networks: A network calculus complement,'' in \emph{Proc.
  of {IEEE} {ICC}}, Jun. 2006, pp. 591--596.

\bibitem{liebeherr:availbw}
J.~Liebeherr, M.~Fidler, and S.~Valaee, ``A system theoretic approach to
  bandwidth estimation,'' \emph{{IEEE/ACM} Trans. Netw.}, vol.~18, no.~4, pp.
  1040--1053, Aug. 2010.

\bibitem{luebben:availbw2}
R.~L\"{u}bben, M.~Fidler, and J.~Liebeherr, ``Stochastic bandwidth estimation
  in networks with random service,'' \emph{{IEEE/ACM} Trans. Netw.}, vol.~22,
  no.~2, pp. 484--497, Apr. 2014.

\bibitem{luebben:availbwdiss}
R.~L{\"u}bben, ``System identification of computer networks with random
  service,'' Ph.D. dissertation, Leibniz Universit{\"a}t Hannover, 2013.

\bibitem{melander:topp}
B.~Melander, M.~Bj{\"o}rkman, and P.~Gunningberg, ``A new end-to-end probing
  and analysis method for estimating bandwidth bottlenecks,'' in \emph{Proc. of
  {IEEE} {Globecom}}, Nov. 2000, pp. 415--420.

\bibitem{strauss:spruce}
J.~Strauss, D.~Katabi, and F.~Kaashoek, ``A measurement study of available
  bandwidth estimation tools,'' in \emph{Proc. of {ACM} {IMC}}, 2003, pp.
  39--44.

\bibitem{jain:slops}
M.~Jain and C.~Dovrolis, ``End-to-end available bandwidth: Measurement
  methodology, dynamics, and relation with {TCP} throughput,'' in \emph{Proc.
  of {ACM} {SIGCOMM}}, Oct. 2002, pp. 295--308.

\bibitem{ribeiro:pathchirp}
V.~Ribeiro, R.~Riedi, R.~Baraniuk, J.~Navratil, and L.~Cottrell, ``{PathChirp}:
  Efficient available bandwidth estimation for network paths,'' in \emph{Proc.
  of {PAM}}, Apr. 2003.

\bibitem{nam:minimalbackloggingbwest}
S.~Y. Nam, S.~Kim, and W.~Park, ``Analysis of minimal backlogging-based
  available bandwidth estimation mechanism,'' \emph{Computer Communications},
  vol.~35, no.~4, pp. 431--443, Feb. 2002.

\bibitem{fidler:legendre}
M.~Fidler and S.~Recker, ``Conjugate network calculus: A dual approach applying
  the {L}egendre transform,'' \emph{Computer Networks}, vol.~50, no.~8, pp.
  1026--1039, Jun. 2006.

\bibitem{rockafellar:convexanalysis}
R.~T. Rockafellar, \emph{Convex Analysis}.\hskip 1em plus 0.5em minus
  0.4em\relax Princeton University Press, 1972.

\bibitem{tabassum:basestationsleeping}
H.~Tabassum, U.~Siddique, E.~Hossain, and M.~J. Hossain, ``Downlink performance
  of cellular systems with base station sleeping, user association, and
  scheduling,'' \emph{{IEEE} Trans. Wireless Commun.}, vol.~13, no.~10, pp.
  5752--5767, Oct. 2014.

\bibitem{Fidler:2014:CDE}
M.~Fidler, R.~L{\"u}bben, and N.~Becker, ``{Capacity-Delay-Error-Boundaries: A
  Composable Model of Sources and Systems},'' \emph{{IEEE} Trans. Wireless
  Commun.}, vol.~14, no.~3, pp. 1280--1294, Mar. 2015.

\end{thebibliography}
\end{document}